\documentclass{amsart}
\usepackage[margin=1in]{geometry}
\usepackage{fancyhdr}
\usepackage{amsfonts,amsmath,amssymb}
\usepackage{xspace}
\usepackage{hyperref}

\newtheorem{theorem}{Theorem}[section]
\newtheorem{proposition}[theorem]{Proposition}
\newtheorem{observation}[theorem]{Observation}
\newtheorem{lemma}[theorem]{Lemma}
\newtheorem{corollary}[theorem]{Corollary}
\newtheorem{definition}[theorem]{Definition}

\newcommand{\bE}{\ensuremath{\mathbf{E}}\xspace}
\newcommand{\bP}{\ensuremath{\mathbf{P}}\xspace}

\begin{document}
\title{Improved bounds and algorithms for graph cuts and network reliability$^{1}$}
\author{
{\sc David G.~Harris}$^{2}$
\and
{\sc Aravind Srinivasan}$^{3}$
}

\setcounter{footnote}{0}

\addtocounter{footnote}{1}
\footnotetext{
{A preliminary version of this work appeared in the \emph{Proc.\ ACM-SIAM Symposium on Discrete Algorithms} (SODA), 2014}.}
\addtocounter{footnote}{1}
\footnotetext{Department of Computer Science, University of Maryland, 
College Park, MD 20742. 
Research supported in part by NSF Awards CNS-1010789 and CCF-1422569.
Email: \texttt{davidgharris29@gmail.com}.}
\addtocounter{footnote}{1}
\footnotetext{Department of Computer Science and
Institute for Advanced Computer Studies, University of Maryland, 
College Park, MD 20742. 
Research supported in part by NSF Awards CNS-1010789 and CCF-1422569, and a research award by Adobe, Inc.
Email: \texttt{srin@cs.umd.edu}.}

\date{}
\maketitle

\sloppy
\pagestyle{plain}

\begin{abstract}
Karger (\emph{SIAM Journal on Computing}, 1999) developed the first fully-polynomial approximation scheme to estimate the probability that a graph $G$ becomes disconnected, given that its edges are removed independently with probability $p$. This algorithm runs in $n^{5+o(1)} \epsilon^{-3}$ time to obtain an estimate within relative error $\epsilon$.

We improve this run-time through algorithmic and graph-theoretic advances. First, there is a certain key sub-problem encountered by Karger, for which a generic estimation procedure is employed; we show that this has a special structure for which a much more efficient algorithm can be used. Second, we show better bounds on the number of edge cuts which are likely to fail. Here, Karger's analysis uses a variety of bounds for various graph parameters; we show that these bounds cannot be simultaneously tight. We describe a new graph parameter, which simultaneously influences all the bounds used by Karger, and obtain much tighter estimates of the cut structure of $G$. These techniques allow us to improve the runtime to $n^{3+o(1)} \epsilon^{-2}$; our results also rigorously prove certain experimental observations of Karger \& Tai (\emph{Proc.\ ACM-SIAM Symposium on Discrete Algorithms}, 1997). Our rigorous proofs are motivated by certain non-rigorous differential-equation approximations which, however, provably track the worst-case trajectories of the relevant parameters. 

A key driver of Karger's approach (and other cut-related results) is a bound on the number of small cuts: we improve these estimates when the min-cut size is ``small'' and odd, augmenting, in part, a result of Bixby (\emph{Bulletin of the AMS}, 1974). 


\end{abstract}

\section{Introduction}
Let $G$ be a connected undirected multi-graph with vertex set $V$; as usual, we let $|V| = n$. Unless stated
otherwise, the graphs we deal with will be \emph{multi-graphs with no self-loops}, presented in \emph{adjacency-matrix format}. We define $R(p)$, the reliability polynomial of $G$, to be the probability that the graph remains connected when edges are removed independently with probability $p$. One can verify that this is a polynomial in $p$. This polynomial has various physical applications, for example determining the reliability of a computer network or power grid.
Although there is no currently known algorithm for estimating $R(p)$, the complementary probability $U(p) = 1 - R(p)$, which we call the \emph{unreliability} of the graph, can be estimated in polynomial time.  In a breakthrough paper (\cite{karger}, see also \cite{karger2}), Karger developed the first fully-polynomial randomized approximation scheme (FPRAS) to estimate $U(p)$ up to relative error $\epsilon$ in time polynomial in $n$ and $1/\epsilon$; even a constant-factor approximation was not known prior to his work.  

We let $\exp(x)$ denote $e^x$; all logarithms are to \textbf{base} $\mathbf{e}$ unless indicated otherwise.
With an appropriate choice of parameters, Karger's algorithm can run in time $\exp \bigl( 3 \log n + \sqrt{\log n (4 \log n + \sqrt{2} \log (1/\epsilon))} + o(\log n) \bigr)$. In particular, if $\epsilon = \Theta(1)$, this is $n^{5 + o(1)}$. (We assume that $\epsilon \leq O(1)$, say $\epsilon \leq 1/2$, throughout.)
Karger's algorithm is based on an algorithm for finding graphs cuts which have close to minimal weight. This algorithm is the \emph{Contraction Algorithm}, first introduced by \cite{kar93}. This algorithm has emerged as an important building block for many graph algorithms. Just as importantly, it can be viewed as a stochastic process which provides structural results about the graph cuts, for example, a bound on the number of small cuts.

In this paper, we provide a more detailed analysis of the Contraction Algorithm and its consequences. This enables us to show a variety of improved bounds and algorithms for graph problems. The focus of this paper is on improving the algorithm for estimating $U(p)$.

The following definition will be useful: 
\begin{definition}
The minimum cut-size in $G$, also known as the edge-connectivity of $G$, will be denoted by $c$. Given $\alpha \geq 1$, an ``$\alpha$-cut'' in $G$ is a cut with at most $\alpha c$ edges. 
\end{definition}

\smallskip \noindent \textbf{Our results.}

\smallskip
The main focus of our analysis is a much faster algorithm for estimate $U(p)$:
\begin{theorem}
There is an algorithm to estimate $U(p)$ in time $n^{3+o(1)} \epsilon^{-2}$.
\end{theorem}

Our analysis of the Contraction Algorithm will also allow purely structural bounds on the cut structure of certain graphs:
\begin{theorem}
Suppose $c$ is odd and $\alpha \geq 1$. Then there are at most $O(n^{2 \alpha \frac{c}{c+1}})$ $\alpha$-cuts. This exponent is optimal, in the sense that there are counter-examples to bounds of the form $O(n^{\alpha x})$ for any $x < \frac{2 c}{c+1}$.
\end{theorem}

In \cite{karger-stein}, an algorithm called the \emph{Recursive Contraction Algorithm} was developed for finding all $\alpha$-cuts.  We obtain a small improvement to this algorithm for the general problem of finding $\alpha$-cuts. Unlike the previous improvements, this only reduces the logarithmic terms:
\begin{theorem}
There is an algorithm that accepts as input a graph $G$ and a real number $\alpha \geq 1$. This algorithm succinctly enumerates, with high probability, all $\alpha$-cuts of $G$ in time $O(n^{2 \alpha} \log  n)$.
\end{theorem}
By contrast,  the algorithm of \cite{karger-stein} requires time $O(n^{2 \alpha} \log ^2 n)$.

We note that our algorithm focuses on estimating the parameter $U(p)$, but there are other connectivity measures for which a strategy of enumerating ``representative'' cuts may be used. \cite{karger} gives a similar type of algorithm for estimating multi-way connectivity and \cite{weinshall} discusses an algorithm for estimating parameters related to clustering. Obtaining an FPRAS -- or even a constant-factor approximation -- for $R(p)$ remains a very intriguing open problem. 

\subsection{Overview of Karger's algorithm, and our improvements}
Karger's algorithm for estimating $U(p)$ essentially consists of two separate algorithms, depending on the size of $U(p)$. When $U(p) > n^{-2-\delta}$, where $\delta$ is some constant, then Monte Carlo sampling provides an accurate estimate. As the samples are unbiased with relative variance $1/U(p)$, then after $n^{2+\delta} \epsilon^{-2}$ samples we estimate $U(p)$ accurately. Naively, it might appear to require $O(m)$ time per sample (one must process each edge). A clever sparsification strategy described in Karger \cite{karger} reduces this to $n^{1+o(1)}$ time per sample. We will not modify Karger's algorithm for Monte Carlo sampling.

When $U(p)$ is small, Monte Carlo sampling no longer can produce an accurate estimate in reasonable time. In the regime $U(p) < n^{-2-\delta}$, Karger develops an alternative algorithm. In this case, Karger shows that the event of graph failure is dominated by the event that  a \emph{small} cut of $G$ has failed. This is done by analyzing the number of cuts of small weight. In particular, there is some $\alpha^*$ such that $U(p)$ is closely approximated by the probability that an $\alpha^*$-cut has failed. Karger provides an upper bound on the critical value $\alpha^*$; for example, when $\epsilon = \Theta(1)$, we have $\alpha^* \leq 1 + 2/\delta = O(1)$. 

This significantly simplifies the problem, because instead of having to consider the exponentially large collection of all cuts of $G$, we can analyze only the polynomial-size collection of $\alpha^*$-cuts. Using an algorithm developed in \cite{karger-stein}, the \emph{Recursive Contraction Algorithm}, Karger's algorithm can catalogue all such $\alpha^*$-cuts in time $n^{2 \alpha^* + o(1)}$. We refer to this as the \emph{cut-enumeration} phase of Karger's algorithm. Note that there is a tradeoff between the Monte Carlo phase and cut-enumeration phase, depending on the size of $\delta$. 

Having catalogued all such small cuts, we need to piece them together to estimate the probability that one of them has failed. For this, Karger uses as a subroutine a statistical procedure developed by Karp, Luby, Madras \cite{lubydnf}. This procedure can examine any collection of partially overlapping clauses, to provide an unbiased estimate with relative error $O(1)$ that one such clause is satisfied. By running $\epsilon^{-2}$ iterations of this, we achieve the desired error bounds for $U(p)$. The running time is linear in the size of the collection of clauses. We refer to this as \emph{statistical sampling} phase of Karger's algorithm.

We will improve both the cut-enumeration phase and statistical sampling phase of Karger's algorithm. To improve the cut-enumeration phase, we will show a tighter bound on $\alpha^*$. In \cite{implementing}, a version of Karger's algorithm was programmed and tested on real graphs. This found experimentally that the average number of cuts needed to accurately estimate the graph failure probability was about $n^3$, not the worst case $n^{5}$ as predicted by \cite{karger}. In this paper, we show this fact rigorously. We will show that the number of small graph cuts in reliable graphs is far smaller than in general graphs.

The basic idea is to analyze the dynamics of the Contraction Algorithm, and to show that any small cut $C$ has a large probability of being selected by this algorithm. This analysis is quite difficult technically, because it requires analyzing the three-way interconnections between the number of $\alpha$-cuts of $G$, the effectiveness of the Contraction Algorithm, and the value of $U(p)$. The critical parameter that ties these three together is the \emph{expected number of failed graph cuts}. 

Our improvement to the statistical sampling phase is more straightforward. The algorithm of \cite{lubydnf} is able to handle very general collections of clauses, which can overlap in complicated ways. However, this is unnecessary for our purposes. We will show that the collection of cuts produced by the cut-enumeration phase has a simple distribution, both statistically and computationally. This allows us to use a faster and simpler statistical procedure to piece together the $\alpha^*$-cuts. As a result, the running time of each sample reduces to about $n^2$; in particular, we do not need to read the entire collection of $\alpha^*$-cuts for each sample. 
 
In total, we reduce the running time of Karger's algorithm to about $n^{3+o(1)} \epsilon^{-2}$. Obtaining faster run-times with Karger's recipe appears to be quite difficult; for if $c$ is the size of the min-cut, then in the regime where $p^c \geq n^{-2 + \Omega(1)}$, it is not clear how to stop only at ``small'' cuts, and one appears to need Monte-Carlo sampling -- which requires a runtime that matches ours. 

We note that another approach to estimating $U(p)$ has been discussed in \cite{rta}. This algorithm finds a pair of minimal cuts which are mostly disjoint, and hence there is a negligible probability that both cuts fail simultaneously. By continuing this process, it is found in \cite{rta} experimentally that a branching process can efficiently enumerate the probability that some small cut fails. While this algorithm is promising, the analysis of \cite{rta} is  fully experimental, and it appears that in the worst case this approach might require super-polynomial time.

Around the time of our finalizing the published version of this work, a newer algorithm by Karger was developed for estimating $U(p)$, which is not based on cut-enumeration \cite{karger3}. This algorithm has a run-time of $O(n^3 \text{polylog}(n)/\epsilon^2)$, which is strictly better than ours.

\subsection{Analyzing the trajectory of stochastic processes}
\label{stoch-proc-sec}
The Contraction Algorithm is at the heart of our algorithm, and our main contribution is giving an improved and tighter analysis of it. This is a complex stochastic process, which applies a sequence of random transformations to the graph $G$. A single step can radically alter the trajectory of this process. In the original paper \cite{karger-stein} which introduced this algorithm, very simple and elementary methods were used to bound its behavior. These bounds can be tight for general graphs, but for other types of graphs these bounds are not accurate. However the analysis has not been improved since \cite{karger-stein}.

To analyze the Contraction Algorithm, we introduce a new methodology for approximating stochastic systems by continuous, deterministic dynamical systems. We believe this method may have applications to other types of systems, and we will describe it here at a high level of generality.

Suppose we want to upper-bound $\bE[f(G)]$, where $f$ is some function of the stochastic process we are interested in (for us, essentially $f$ counts the number of edges found in the graphs produced at intermediate stages of the Contraction Algorithm). We begin by making any number of heuristic guesses as to the behavior of the process. These guesses can be totally unjustified and are often nothing more than wishful thinking. For example, we may assume that any or all random variables are deterministic; we may assume that certain parameters are as large as possible; we replace the discrete time intervals with continuous time. In this way, we develop a dynamical system which we believe should approximate the behavior of $\bE[f(G)]$. This dynamical system can be solved to produce some function $\tilde f(G)$; we guess (hope) that $\bE[f(G)] \leq \tilde f(G)$ for all graphs $G$.

So far, our approach mirrors that of \cite{wormald}, which gave general strategies for analyzing stochastic systems. According to \cite{wormald}, the next stage of the solution would be to prove that our assumptions were approximately correct. For example, instead of treating a random variable as deterministic, we may use concentration inequalities to prove that it is centered around its mean (and thus is nearly deterministic). In the case of the Contraction Algorithm, this approach does not seem to work; we are not able to show concentration inequalities for the relevant random variables. The reason is that a single step of the Contraction Algorithm, namely contracting an edge, can change the graph in a far-reaching way; this makes it hard to show the type of ``local effect'' required for a concentration inequality. 

By contrast, we will prove that any violation of these heuristic assumptions will only help us; we do not necessarily show that the heuristic assumptions are true or approximately true. We will use induction on $G$ to prove that for any graph $G$ we have $\bE[f(G)] \leq \tilde f(G)$. 

To explain how this works, for expository purposes suppose that the function $\tilde f$ only depends on the final stage of the stochastic process and that $\tilde f$ is a function not of the full graph $G$ but only a single, simple-to-compute parameter $\phi(G)$, so that $\tilde f(G) = g( \phi(G) )$ (for example, $\phi(G)$ is the number min-cuts of $G$). Now, after the initial time-step of the stochastic process, the graph $G$ transforms into the graph $H$ (where now $H$ is a random variable). Thus, we now have $\bE[f(G)] =\bE_H [ \bE[f(H)] ]$; here we are splitting the expectation into two pieces --- the random choice for the first time-step, and the random choices for the remaining time-steps. 

We will now use our induction hypothesis applied to the subgraph $H$; this gives us that $\bE[f(H)] \leq \tilde f(H)$. Thus, so far we have shown that $\bE[f(G)] \leq \bE_H [ \tilde f (\phi(H))]$.

Now, suppose that the function $g(x)$ happens to be an increasing, concave-down function of $x$. In this case, we can apply Jensen's inequality to this expression, deriving
$\bE_H [ \tilde f (H)] \leq g( \bE[\phi(H)] )$, and hence that $\bE[f(G)] \leq g( \bE[\phi(H)] )$.

Thus, in order to continue the induction proof, we merely need to show that
$$
g( \bE[ \phi (H) ] ) \leq g( \phi(G) )
$$
 
We have thus reduced our trajectory analysis to computing $\bE[ \phi(H) ]$ (which involves only a single time-step of the stochastic process and the relatively simple function $\phi$) and a simple inequality for a concrete function $g$.

Overall, this process depends very heavily on the concavity and monotonicity properties of the function $\tilde f$. This is very similar to how one can apply Jensen's inequality to interchange the expectation of a function with the function of its expectation. So, in a sense we can view our method as a ``generalized Jensen's inequality'' for stochastic systems.

\subsection{Outline}
In Sections~\ref{Aprelimsec} and ~\ref{Afailedcuts}, we begin by reviewing some key results and definitions concerning graph cuts. Most of these results are recapitulations of Karger \cite{karger}.

In Section~\ref{Asec:statistical}, we describe a new statistical sampling algorithm for the following problem: given a fixed collection $\mathcal A$ of cuts, what is the probability that at least one cut from $\mathcal A$ fails?  We show that this algorithm has a faster running time, and is just as accurate, compared to the algorithm of \cite{lubydnf} used in \cite{karger,karger2}.  Although this sampling algorithm is the final part of our overall algorithm to estimate $U(p)$, we will discuss it first because it is relatively self-contained and because it is less technically challenging than the rest of the paper.

In Section~\ref{Asec:contraction}, we introduce the Contraction Algorithm, an algorithm described by \cite{karger-stein}. This algorithm provides an efficient randomized procedure for finding small-weight cuts in $G$. It also can be viewed as a probabilistic process which can provide purely structural bounds on $G$.  We will be interested in the following question, which is more general than that considered by \cite{karger-stein}: suppose we are given a fixed target cut $C$ of $G$. Under what circumstances does the Contraction Algorithm select $C$? What can we say about the dynamics of the Contraction Algorithm \emph{in those cases in which $C$ will ultimately be selected}? We will prove a series of technical Lemmata which describe these dynamics. Roughly speaking, we ``factor out'' $C$ to show that the Contraction Algorithm has a \emph{uniform behavior} regardless of which target cut $C$ (if any) we are interested in.

In Section~\ref{Asec:oddc} we apply this machinery to analyze the Contraction Algorithm in the case in which the (unweighted) graph $G$ has $c$ an odd number. These results are not necessary for our analysis of Karger's algorithm. We include them here for two reasons. First, the number of small cuts is basically known in the case of even $c$ --- the worst case behavior comes from a cycle graph, in which there are $n$ vertices in a ring with $c/2$ edges between successive vertices. The case of odd $c$ has been mostly overlooked in the literature.  Second, the case of small odd $c$ provides an easier warm-up exercise for the analysis in Section~\ref{Azsec}. In Section~\ref{Asec:oddc}, we keep track of the number of min-cuts remaining in $G$ during the evolution of the Contraction Algorithm; in Section~\ref{Azsec} we must keep track of the number of cuts of all sizes. Section~\ref{Azsec} analyzes how the dynamics of the Contraction Algorithm are affected by the magnitude of $U(p)$. The critical parameter is $\bar Z$, the \emph{expected number of failed cuts} in $G$. We show how $\bar Z$ affects the number of edges which are available in any round of the Contraction Algorithm, and we show how $\bar Z$ itself changes during the execution.

The proof method used in Sections~\ref{Asec:oddc} and \ref{Azsec} follows the outline of Section~\ref{stoch-proc-sec}. We give a function $\tilde f$, and prove by induction that the probability of a cut surviving the Contraction Algorithm is at least $\tilde f$. This function $\tilde f$ is very complicated and could not have been guessed from first principles.  Instead, we first simplify our stochastic process using unwarranted independence and monotonicity assumptions (for example, we assume that the subgraphs produced during the Contraction Algorithm have as few edges as possible subject to the minimum-cut condition), and then translate the stochastic time-steps into a differential equation which can be solved in closed form. This heuristic analysis gives us the mysterious function $\tilde f$, which we then rigorously prove correct. In particular, these ``unwarranted independence and monotonicity assumptions'' actually correspond to the worst-case behavior of our stochastic process! 

In Section~\ref{Asec:smallcutunreliability}, we apply the analysis of Section~\ref{Azsec} to show that $\alpha^*$, which is the bound of the size of the cuts necessary to approximate $U(p)$, is smaller than the bound given by Karger. We note that while Karger gave a simple bound which could be computed explicitly, our bound depends on the graph parameter $\bar Z$ which we cannot determine easily. However, using this formula, one can determine an upper bound on the total number of iterations of the Contraction Algorithm that are needed to find the $\alpha^*$-cuts; this bound is irrespective of $\bar Z$.

In Section~\ref{Asec:rca}, we analyze the Recursive Contraction Algorithm. This is an algorithm which allows us to run the Contraction Algorithm ``in bulk". If we are interested in finding many graph cuts, we must run the Contraction Algorithms multiple times. However, most of the work for processing the graph can be amortized across these multiple runs. In particular, we can essentially reduce the running time for the Contraction Algorithm from $n^2$ (dominated by reading the original graph) to $n^{a}$ where $a$ is a small constant.

In Section~\ref{put-it-sec}, we describe how to combine all the pieces and obtain a full algorithm. This includes deciding when to use Monte Carlo sampling and when to enumerate the small cuts, a detail omitted in Karger \cite{karger}. Finally, Section~\ref{sec:concl} concludes. 

\section{Preliminaries}
\label{Aprelimsec}
For a multi-graph $G$, we define a \emph{cut} of $G$ to be a partition of the vertices into two classes $V = A \sqcup A'$, with $A, A' \neq \emptyset$. The sets $A, A'$ are the \emph{shores} of the cut. We distinguish this from an \emph{edge-cut}, which is a subset $E'$ of the edges of $G$ such that removal of $E'$ disconnects $G$. Every cut of $G$ induces an edge-cut of $G$, namely the edges crossing from $A$ to $A'$. The \emph{weight} or \emph{size} of an edge-cut is the number of edges it contains. We say an edge-cut is \emph{minimum} if it has the smallest size of all edge-cuts, and we denote by $c$ the size of the smallest edge-cut. We say an edge-cut is \emph{minimal} if no proper subset is an edge-cut.  

We will often abuse notation so that a cut $C$ may refer to either the vertex-partition \emph{or} the edge-cut it induces. Unless stated otherwise, whenever we refer to a cut $C$, we view $C$ as a set of \emph{edges}. 

As mentioned before, for any $\alpha \geq 1$ we define an $\alpha$-cut to be a cut (not necessarily minimal) whose corresponding edge-cut has at most $\alpha c$ edges, and we define a \emph{min-cut} to be a $1$-cut, i.e. a cut of weight $c$ exactly.

In this paper, we seek to estimate the probability that the graph becomes disconnected when edges are removed independently with probability $p$. When edges are removed in this way, we say that an edge-cut \emph{fails} if and only if all the corresponding edges are removed. In this case, we may concern ourselves solely with cuts. The reason for this is that graph $G$ becomes disconnected if and only if some cut of $G$ fails.

This method can be generalized to allow each edge to have its own independent failure probability $p_e$. As described in \cite{karger}, given a graph $G$ with non-uniform edge failure probabilities, one can transform this to a multi-graph $G'$ with uniform edge failure probabilities by replacing each edge of $G$ with a bundle of edges, and setting $p$ appropriately. 

We will assume that the graph $G$ is presented as an adjacency matrix with $n^2$ words. Each word records the number of edges that are present between the indicated vertices.  As described in \cite{kar93}, it is possible to transform an arbitrary graph, in which the number of edges may be unbounded, into one with cells of size $O(\log (1/\epsilon) + \log n)$ with similar reliability; hence the arithmetic operations will take polylogarithmic time in any computational model. This transformation may have a running time which is super-polynomial in $n$, although it is close to linear time as a function of the input data size. We will ignore these issues, and simply assume that we can perform arithmetic operations to precision $\epsilon$ and random number generation of an entire word in a single time step.

As our algorithm closely parallels Karger's, we use the notation of \cite{karger} wherever possible. Recall that $c$ is the size of the smallest edge-cut.
As in \cite{karger}, we let $p^c = n^{-2-\delta}$, with $\delta > 0$ being the case of primary interest (the complementary case can be handled by simple Monte-Carlo sampling). We note that $\delta$ can be determined in time $n^{2+o(1)}$, simply by finding the size of the minimum cut. Hence we can assume $\delta$ is known.  Often we will derive estimates assuming $\delta \geq \delta_0$; here $\delta_0$ can be an arbitrarily small positive constant. This assumption allows us to simplify many of the asymptotic notations, whose constant terms may depend on $\delta_0$. In fact, our algorithm is best balanced when $\delta \rightarrow 0$. It is not hard to see that all of our asymptotic bounds when $\delta \geq \delta_0$ can be relaxed to allow $\delta \rightarrow 0$ sufficiently slowly, for example as $\delta = 1/\log \log \log \log n$. 

When we analyze $\alpha$-cuts, we are primarily interested in the case when $\alpha$ is a slow-growing function. (In fact, the case where $\alpha = O(1)$ would be basically sufficient to analyze all of our algorithmic improvements.) We will show how to bound the number of such cuts.

The following is a well-known theorem of \cite{karger:thesis}, and indeed, the rest of this section is basically a recapitulation of Karger's work from \cite{karger:thesis}. 

\begin{theorem}
\label{Acutsetthm1} 
(\textbf{\cite{karger:thesis}})
The number of $\alpha$-cuts is at most $n^{2 \alpha}$.
\end{theorem}

The following combinatorial principle, which is basically a form of integration by parts, will be used in a variety of contexts.
\begin{proposition}
\label{Aexpsumlemma}
Let $F: \mathbf R \rightarrow [0, \infty)$ be any increasing function with distributional derivative $dF(x)$ and with the property that, for any $\alpha \geq 1$, the number of $\alpha$-cuts is at most $F(\alpha)$. Let $g: \mathbf R \rightarrow \mathbf [0, \infty)$ be any continuous decreasing function. Then
$$
\sum_{\substack{\text{Cuts $C$ of weight} \\ |C| \geq \alpha c}} \hspace{-0.3in} g(|C|/c) \leq F(\alpha) g(\alpha) + \int_{x = \alpha}^{\infty} g(x) dF(x)
$$
\end{proposition}
\begin{proof}
Enumerate the weights of all cuts whose weight is at least $\alpha c$ in sorted order as $\alpha c \leq r_1 \leq r_2 \leq \dots \leq r_k$. Observe that for any $i = 1, \dots, k$ the definition of $F$ implies that $i \leq F(r_i/c)$. 

Then 
{\allowdisplaybreaks
\begin{align*}
\sum_{\substack{\text{Cuts $C$ of weight} \\ |C| \geq \alpha c}} \hspace{-0.3in} g(|C|/c) &= \sum_{i=1}^{k} g(r_i/c) = k g(r_k/c) + \sum_{i=1}^{k-1} i \Bigl( g(r_i/c) - g( r_{i+1}/c ) \Bigr) \\
&\leq F(r_k/c) g(r_k/c) + \sum_{i=1}^{k-1} F(r_i/c) \Bigl( g (r_i/c) - g( r_{i+1}/c ) \Bigr) \\
& \qquad \qquad \text{as $i \leq F(r_i/c)$ and $g (r_i/c) - g(r_{i+1}/c ) \geq 0$} \\
&=g(r_1/c) F(r_1/c) + \sum_{i=2}^k g(r_i/c) \Bigl( F(r_i/c) - F(r_{i-1}/c) \Bigr) \\
&= g(r_1/c) \Bigl( F(\alpha) + \int_{x=\alpha}^{r_1/c} dF(x) \Bigr) + \sum_{i=2}^{k} g(r_i/c) \int_{x=r_{i-1}/c}^{r_i/c} dF(x) \\
&= F(\alpha) g(r_1/c) + \int_{x=\alpha}^{r_1/c} g(r_1/c) dF(x) + \sum_{i=2}^k \int_{x=r_{i-1}/c}^{r_i/c} g(r_i/c) dF(x) \\
&\leq F(\alpha) g(\alpha) + \int_{x=\alpha}^{r_1/c} g(x) dF(x) + \sum_{i=2}^k \int_{x=r_{i-1}/c}^{r_i/c} g(x) dF(x)  \qquad \text{as $g$ is decreasing} \\
&= F(\alpha) g(\alpha) + \int_{x=\alpha}^{r_k/c} g(x) dF(x) \leq F(\alpha) g(\alpha) + \int_{x=\alpha}^{\infty} g(x) dF(x)
\end{align*}
}

\end{proof}
Combining Proposition~\ref{Aexpsumlemma} and Theorem~\ref{Acutsetthm1}, gives a simple proof of the following result from \cite{karger}:
\begin{corollary}
\label{Aexpsumcorr1}
Let $0 < q < n^{-2 - \Omega(1)}$, $\alpha \geq 1$. Then
$$
\sum_{\substack{\text{Cuts $C$ of weight} \\ |C| \geq \alpha c}} \hspace{-0.3in} q^{|C|/c} = O( (n^2 q)^{\alpha})
$$
\end{corollary}
\begin{proof}
Apply Proposition~\ref{Aexpsumlemma}, using $F(x) = n^{2 x}$ and $g(x) = q^{x}$. We have
\begin{align*}
\sum_{\substack{\text{Cuts $C$ of weight} \\ |C| \geq \alpha c}} \hspace{-0.3in} q^{|C|/c} &\leq q^{\alpha} n^{2 \alpha} + \int_{x = \alpha}^{\infty} q^{x} \times 2 n^{2x} \log n \ dx \\
&= (n^2 q)^{\alpha} - \frac{2 (n^2 q)^{\alpha} \log n}{\log (n^2 q)} = (n^2 q)^{\alpha} \Bigl(1 - \frac{\log n^2}{\log (n^2 q)} \Bigr) \\
&\leq (n^2 q)^{\alpha} \Bigl(1 - \frac{\log n^2}{\log (n^2 n^{-2 - \Omega(1)})} \Bigr) \\
&= O( (n^2 q)^{\alpha} )
\end{align*}
\end{proof}

The following result is as example of this principle; it will later be strengthened in Proposition~\ref{Afaillemma}.
\begin{corollary}
\label{Acutfailcorr1}
(\textbf{\cite{karger}}) 
Suppose $\delta \geq \delta_0$ for some constant $\delta_0 > 0$. Then the probability that a cut of weight $\geq \alpha c$ fails is at most $O(n^{-\alpha \delta})$.

\end{corollary}
\begin{proof}
By the union-bound, 
$$
\bP(\text{Cut of weight $\geq \alpha c$ fails}) \leq \hspace{-0.2in} {\displaystyle \sum_{\substack{\text{Cuts $C$ of weight} \\ |C| \geq \alpha c}}} \hspace{-0.2in}  (p^c)^{|C|/c}.
$$
Now apply Corollary~\ref{Aexpsumcorr1}.
\end{proof}

This leads to one of the key theorems of Karger's original work:
\begin{theorem}
\label{Akargerthm1}
(\textbf{\cite{karger}}) 
Suppose $\delta \geq \delta_0$ for some constant $\delta_0 > 0$. Then $U(p)$ can be approximated, up to relative error $O(\epsilon)$, by the probability that a cut of weight $\leq \alpha^* c$ fails, where
$$
\alpha^* = 1 + 2/\delta - \frac{\log \epsilon}{\delta \log n}
$$
\end{theorem}
\begin{proof}
The \emph{absolute error} committed by ignoring cuts of weight $\geq \alpha^* c$ is at most the probability that such a cut fails; by Corollary~\ref{Acutfailcorr1} it is at most $O(n^{-\alpha^* \delta})$.

The minimum cut of $G$ fails with probability $p^c = n^{-2-\delta}$, so $U(p) \geq p^c = n^{-2-\delta}$.

Hence the \emph{relative error} committed by ignoring cuts of weight $\geq \alpha^* c$ is at most
$$
\text{Rel err} = O(\frac{ n^{-\alpha^* \delta} }{n^{-2-\delta} }) = O(\epsilon)
$$
\end{proof}

We do not want to get ahead of ourselves, but one of the main goals of this paper will be to improve on the estimate of Theorem~\ref{Akargerthm1}. In proving Theorem~\ref{Akargerthm1}, we used two quite different bounds. These bounds are tight separately, but for very different kinds of graphs. In estimating the absolute error, we bound the number of $\alpha$-cuts that may appear in $G$ by $O(n^{2 \alpha})$. This is tight for cycle graphs, as discussed in Lemma~\ref{Amultiple-cut-lemma1}, which have a very large number of cuts. On the other hand, when we estimate $U(p) \geq p^c$, we are assuming that the graph has just a single small cut. As we will show, no graph can have both bounds be simultaneously tight.

\section{The total number of failed cuts}
\label{Afailedcuts}
We define the random variable $Z$ to be the number of cuts of $G$ which fail when edges are removed independently with probability $p$. The expectation of $Z$, also known as the \emph{partition function} is $$
\bar Z = \sum_{\text{cuts $C$}} p^{|C|}.
$$
This is an overestimate of the graph failure probability which ignores the overlap between the cuts. The  number of failed cuts $Z$, and its expectation $\bar Z$, will play crucial roles in our estimates. We take here an approach based on \cite{bk15}, inspired by establishing a connection between the number of failed cuts and an appropriately scaled binomial random variable. 

\begin{proposition}
Suppose we remove a subset $L$ of edges from $G$, and the resulting graph $H$ has $R$ connected components. Then $L$ causes exactly $2^{R-1} - 1$ cuts of $G$ to fail.
\end{proposition}
\begin{proof} Any partition of the $R$ components corresponds to a failed cut.
\end{proof}

\begin{lemma}
\label{Amultiple-cut-lemma1}
Let $p^c = n^{-2-\delta}$ for $\delta \geq \delta_0 > 0$. Suppose we remove edges from $G$ with probability $p$. Let $H$ denote the resulting graph and let $R$ be its number of connected components. Then for $n$ sufficiently large and $r \geq 1$ we have
$$
\bP(R \geq r) \leq 1.01 n^{-r \delta/2} /r!
$$ 
\end{lemma}
\begin{proof}
As shown in \cite{lomonosov1}, \cite{lomonosov2}, \cite{karger}, the probability that $H$ has $\geq r$ connected components under edge-failure probability $p$, is at most the probability that the graph $C$ has $\geq r$ connected components under edge failure $q = p^{c/2}$, where $C$ is the cycle graph which contains $n$ vertices and a single edge between successive vertices in the cycle. If $r$ edges fail in $C$, then the resulting number of distinct components is $\max(1,r)$. Hence the probability that $C$ results in $r$ distinct components is the probability that a binomial random variable, with $n$ trials and success probability $q$, is at least $r$:
\begin{align*}
\bP(\text{$\geq r$ connected components}) &\leq \sum_{i=r}^n \binom{n}{i} q^i (1-q)^{n-i} \leq \sum_{i=r}^n n^i (n^{i (-2-\delta)/2}/r!)  \leq \frac{n^{ (\delta/2)(1 - r)}}{r! (n^{\delta/2} - 1) }
\end{align*}

As $\delta > \delta_0 > 0$, then $n^{\delta/2} \geq 1000$ for $n \geq \Omega(1)$. Thus $n^{\delta/2} - 1 \geq 0.999 n^{\delta/2}$ and we have
$$
\bP(\text{$\geq r$ connected components}) \leq \frac{n^{ (\delta/2)(1 - r)}}{0.999 r! n^{\delta/2}} \leq 1.01 n^{-\delta r/2}/r!
$$
\end{proof}

Thus $\bar Z$ is bounded within a relatively small range:
\begin{corollary}
\label{Azbarboundcorr}
For $\delta \geq \delta_0 > 0$ and for $n$ sufficiently large,
$$
n^{-2-\delta} \leq \bar Z \leq n^{-\delta}
$$
\end{corollary}
\begin{proof}
The lower bound $\bar Z\geq n^{-2-\delta}$ follows since the probability that any one min-cut fails is $n^{-2-\delta}$.

To show the upper bound, we have
\begin{align*}
\bar Z = \bE[Z] &= \bE[2^{R-1}-1] =\sum_{r \geq 2} \bP( R \geq r ) 2^{r-2} \\
&\leq 1.01 \sum_{r \geq 2} \frac{2^{r-2} n^{-r \delta/2}}{r!} \qquad \text{by Lemma~\ref{Amultiple-cut-lemma1}} \\
&\leq 0.26 \left(-2 n^{-\delta /2}+e^{2 n^{-\delta /2}}-1\right)
\end{align*}

For $n$ sufficiently large, the exponent $2 n^{-\delta/2}$ approaches $0$. As $e^x \leq 1 + x + 0.51 x^2$ for $x$ sufficiently close to $0$, it follows that
$$
\bar Z \leq 0.26 \bigl(-2 n^{-\delta /2}+(1 + 2 n^{-\delta /2} + 0.51 \times 4 n^{-\delta})-1 \bigr) = 0.5304 n^{-\delta}
$$
as desired.

\end{proof}

In light of Corollary~\ref{Azbarboundcorr}, it is often convenient to write 
$$
p^c = n^{-2-\delta} \qquad \bar Z = n^{-2-\delta + \beta}
$$
Here, $\delta, \beta$ should be thought of as ``parameters'' of the graph; $\delta$ measures the probability that the smallest cut of $G$ fails, while $\beta \in [0,2]$ counts the number of small cuts. The case $\beta \approx 0$ means there is only a single small cut and the other cuts are very large; the case $\beta \approx 2$ corresponds to the cycle graph. This parametrization will appear many times throughout this paper. We will always assume that $\delta > \delta_0 > 0$, where $\delta_0$ is a constant which may be made arbitrarily small. All of the hidden asymptotic terms may depend on $\delta_0$. We will not state this explicitly for the remainder of the paper.

We next show that the expected number of failed cuts remains small, even after conditioning on rare events:
\begin{proposition}
\label{zbar-cond-prop}
Let $\mathcal E$ be an event with probability $\bP(\mathcal E)$. Then
$$
\bE[Z \mid \mathcal E] \leq O( \bP(\mathcal E)^{-\frac{2}{\delta \log_2 n}} )
$$
\end{proposition}
\begin{proof}

Let $q = \bP(\mathcal E)$.  If $n < O(1)$, then $Z \leq 2^n \leq O(1)$ while $q^{-\frac{2}{\delta \log_2 n}} \geq 1$, and so the result holds. So we may assume that $n$ is larger than any necessary constant throughout this proof.

Then:
{\allowdisplaybreaks
\begin{align*}
\bE[Z \mid \mathcal E] &= \bE[2^{R-1} - 1 \mid \mathcal E] \leq \sum_{r=1}^{\infty} 2^{r-1} \bP( R \geq r \mid \mathcal E) \leq \sum_{r=1}^{\infty} 2^{r-1} \min(1, \frac{\bP( R \geq r)}{q}) \\
&\leq \sum_{r=1}^{\infty} 2^{r-1} \min(1, \frac{1.01}{n^{r \delta/2} q}) \qquad \text{by Proposition~\ref{Amultiple-cut-lemma1}} \\
&\leq \sum_{r=1}^{\lceil x \rceil -1} 2^{r-1} + \sum_{r=\lceil x \rceil}^{\infty} \frac{1.01 \times 2^{r-1}}{q n^{r \delta/2}}  \qquad \text{where $x = \frac{2 \log(1/q)}{\delta \log n}$} \\
&\leq 2^{\lceil x \rceil} + \frac{1.01 \times 2^{\lceil x \rceil - 1} n^{ (\delta/2) (1 - \lceil x \rceil) }}{q (n^{\delta/2} - 2) } \qquad \text{as $2 n^{-\delta/2} < 1$} \\
&\leq 2^{x+1} + \frac{2^{x+1} n^{ (\delta/2) (1 -  x ) }}{q (n^{\delta/2} - 2) }  \\
&\leq 2^{x+1} + \frac{2^{x+1} n^{ (\delta/2) (1 -  x ) }}{q (n^{\delta/2}/2)} \qquad \text{for $n > 2^{4/\delta}$} \\
&\leq O( 2^{\frac{2 \log(1/q)}{\delta \log n}} ) = O( q^{-\frac{2}{\delta \log_2 n}} ) 
\end{align*}
 }
\end{proof}

As an application of Proposition~\ref{zbar-cond-prop}, we show that $\bar Z$ is nearly equivalent asymptotically to $U(p)$, although the former is much more tractable.
\begin{proposition}
\label{Aestimatelemma}
For $\delta \geq \delta_0 > 0$, we have $U(p) = \Theta(\bar Z)$ .
\end{proposition}
\begin{proof}

By the union bound, $U(p) \leq \bar Z$.

Now note that $\bE[Z] = \bE[Z \mid Z\geq 1] \bP(Z \geq 1)$. So $U(p) = \bP(Z \geq 1) = \frac{\bE[Z]}{\bE[Z \mid Z \geq 1]}$. So it suffices to show that $\bE[Z \mid Z \geq 1] = O(1)$. For this we have:
\begin{align*}
\bE[Z \mid Z \geq 1] &\leq O( U(p)^{-\frac{2}{\delta \log_2 n}} ) \qquad \text{by Proposition~\ref{zbar-cond-prop}} \\
&\leq O( (n^{-2-\delta})^{-\frac{2}{\delta \log_2 n}} ) \qquad \text{as $U(p) \geq p^{c} = n^{-2-\delta}$} \\
&=O( 2^{2 + 4/\delta}) \leq O(2^{2 + 4/\delta_0}) = O(1)
\end{align*}
\end{proof}

As in Karger, our strategy will be to use Monte Carlo estimation for $U(p)$ when $U(p)$ is large. The cut-enumeration will be only be performed when $U(p)$ is small. For this, we have the following useful bounds:
\begin{proposition}
\label{k-prop}
Let $K > 2$ be an arbitrary constant. Then for $U(p) < n^{-K}$ and $n$ sufficiently large, we have the following:
\begin{enumerate}
\item $\delta \geq K-2 > 0$.
\item $\delta \geq \beta + \Omega(1)$.
\item $\beta \in [0,2]$.
\end{enumerate}
\end{proposition}
\begin{proof}
First, observe that if a min-cut fails, then $G$ fails, so $U(p) \geq p^c = n^{-2-\delta}$. This implies that $\delta \geq K - 2$. So we may set $\delta_0 = K - 2 > 0$, which satisfies (1).

Also, by Proposition~\ref{Aestimatelemma}, we have $U(p) = \Theta(\bar Z) = \Theta(n^{-2-\delta+\beta})$.  Hence $\beta \leq 2 + \delta - K + O(1/\log n) \leq -\Omega(1) + \delta + O(1/\log n)$. This implies that for $n$ sufficiently large $\beta \leq \delta - \Omega(1)$.

The bound $\beta \in [0,2]$ is simply a restatement of Proposition~\ref{Azbarboundcorr}.
\end{proof}

\textbf{For the majority of this paper,  we will assume that $U(p) < n^{-K}$ for some fixed $K > 2$, and that $n$ is sufficiently large. Thus, all of the conclusions of Proposition~\ref{k-prop} hold. We will note this at the beginning of each relevant section; all the results in that section may depend on this assumption and all the resulting asymptotic notations may depend upon $K$.}

\section{Estimating the failure probability for a collection of cuts}
\label{Asec:statistical}
In this section, we address the following statistical problem. Suppose we are given some collection $\mathcal A$ consisting of $N$ cuts, which includes at least one min-cut. We define the \emph{failure probability of $\mathcal A$} as 
$$
U_{\mathcal A}(p) = \bP( \text{at least one cut from $\mathcal A$ fails} )
$$

\textbf{In this section, we will assume $U(p) < n^{-K}$ for some constant $K > 2$; thus the conclusions of Proposition~\ref{k-prop} hold.}

Our goal here is to estimate $U_{\mathcal A} (p)$ to within relative error $\epsilon$. This is a key subroutine for Karger's algorithm, and it will also be used for ours. For this, Karger uses an elegant algorithm developed by Karp, Luby, Madras \cite{lubydnf}; in this context, that algorithm has a running time of roughly $n^2 N \epsilon^{-2}$. We will give a new algorithm running in time $O( N \text{polylog}(Nn) + n^2 \epsilon^{-2})$.

We make some elementary observations first. Let us define the random variable
$$
Z_{\mathcal A} = \sum_{C \in \mathcal A} [\text{$C$ fails}],
$$
where we use the Iverson notation so that $[\text{$C$ fails}]$ is $1$ if $C$ fails, and 0 otherwise. 

Then $U_{\mathcal A}(p) = \bP( Z_{\mathcal A} \geq 1 ) $. Clearly $U_{\mathcal A}(p) \leq U(p)$ and $Z_{\mathcal A} \leq Z$. Also, $U_{\mathcal A} (p) \geq p^c$ as $\mathcal A$ contains at least one min-cut.

\subsection{Data structures for $\mathcal A$} 
\label{dsa-sec} Each cut in $\mathcal A$ might require $\Omega(n)$ space to store explicitly and thus even reading $\mathcal A$ might require $\Omega(N n)$ time. We will give a running time which may be smaller than this, and thus we must store and generate $\mathcal A$ in a compressed form. 

Instead of storing the full information about each cut $C \in \mathcal A$, we will store three smaller pieces of information, each taking just $O(\log(N n))$ words of storage, instead: 
\begin{enumerate} 
\item A unique identifier for $C$
\item The weight of $C$
\item A pointer which allows us to fully reconstruct $C$ in time $O(n^2)$
\end{enumerate}

In order to assign every cut of $\mathcal A$ a unique identifier, we may use a simple random hashing scheme. We choose a random function $H: V \rightarrow [2^b]$, where $b = \phi \log N$ and $\phi > 0$ is some constant. We then define the identifier for a cut $C$ which has shores $A, V - A$ with vertex $1$ in $A$, by $H(C) = \sum_{v \in A} b(v)$ where the sum is taken modulo $2^b$. It is not hard to see that for $\phi$ a sufficiently large that with high probability all cuts of $\mathcal A$ receive distinct identifiers.

Because the cuts of $C$ are assigned distinct identifiers, we may sort $\mathcal A$ and remove any duplicates in time $O( N \log^2 (N n))$. Thus, we may assume that after an additional pre-processing step, we have that $N \leq 2^n$. This is a crude bound, but it will be sufficient for many purposes of analyzing our algorithm.

\begin{proposition}
\label{check-in-a-prop}
Suppose that $\mathcal A$ has been sorted. Then given any cut $C$ of $G$, there is an algorithm running in time $O(n^2)$ to determine if $C \in \mathcal A$.
\end{proposition}
\begin{proof}
We compute $H(C)$ in time $O(n)$. As all the cuts of $\mathcal A$ receive distinct identifiers, there is at most one $C' \in \mathcal A$ with $H(C') = H(C)$. An index to the cut $C'$ can be found via binary search in time $O(b \log N) \leq O(\log^2 N)$. Using our crude bound $N \leq 2^n$, this is at most $O(n^2)$.

Now, if there is no such $C' \in \mathcal A$, then we know that $C \notin \mathcal A$ and we are done. Otherwise, $C \in \mathcal A$ if and only if $C = C'$. To check this, we use the pointer to fully reconstruct $C'$ from its index; this step takes $O(n^2)$. At this stage, both $C$ and $C'$ are represented explicitly so they can be compared in time $O(n)$.
\end{proof}

\subsection{The statistical estimation algorithm}
We take as our starting point a simpler algorithm, which was presented as a toy or warm-up exercise in Karp, Luby, Madras \cite{lubydnf}. We will see how to implement it faster than their alternative, more sophisticated algorithms.
\ttfamily
\begin{enumerate}
\item[1.] Select a cut $Q$ from $\mathcal A$ with probability $\propto p^{|Q|}$.
\item[2.] Let $L$ be a random subset of the edges in $G - Q$, in which each such edge is chosen independently with probability $p$. Let $H = G - Q - L$.
\item[3.] Count how many cuts in $\mathcal A$ have failed in $H$; let $J$ denote the number of such cuts.
\item[4.] Estimate $\hat U_{\mathcal A} = \frac{\sum_{C \in \mathcal A} p^{|C|}}{J}$.
\end{enumerate}
\rmfamily

We can perform step (1) by choosing a random real number in the range $[0, \sum p^{|C|}]$ and then performing a binary search; as $\mathcal A$ explicitly stores the weight of all the cuts, this step can be implemented (with some pre-processing) in time $O(\log (N n))$.

The main cost of this algorithm is step (3). If $\mathcal A$ is an arbitrary collection of edge sets, then seemingly the only way to do this would be to check each member of $\mathcal A$ to see whether it failed in $H$. This would require roughly $O(m N)$ time per execution. This high running time is the main motivation for the more sophisticated algorithm of \cite{lubydnf}. 

However, $\mathcal A$ is a collection of \emph{cuts}, and we can significantly accelerate step (3) by taking advantage of this fact. We implement step (3) by enumerating all the failed cuts of $H$, and checking for each if it is present in $\mathcal A$. We will show that $H$ typically has a small number of cuts, and so this step is relatively inexpensive.

We first show that $\hat U_{\mathcal A}$ is a good estimator for $U_{\mathcal A}(p)$, which in turn is a good estimator for $U(p)$.
\begin{proposition}
\label{zaprop}
For any cut $C \in \mathcal A$, the distribution of $H$ in step (2) of this estimation procedure, conditional on $Q = C$, is the same as the distribution of the graph $G'$ resulting from edge failures, conditional on the event that $C$ fails.

For any cut $C \in \mathcal A$, the distribution of $J$ conditional on $Q = C$ is the same as the distribution of $Z_{\mathcal A}$ conditional on the event that $C$ fails.
\end{proposition}
\begin{proof}
Because the edge failures are independent, then conditioning on the event that $C$ fails is equivalent to removing the edges of $C$ and letting the remaining edges fail independently. This is exactly the random distribution of graphs selected in step (2) of our estimation algorithm.
\end{proof}

\begin{proposition}[\cite{lubydnf}]
The statistic $\hat U_{\mathcal A}$ is an unbiased estimator, viz.
$$
\bE[\hat U_{\mathcal A}] = U_A(p)
$$
\end{proposition}
\begin{proof}
The proof is shown originally in \cite{lubydnf}; we present a modified proof here for completeness.

Consider the following random process: we randomly allow edges to fail with probability $p$. If $Z_{\mathcal A} \geq 1$, then we choose one failed cut of $\mathcal A$ uniformly at random, and mark it. 

Note that some cut of $\mathcal A$ fails if and only if exactly cut in $\mathcal A$ is marked. So
{\allowdisplaybreaks
\begin{align*}
U_{\mathcal A}(p) &= \sum_{C \in {\mathcal A}} P(\text{$C$ marked}) \\
&= \sum_{C \in \mathcal A} P(\text{$C$ fails}) \bE[1/Z_{\mathcal A} \mid \text{$C$ fails}] \\
&= \sum_{C \in \mathcal A} p^{|C|} \bE[1/J \mid Q = C] \qquad \text{by Proposition~\ref{zaprop}} \\
&= \sum_{C \in \mathcal A} P(Q = C) \bigl( \sum_{C \in A} p^{|C|} \bigr) \bE[1/J \mid Q = C] \\
&= \bigl( \sum_{C \in \mathcal A} p^{|C|} \bigr) \bE[1/J] = \bE[ \hat U_{\mathcal A} ]
\end{align*}
}
\end{proof}

\begin{proposition}
\label{z-exp-cond-prop}
Suppose that $C \in \mathcal A$. Let $H$ be the graph resulting from removing $C$ and allowing other edges to fail independently. Let $Z$ denote the total number of failed cuts of $H$. Then
$$
\bE[Z \mid Q = C] \leq e^{O(|C|/c)} 
$$
\end{proposition}
\begin{proof}
We have that
\begin{align*}
\bE[Z \mid Q = C] &= \bE[ Z \mid \text{$C$ fails} ] \qquad \text{by Proposition~\ref{zaprop}} \\
&\leq O( (p^{|C|} )^{-\frac{2}{\delta \log_2 n}}) \qquad \text{by Proposition~\ref{zbar-cond-prop}} \\
&= O(2^{\frac{2 (|C|/c) (2+\delta)}{\delta}}) \\
&= 2^{O(|C|/c)}
\end{align*}
\end{proof}

\begin{proposition}
Suppose that $\mathcal A$ contains a min-cut. Then
$$
\bE[ \hat U_{\mathcal A}^2] \leq O( \bE[ \hat U_{\mathcal A}]^2 )
$$
That is, the statistic $\hat U_{\mathcal A}$ has relative variance $O(1)$.
\end{proposition}
\begin{proof} 
We have $\bE[ U_{\mathcal A}^2 ] / \bE[U_{\mathcal A}]^2 = \bE[1/J^2] / \bE[1/J]^2 \leq 1/\bE[1/J]^2$. Thus, it will suffices to show $\bE[1/J] \geq \Omega(1)$. 

First, we claim that there is some constant $x > 1$ such that $\bP(|Q| \leq x c) \geq 1/2$. For,
{\allowdisplaybreaks
\begin{align*}
\bP(|Q| \geq x c) &= \frac{\sum_{C \in \mathcal A,|C| \geq x c} p^c}{\sum_{C \in \mathcal A} p^c} \\
&\leq \frac{\sum_{|C| \geq x c} p^{|C|}}{n^{-2-\delta}} \qquad \text{as $\mathcal A$ contains a min-cut} \\
&\leq O(n^{2+\delta} n^{-(2+\delta) x} n^{2 x}) \qquad \text{by Corollary~\ref{Aexpsumcorr1}} \\
&< 1/2, \qquad \text{for $x$ sufficiently large constant}
\end{align*}
}

Now, conditional on the event that $|Q| \leq x c$, by Proposition~\ref{z-exp-cond-prop} we have that $\bE[Z] \leq e^{O(x)}$. 

As $J \leq Z$, this implies that $\bE[1/J] \geq \bE[1/J \mid Q \leq x c] \bP( Q \leq x c) \geq 1/(x c) P( Q \leq x c) \geq e^{-O(x)} \times 1/2 \geq \Omega(1)$.

\end{proof}

Finally, we prove that we can execute this estimation procedure relatively quickly --- much more quickly than by reading in the entire collection $\mathcal A$:
\begin{proposition}
Suppose that $\mathcal A$ contains a min-cut. Then, after preprocessing steps, the expected running time of the estimation algorithm is $O(n^2)$.
\end{proposition}
\begin{proof}
We can enumerate the connected component structure of $H$ in time $O(n^2)$. By Proposition~\ref{check-in-a-prop}, it requires $O(n^2)$ to check whether any cut of $H$ is present in $\mathcal A$. Thus, we may compute $J$ in overall running time is $O(Z n^2)$. Thus, the expected time for this step is $\leq O(n^2) \bE[Z]$.  We estimate $\bE[Z]$ by breaking up the sum in contributions from small cuts (weight is $\leq x c$) and big cuts (weight $> x c$); here $x$ is a parameter we will discuss later.
{\allowdisplaybreaks
\begin{align*}
\bE[Z] &= \bE[Z \mid |Q| \leq x c] \bP(|Q| \leq x c) + \sum_{\substack{C \in \mathcal A \\ |C| > x c}} \bE[Z \mid Q = C] \bP( Q = C) \\
&= 2^{O(x)} + \sum_{C \in \mathcal A, |C| > x c} \bE[Z \mid Q = C] \bP( Q = C) \qquad \text{Proposition~\ref{z-exp-cond-prop}} \\
&\leq 2^{O(x)} + \sum_{C \in \mathcal A, |C| > x c} \bE[Z \mid Q = C] \frac{p^{|C|}}{\sum_{C' \in \mathcal A} p^{|C'|}} \\
&\leq 2^{O(x)} + n^{2+\delta} \sum_{C \in \mathcal A, |C| > x c} \bE[Z \mid Q = C] p^{|C|} \qquad \text{as $\mathcal A$ contains a min-cut} \\
&\leq 2^{O(x)} + n^{2+\delta}  \sum_{C \in \mathcal A, |C| > x c} \bE[Z \mid \text{$C$ failed}] p^{|C|} \qquad \text{Proposition~\ref{zaprop}} \\
&\leq 2^{O(x)} + n^{2+\delta}  \sum_{|C| > x c} e^{O(|C|/c)} p^{|C|} \qquad \text{Proposition~\ref{z-exp-cond-prop}} \\
&\leq 2^{O(x)} + n^{2+\delta} O(n^{-\delta x}) \qquad \text{Corollary~\ref{Aexpsumcorr1}}
\end{align*}
}

Now setting $x = 1 + 2/\delta = O(1)$, we see that this overall expression is $O(1)$.

\end{proof}

We summarize our estimation procedure:
\begin{proposition}
\label{summarize-est-prop}
Let $\mathcal A$ be a collection of cuts which includes at least one min-cut. Then there is an algorithm for estimating $U_{\mathcal A} (p)$ to within relative error $\epsilon$ with probability $> 0.99$, and running in time $O(N \text{polylog}(N n) + n^{2} \epsilon^{-2})$.
\end{proposition}
\begin{proof}
We have pre-preprocessing steps including sorting $\mathcal A$ and computing $\sum_{C \in \mathcal A} p^{|C|}$ which can be executed in time $O(N \text{polylog}(N n))$.

After this, a single iteration has relative error $O(1)$ and costs $O(n^{2})$ expected time. We repeat $\lambda \epsilon^{-2}$ independent trials of this algorithm. This reduces the relative variance of the resulting unbiased statistic to $O(\lambda^{-1} \epsilon^{-1})$. By Chebyshev's inequality, for  $\lambda$ sufficiently large, the resulting statistic estimates $U_{\mathcal A}(p)$ to relative error $\epsilon$ with probability $> 0.99$.
\end{proof}

\section{The Contraction Algorithm}
\label{Asec:contraction}
We now turn our attention to how to find a representative set of cuts to use to estimate $U(p)$. Key to this is an algorithm of \cite{karger-stein} for finding cuts in a graph. This algorithm is called the Contraction Algorithm, and is based on a graph transformation called \emph{contraction} which will appear throughout this paper. Given a graph $G$ and an edge $e = \langle x, y \rangle$ of $G$, then we define the \emph{contraction} $G \slash e$ by identifying $x$ and $y$. More formally, we create a new vertex labeled $xy$, we delete the vertices $x,y$ and for each edge $\langle x, v \rangle \in G$ we create a new edge $\langle xy , v \rangle$; similarly, for each each $\langle y, v \rangle \in G$ we create a new edge $\langle xy, v \rangle$. If there are any loops created by this process, we delete them.

Suppose that a graph $H$ is derived from $G$ by performing a series of contractions. Then we say that $H$ is a \emph{contraction-subgraph} of $G$. Every vertex of $H$ corresponds to a set of vertices in $G$, any edge of $H$ corresponds to an edge of $G$, and any cut of $H$ corresponds to a cut of $G$. Given a cut $C$ of $G$, we say that $C$ \emph{survives} in $H$ if $C$ corresponds to some cut of $H$; equivalently, this holds if every edge of $C$ corresponds to some edge of $H$. We often abuse notation, so that for any edge $e \in H$ we also write $e$ to refer to the corresponding edge of $G$.

For any set of edges $L = \{e_1, \dots, e_l \}$, we may also define $G \slash L$ as follows. We initially set $H_1 = G$ and proceed through the edges $e \in L$ one by one; if edge $e_i$ is present in $H_i$ then $H_{i+1} = H_i \slash e_i$; otherwise, $H_{i+1} = H_i$. We then define $G \slash L = H_{l+1}$. Although this definition depended on the ordering of the edges in $L$, it is a simple exercise to see that changing this ordering does not change $G \slash L$.

The Contraction Algorithm takes as input a parameter $\alpha$, which is roughly speaking the size of the cuts that it is designed to find. We define the Contraction Algorithm with parameter $\alpha$ as follows:
\ttfamily
\begin{enumerate}
\item[1.] Initialize $G_n = G$. 
\item[2.] Repeat for $i = n, n-1, \dots, \lceil 2 \alpha \rceil + 1$:
\begin{enumerate}
\item[3.] Select an edge $F_i$ of $G_i$ uniformly at random.
\item[4.] Set $G_{i-1} \leftarrow G_i \slash F_i$
\end{enumerate}
\item[5.] Output a cut $C$ sampled uniformly from $G_{\lceil 2 \alpha \rceil}$.  (We say the Contraction Algorithm \emph{selects $C$})
\end{enumerate}
\rmfamily

At each stage $i$ of the Contraction  Algorithm, the graph $G_i$ has exactly $i$ vertices. We let $M_i$ denote the number of edges in $G_i$. We refer to the sequence of edges $\langle F_n, F_{n-1}, \dots, F_{\lceil 2 \alpha \rceil} \rangle$ as the \emph{history} of the Contraction Algorithm; note that the full sequence of graphs $G_i$ is completely determined by $F_n, \dots, F_{\lceil 2 \alpha \rceil + 1}$. 

We may also refer to the history of the Contraction Algorithm up to stage $i$, which we denote by $F_{> i}$; by this we mean the sequence of edges $F_{> i} = \langle F_n, \dots, F_{i+1} \rangle$; these determine $G_n, \dots, G_i$ uniquely.

We are only partially interested in the Contraction Algorithm as an algorithm \emph{per se}. We will often instead view it as a stochastic process which allows us to analyze the number of cuts that are present in the graph. We refer to this stochastic process as $\mathcal{CA}(G)$. Thus we may write, for instance $\bE_{\mathcal{CA}(G)}[M_i]$; this means that we are taking the expected value of the number of edges present in the graph $G_i$, when we run the Contraction Algorithm starting with graph $G$.

We will show that, for an arbitrary cut $C$, the Contraction Algorithm selects $C$ with at least a certain probability. In this type of analysis, we regard $C$ as \emph{fixed}. We refer to $C$ as the \emph{target cut}. In order for the Contraction Algorithm to find the target cut $C$, a necessary condition is that $C$ must remain in the final $G_{\lceil 2 \alpha \rceil}$. (This is not a sufficient condition, as we still must randomly select $C$ among the cuts of $G_{\lceil 2 \alpha \rceil}$.) If $C$ remains in $G_{\lceil 2 \alpha \rceil}$, then we say that $\mathcal{CA}(G)$ \emph{succeeded} and we say that $C$ \emph{survived}. We also say that $C$ \emph{survived} the Contraction Algorithm. We emphasize that this only for the purposes of analysis; the Contraction Algorithm itself is run without any target in mind.

When we run the Contraction Algorithm with parameter $\alpha$, we say that a cut $C$ \emph{survived to stage $i$} if $C$ remains in $G_i$. Sometimes we may say that $C$ survives to stage $i$ without specifying the parameter $\alpha$; here we assume that $\alpha$ is any arbitrary real number which is at most $i/2$. This is acceptable because $G_i$ does not depend on $\alpha$ itself, only whether $i \geq 2 \alpha$.

\smallskip \noindent \textbf{The contraction process for $C$.} When we fix a target cut $C$, we may imagine a related stochastic process which we refer to as the \emph{Contraction Process}. This process may be defined more broadly in terms of any edge-set $L$, and we denote it by $\mathcal{CP}(G,L)$:
\ttfamily
\begin{enumerate}
\item[1.] Initialize $G_n = G$. 
\item[2.] Repeat for $i = n, n-1, \dots, 1$:
\begin{enumerate}
\item[3.] If all the edges of $G_i$ are in $L$, then terminate the loop.
\item[4.] Otherwise, select an edge $F_i$ of $G_i - L$ uniformly at random.
\item[5.] Set $G_{i-1} \leftarrow G_i \slash F_i$
\end{enumerate}
\end{enumerate}
\rmfamily

If this process terminates at stage $i$ (that is, all the edges of $G_i$ are in $L$), then we say that $G_j = \bot$ for $j < i$. Note that it is possible for some edge of $L$ to become contracted, even though they are not themselves selected, if another edge with the same end-points is selected. This never happens if $L$ is a cut; in the Contraction Process for $C$, all the edges of $C$ remain present in the resulting contraction-subgraphs. 

Finally, observe that one may view the Contraction Algorithm as a special case, namely $\mathcal{CA} = \mathcal{CP}(G, \emptyset)$.

For $i = 1, \dots, n$ let $M_i$ denote the number of edges in the graph $G_i$. We begin with some useful elementary bounds on $M$. Every vertex of $v$ of the graph $G_r$ defines a cut, namely, $v$ is on one shore and all the other vertices are on the other shore. As the minimum cut size is $c$, this implies that every vertex must have degree $\geq c$, so $M_r \geq r c/2$.  Thus during the process $\mathcal {CP}(G, L)$, as long as $|L| \leq r c / 2$ then $G_r \neq \bot$. In fact, this bound will always hold in our analysis, so that we will never need to deal with the case of $G_r = \bot$.

The following proposition shows the connection between the Contraction Algorithm and the Contraction Process for a cut $C$:
\begin{proposition}
\label{cp-ca-prop}
Suppose that $e_n, \dots, e_{i+1}$ are edges of $G$ but not edges of the cut $C$ and let $M_{i+1}, \dots, M_n$ denote their edge counts. 

Then
$$
\bP_{\mathcal{CA}(G)}( F_{> i} = \langle e_n, \dots, e_{i+1} \rangle ) = \Bigl( \prod_{r=i+1}^n (1 - \frac{|C|}{M_r}) \Bigr) \bP_{\mathcal{CP}(G,C)} (F_{> i} = \langle e_n, \dots, e_{i+1} \rangle)
$$
when the Contraction Algorithm is applied with parameter $\lceil 2 \alpha \rceil \leq i$.
\end{proposition}
\begin{proof}
We prove this by induction on $G$. The base case is when $G$ has $i$ vertices; in this case both the LHS and RHS are easily seen to be one.

So we now consider the induction step. In the first stage of the Contraction Algorithm, we select some edge $F_n$ uniformly from $G$. The probability that $F_n = e_n$ is $\frac{1}{M_n}$. Conditional on this event, the remaining stages of the Contraction Algorithm are equivalent to the Contraction Algorithm applied to the graph $G \slash F_n = G \slash e_n$. The edge $e_n$ is not present in $C$, and so cut $C$ remains a cut of $G \slash e_n$. So we can apply the induction hypothesis to it.
\begin{align*}
\bP_{\mathcal{CA}(G)}( F_{> i} = \langle e_n, \dots, e_{i+1} \rangle ) &= \frac{1}{M_n} \bP_{\mathcal{CA}(G \slash e_n)} (F_{> i} = \langle e_{n-1}, \dots, e_{i+1} \rangle) \\
&= \frac{1}{M_n}  \Bigl( \prod_{r=i+1}^{n-1} (1 - \frac{|C|}{M_r}) \Bigr) \bP_{\mathcal{CP}(G \slash e_n,C)} (F_{> i} = \langle e_{n-1}, \dots, e_{i+1} \rangle) \qquad \text{ind. hypothesis} \\
&= \Bigl( \prod_{r=i+1}^{n} (1 - \frac{|C|}{M_r}) \Bigr) \frac{1}{M_n - |C|} \bP_{\mathcal{CP}(G \slash e_n,C)} (F_{> i} = \langle e_{n-1}, \dots, e_{i+1} \rangle) \\
&= \Bigl( \prod_{r=i+1}^{n} (1 - \frac{|C|}{M_r}) \Bigr) \bP_{\mathcal{CP}(G,C)} (F_{> i} = \langle e_n, e_{n-1}, \dots, e_{i+1} \rangle)
\end{align*}
and the induction is proved.
\end{proof}

\begin{corollary}
\label{Akillmanylemma1}
Let $C$ be any cut. The probability that cut $C$ survives to the $i^{\text{th}}$ stage of the Contraction Algorithm (i.e. the probability that cut $C$ is present in the graph $G_i$), is exactly equal to  
$$
\bP(\text{$C$ survives to $G_i$}) = \bE_{\mathcal CP(G, C)} \Bigl[ \prod_{r=i+1}^n (1 - \frac{|C|}{M_r}) \Bigr]
$$
\end{corollary}
\begin{proof}
We sum  over all partial histories up to stage $i$:
\begin{align*}
\bP(\text{$C$ survives to $G_i$}) &= \sum_{\substack{e_n, \dots, e_{i+1} \\ \text{all disjoint to $C$}}} \bP_{\mathcal{CA}(G)} (F_{> i} = \langle e_n, \dots, e_{i+1} \rangle) \\
 &= \sum_{\substack{e_n, \dots, e_{i+1} \\ \text{all disjoint to $C$}}} \Bigl( \prod_{r=i+1}^n (1 - \frac{|C|}{M_r}) \Bigr) \bP_{\mathcal{CP}(G, C)} (F_{> i} = \langle e_n, \dots, e_{i+1} \rangle) \\
 &= \bE_{\mathcal{CP}(G,C)} \Bigl[ \prod_{r=i+1}^n (1 - \frac{|C|}{M_r}) \Bigr]
\end{align*}
\end{proof}

One difficulty with using Proposition~\ref{cp-ca-prop} and Corollary~\ref{Akillmanylemma1} is the fact that the expression $\prod_{r=i+1}^n (1 - |C|/M_r)$ is a non-linear function of $|C|$ and $M_{i+1}, \dots, M_n$. We often prefer to use a ``linearized'' potential function $S_i$ that approximates it. For any  integer $i \geq 0$ we define
$$
S_i = \frac{c}{M_{i+1}} + \frac{c}{M_{i+2}} + \dots + \frac{c}{M_n}
$$
Heuristically, one may think of $S_i$ as counting the expected number of times a min-cut is selected during the Contraction Algorithm. We have the following simple bound on $S_i$:
\begin{proposition}
\label{Aeqn:simple-s-bound}
For any integer $i \in [1,n]$ we have:
$$
S_i \leq 2 \log(n/i)
$$
\end{proposition}
\begin{proof}
We have that:
\begin{align*}
S_i &= \sum_{r=i+1}^n \frac{c}{M_r} \leq \sum_{r=i+1}^n \frac{c}{r c/2} \leq \int_{r=i}^n \frac{2 dr}{c} = 2 \log(n/i).
\end{align*}
\end{proof}

The next Lemma shows how the potential function $S_i$ gives a good approximation to Proposition~\ref{Akillmanylemma1}.
\begin{lemma}
\label{Aslemma}
Let $C$ be a cut of weight $\alpha c$. For any real number $\alpha \geq 1$ and integer $i \geq 2 \alpha$ we have that:
$$
16^{-\alpha} e^{-\alpha S_i} \leq \prod_{r=i+1}^n (1 - \frac{|C|}{M_r}) \leq e^{-\alpha S_i}
$$
\end{lemma}
\begin{proof}
The upper bound is easy:
\begin{align*}
\prod_{r = i+1}^n (1 - \frac{\alpha c}{M_r})] &\leq \prod_{r = i+1}^n e^{- \frac{\alpha c}{M_r}} = e^{-\sum_{r=i+1}^n \frac{\alpha c}{M_r}} = e^{-\alpha S_i}
\end{align*}

To show the lower bound:
\begin{align*}
\prod_{r=i + 1}^{n} (1 - \frac{\alpha c}{M_r}) &\geq \prod_{r=i + 1}^{n} e^{- \frac{\alpha c}{M_r}} \times \bigl( 1 - \bigl( \frac{\alpha c}{M_r} \bigr)^2 \bigr) \qquad \text{as $1 - x \geq e^{-x} (1 - x^2)$ for $x \in [0,1]$} \\
&=e^{-\alpha S_i} \times \prod_{r=i+1}^n \bigl( 1 - \bigl( \frac{\alpha c}{M_r} \bigr)^2 \bigr)  \\
&\geq e^{-\alpha S_i} \times \prod_{r=i+1}^n \bigl( 1 - \frac{\alpha^2 c^2}{r^2 c^2/4} \bigr) \qquad \text{as $M_r \geq r c/2$}  \\
&= e^{-\alpha S_i} \times \prod_{r=i+1}^n \bigl( 1 - \frac{4 \alpha^2}{r^2} \bigr)
\end{align*}

Thus, we need to estimate the expression $\prod_{r = i + 1}^n (1 - \frac{4 \alpha^2}{r^2})$. This is an exercise in calculus:
\begin{align*}
\prod_{r = i + 1}^n (1 - \frac{4 \alpha^2}{r^2}) &= \exp \Bigl( \sum_{r = i + 1}^n \log(1 - \frac{4 \alpha^2}{r^2}) \Bigr) \\
&\geq \exp\Bigl(  \int_{r = 2 \alpha }^{\infty} \log(1 - \frac{4 \alpha^2}{r^2}) dr \Bigr) \qquad \text{as $i \geq 2 \alpha$} \\
&=  \lim_{\substack{x \rightarrow \infty\\y \rightarrow  2 \alpha^{+}}} \exp\Bigl( r \log(1-\frac{4 \alpha ^2}{r^2} )+4 \alpha  \tanh^{-1}(\frac{r}{2 \alpha }) \Bigr|_{y}^x \Bigr) \\
&= \exp( -2 \alpha \log 4) = 16^{-\alpha}\\
\end{align*}
which gives us the desired lower bound.
\end{proof}

This can be applied to obtain a simple bound for the probability of selecting a given cut $C$:
\begin{lemma}
\label{Aslemma2}
Let $C$ be any cut of weight $\alpha c$. For $\alpha' \geq \alpha$, the probability that the Contraction Algorithm with parameter $\alpha'$ selects cut $C$ is at least
$$
\bP(\text{Contraction Algorithm selects cut $C$}) \geq e^{-O(\alpha')} \bE_{\mathcal CP(G, C)} [ e^{-\alpha S_{\lceil 2 \alpha \rceil}} ]
$$
\end{lemma}
\begin{proof}
First, suppose that $\alpha' \geq n/2$. In this case, the Contraction Algorithm simply selects a cut uniformly from $G$, so that $C$ is selected with probability $2^{-n-1} = e^{-O(\alpha')}$. Also, observe that $S_{\lceil 2 \alpha \rceil} \leq 0$ so that $\bE_{\mathcal{CP}(G,C)}[e^{-\alpha S_{\lceil 2 \alpha \rceil}}] \leq 1$. So the result holds easily in this case.

Next, suppose that $\alpha' \leq n/2$. By Lemma~\ref{Aslemma2}, the probability that the cut $C$ survives to the graph $G_{\lceil 2 \alpha' \rceil}$ is at least
$16^{-\alpha'} \bE_{\mathcal CP(G, C)} [ e^{-\frac{\alpha c}{c} S_{\lceil 2 \alpha' \rceil}} ]$.

Next, note that if $C$ survives to $G_{\lceil 2 \alpha' \rceil}$, then as the Contraction Algorithm selects a cut from this graph uniformly at random, $C$ is selected with probability $2^{1 - \lceil 2 \alpha' \rceil} \geq 2^{-O(\alpha')}$.

Overall, the probability that $C$ is selected is at least $e^{-O(\alpha')} \bE_{\mathcal CP(G, C)} [ e^{-\alpha S_{\lceil 2 \alpha \rceil}}]$ as desired.
\end{proof}

Note that by Jensen's inequality,  
$$
 \bE_{\mathcal CP(G, C)} [ e^{-\alpha S_{\lceil 2 \alpha \rceil}} ] \geq \exp(-\alpha \bE_{\mathcal CP(G,C)}[S_{\lceil 2 \alpha \rceil}]);
$$
thus, to apply Lemma~\ref{Aslemma2}, it will suffice to compute $\bE[S_{\lceil 2 \alpha \rceil}]$. Also observe that by combining Proposition~\ref{Aeqn:simple-s-bound} with Lemma~\ref{Aslemma}, we get a simple and well-known lower bound on the probability of retaining $C$ that does not depend on any other properties of $G$:
\begin{corollary}
\label{Asimplecorr}
The probability that the Contraction Algorithm with parameter $\alpha$ selects a given $\alpha$-cut $C$ is at least
$$
\bP(\text{select C}) \geq e^{-O(\alpha)} \bigl( n/\alpha \bigr)^{-2 \alpha}. 
$$
\end{corollary}
\begin{proof}
This holds easily when $\alpha \geq n/2$. Otherwise,
\begin{align*}
\bP(\text{select C}) &\geq e^{-O(\alpha)} \bE_{\mathcal CP(G,C)}[ e^{-\alpha S_{\lceil 2 \alpha \rceil}} ] \geq e^{-O(\alpha)} \bE_{\mathcal CP(G,C)}[ e^{-\alpha \log( \frac{n}{\lceil 2 \alpha \rceil})}]  \\
&\geq e^{-O(\alpha)} \bE_{\mathcal CP(G,C)}[ e^{-\alpha \log( \frac{n}{2 \alpha})}]  =e^{-O(\alpha)} \bigl( n/\alpha \bigr)^{-2 \alpha}. 
\end{align*}
\end{proof}

The cut $C$ affects the dynamics of the Contraction Process, making it more difficult to analyze than the Contraction Algorithm. The following series of lemmas show how we can ``factor out $C$''.  We show that the Contraction Process for $C$ can be approximated by the Contraction Algorithm on $G \slash L$, where $L \subseteq C$ is a subset of edges from the cut $C$.

\begin{lemma}
\label{Alemma:stoch-dom}
Let $C$ be a cut of a connected graph $G$ and let $L$ a subset of the edges of $G$. Let $M_i$ be the edge counts corresponding to the random process $\mathcal{CP}(G,C)$ and let $M'_i$ be the edge counts corresponding to the random process $\mathcal{CP}(G \slash L, C)$.  

Then for $r > |L|$ the random variable $M_r$ stochastically dominates the random variable $|L| + M_{r - |L|}$.

Note that $C$ is not necessarily a cut of $G \slash L$; we interpret $C$ here only as a set of edges in $G \slash L$.
\end{lemma}
\begin{proof}
As described in \cite{karger}, there is an alternative description of the contraction process for $G, C$: we select a random permutation of the edges in $G - C$. We then process these edges in order, either contracting the edge or ignoring it if it was already contracted.

With this in mind, consider the following coupling process: given a permutation $\rho$ of the edges $e \in G - C$, we run respectively the processes $\mathcal CP(G, C)$ and $\mathcal CP(G \slash L, C)$, using the common permutation $\rho$ to order the relevant edges in both cases. Letting $M_i, M'_i$ denote the edge counts obtained in the respective processes, we claim that for a fixed $\rho$,
$$
M_r \geq M'_{r - |L|} + |L|
$$
and this will show our claim.

For a fixed $\rho$, list the edges of $G - C$ in order of $\rho$ as $e_1, e_2, \dots, e_{m-|C|}$. Observe that both Contraction Processes successively select these edges in this order, and contract them if they are still present in the graph (either $G$ or $G \slash L$ respectively). We say that the process is at stage $k$ if it has processed edges $e_1, \dots, e_k$ in order.

Suppose that when this process is at stage $k$, then the graph $G$ has reduced to $r$ vertices and contains $s$ edges of $L$. Then  $G \slash L$ at this stage must have $r' \geq r - s$ vertices, as it has at most $s$ additional edges contracted away. Also note that every edge in $G \slash L$ at stage $k$ also remains in graph $G$, and in addition $G$ has $s$ edges which are not present for $G \slash L$. Hence
$$
M_{r} \geq M_{r'} + s \geq M'_{r-s} + s
$$

Next, note that every time the number of vertices is decreased by one in the process, so too must the number of edges be reduced by at least one. Hence
$$
M'_{r-s} + s \geq M'_{r-s-1} + s + 1 \geq \dots \geq M'_{r-|L|} + |L|.
$$
and our claim is proved.
\end{proof}

\begin{corollary}

\label{corr-gslashl}
Let $C$ be any cut of $G$, and $L \subseteq C$ a subset of the edges of $C$. Then for any integers $j \geq i \geq 1$ we have
$$
\bE_{\mathcal{CP}(G,C)}[ S_{i}] \leq \bE_{\mathcal{CP}(G \slash L,C)}[ S_{j}] + 2 \log(\frac{|L|+j}{i})
$$
\end{corollary}
\begin{proof}
As $M_r \geq r c/2$ for all $r$, we have
\begin{align*}
\bE_{\mathcal{CP}(G,C)} [ S_{i}] &\leq \bE_{\mathcal{CP}(G,C)} [ S_{j + |L|}] + \sum_{r= i  + 1}^{ j  +|L|} \frac{c}{r c/2}  \\
 &\leq \sum_{r=j+|L|+1}^n \bE_{\mathcal{CP}(G,C)}[\frac{c}{M_r}] + 2\log( \frac{|L| + j}{i}  )
\end{align*}
For each $r$, the random variable $M_{r}$ stochastically dominates the random variable $|L| + M'_{r - |L|}$, where $M'_i$ denote the edge counts under $\mathcal{CP}(G \slash L, C)$. This implies that $\bE_{\mathcal{CP}(G, C)}[ \frac{c}{M_r}] \leq \bE_{\mathcal{CP}(G \slash L, C)}[\frac{c}{|L| + M_{r - |L|}}] \leq \bE_{\mathcal{CP}(G \slash L, C)}[\frac{c}{M_{r-|L|}}]$ and so:
\begin{align*}
\bE_{\mathcal{CP}(G,C)}[ S_{i}] &\leq \sum_{r= j +|L|+1}^n \bE_{\mathcal{CP}(G \slash L, C)}[\frac{c}{M_{r-|L|}}] +  2\log( \frac{|L|+j}{i} )  \\
&= \sum_{r= j +1}^{n-|L|} \bE_{\mathcal{CP}(G \slash L, C)}[\frac{c}{M_r}] +  2\log( \frac{|L|+j}{i} ) \\
&\leq \bE_{\mathcal{CP}(G \slash L, C)} [S_{j}] +  2\log( \frac{|L| + j }{i} )
\end{align*}
\end{proof}

To illustrate how we can use these results, suppose we are interested in the contraction process for $C$, where $C$ contains very few edges. In this case, if we apply Corollary~\ref{corr-gslashl}  with $L = C$ and $i = j = \lceil 2 \alpha \rceil$, then we would have
$$
\bE_{\mathcal{CP}(G,C)}[ S_{\lceil 2 \alpha \rceil}] \leq \bE_{\mathcal{CP}(G \slash C,C)}[ S_{\lceil 2 \alpha \rceil}] + 2 \log(1 + \frac{|C|}{\lceil 2 \alpha \rceil})
$$

However, note that $\mathcal {CP}(G \slash C, C)$ is equivalent to $\mathcal{CA}(G \slash C)$. Also, observe that as $|C|$ is small then $ 2 \log(1 + \frac{|C|}{2 \alpha})  \approx 0$.  Thus,
$$
\bE_{\mathcal{CP}(G,C)}[ S_{\lceil 2 \alpha \rceil}] \lesssim \bE_{\mathcal{CA}(G \slash C)} [S_{\lceil 2 \alpha \rceil}]
$$

We have thus completely factored out $C$ --- we have reduced the analysis of the Contraction Process for $C$ to the behavior of the Contraction Algorithm on a (slightly smaller) graph. 

We note that the choice of parameters $L = C, i = j = \lceil 2 \alpha \rceil$ for Corollary~\ref{corr-gslashl} was only for illustrative purposes, and is not optimal. When we apply this Corollary later, we will choose other parameters which lead to stronger, but more complicated, estimates.

\section{Bounds for small, odd $c$}
\label{Asec:oddc}
Our ultimate goal for this algorithm is show that the reliability of a graph influences the number of cuts it can have. As a warm-up exercise, we will show that graphs with connectivity $c$, where $c$ is a small odd number, have noticeably fewer $\alpha$-cuts than the worst case (where $c$ is even and the graph is a cycle with each edge having multiplicity $c/2$). \textbf{This section is not needed for our main algorithm, and can be skipped if desired.}

We use the following fact about the minimum cuts of $G$, when $c$ is odd. This is shown in \cite{bixby}, \cite{chandran}:
\begin{proposition}
\label{prop:mincuts-odd-c}
Suppose $G$ has minimum cut $c$, for $c$ odd. Then $G$ has at most $2 n$ min-cuts, which are represented by the edges of a spanning tree of $G$. 
\end{proposition}

This is actually the only place we use the fact that $c$ is odd; in fact, our result is true for any graph $G$ which has $\leq 2 n$ min-cuts (or a similar result could be shown for a graph with $O(n)$ min-cuts).
 
Although our goal is to analyze the Contraction Process for a target cut $C$, it will suffice to analyze the unconditioned Contraction Algorithm. The basic strategy of this proof is use induction on the number of vertices to show a bound on $\bE_{\mathcal{CA}(G)}[S_i]$. We will need to track the behavior of $S_i$ not only for the original graph $G$, but for subgraphs $G_r$ which arise during the evolution of the Contraction Algorithm. The formal induction proof by itself is not very intuitive, because it requires guessing a bound on $\bE[S_i]$ and then proving that this bound is correct. So we will give an intuitive and non-rigorous derivation of the proof. 

Consider any min-cut $C$ (not the target cut). By Corollary~\ref{Akillmanylemma1}, it survives to $G_i$ with probability $\bE_{\mathcal{CP}(G,C)}[\prod_{r=i+1}^n (1 - c/M_r)] \leq \bE_{\mathcal{CP}(G,C)}[e^{-S_i}]$. Now, suppose that we ignore the distinction between $\mathcal{CP}(G,C)$ and $\mathcal{CA}(G)$ (the two processes should be almost the same, because $C$ affects only a small part of the graph); in this case, we can approximate that $C$ survives to $G_i$ with probability roughly $\bE_{\mathcal{CA}(G)} e^{-S_i}$.

As this is true for \emph{any} min-cut, the expected number of min-cuts $K_i$ remaining in $G_i$ should be about 
$$
\bE[K_i] \approx k \exp(-\bE[S_i])
$$
where we are being deliberately vague about the scope of the expectation.

Now the neighborhood of each vertex of $G_i$ defines a cut. As $G_i$ has $K_i$ min-cuts, this implies that at most $K_i$ vertices may have the minimum degree $c$, while the others must have degree at least $c+1$, so that
$$
M_i \geq K_i c/2 + (i-K_i) (c+1)/2
$$

We have a bound on the expected value of the random variable $K_i$. It will turn out that, given this fixed value of $\bE[K_i]$, the worst case distribution on $K_i$ is that either $K_i = 0$ or $K_i = i$ (the latter occurring with probability $\bE[K_i]/i$). In this case,
$$
\bE[c/M_i] \leq \frac{\bE[K_i]}{i} \frac{c}{i c/2} + (1 - \frac{\bE[K_i]}{i}) \frac{c}{i (c+1)/2} = \frac{2 (i c + \bE[K_i])}{(c+1) i^2}
$$

This gives us a recurrence relation in $S_i$:
\begin{align*}
\bE[S_{i-1}] &= \sum_{j \geq i} \bE[\frac{c}{M_j}] = \bE[ S_ i ] + \bE[ \frac{c}{M_i} ] \\
& \approx \bE[S_i] + \frac{2 (i c + \bE[K_i])}{(c+1) i^2} \approx \bE[S_i] + \frac{2 (i c + k e^{-\bE[S_i]})}{(c+1) i^2}
\end{align*}

We relax this recurrence relation to a differential equation
\begin{align*}
\frac{d \bE[S_i]}{di} &= -\frac{2 (k e^{-\bE[S_i]} + c i)}{(c+1) i^2} \\
\bE[S_n] &= 0
\end{align*}
which can be solved in closed form to obtain
$$
\bE[S_i] = \log \Bigl[\frac{(i/n)^{\frac{2}{c+1}-1} (2 k + (c-1) n) - 2 k}{(c-1) i} \Bigr]
$$

This derivation makes a number of unwarranted independence and monotonicity assumptions on the behavior of the random variables, which do not hold in general. However, as we will see, this argument does accurately capture the \emph{worst-case} behavior for all the random variables. That is, even though the random variables are not independent, any dependency would only give us better bounds. 

In the following theorem, we carry out the high-level approach discussed in Section~\ref{stoch-proc-sec} and prove that our heuristic formula is in fact a correct bound:
\begin{theorem}
\label{Amainoddthm}
Let $c > 1$. Define the function
$$
f(i,n,k) = \log \Bigl[\frac{(i/n)^{\frac{2}{c+1}-1}  (2 k + (c-1) n) - 2 k}{(c-1) i} \Bigr]
$$

Suppose $G$ is a graph with $n$ vertices and a minimum cut weight of $c$. Suppose that $G$ has at most $k$ min-cuts, where $k$ is an integer in the range $0 \leq k \leq 2 n$. Then for $3 \leq i \leq n$ we have that
$\bE_{\mathcal{CA}(G)}[S_i] \leq f(i,n,k)$. 
\end{theorem}
\begin{proof}
For simplicity, we will defer some technical analysis of the function $f$ to Appendix~\ref{oddc-appendix}.

First, note that by Proposition~\ref{Aprop1x5}, the function $f$ is well-defined.

We induct on $n$. When $n = i$, we have $S_i = 0 = f(i,n,k)$.

Now suppose $n \geq i+1$, and $G$ has $m$ edges and $k \leq 2 n$ cuts of weight $c$. In the first step of the Contraction Algorithm, we select an edge of $G$ to contract, arriving at a new graph $G'$. So $\bE_{\mathcal{CA}(G)}[S_i] = c/m + \bE_{G'} \bE_{\mathcal{CA}(G')}[S_i]$. We have broken the expectation into two components. First, we randomly select the next subgraph $G'$; then, we continue the Contraction Algorithm on that subgraph.

The graph $G'$ has $n-1$ vertices and has $K' \leq 2 (n-1)$ min-cuts, where $K'$ is a random variable. Each min-cut survives to $H$ with probability $(1 - c/m)$, and so $\bE[K'] \leq k (1 - c/m) \leq k e^{-c/m}$.

By the inductive hypothesis,  $$
\bE_{\mathcal{CA}(G')}[S_i] \leq f(i,n-1,K')
$$

By Proposition~\ref{Aprop1x1}, this is a concave-down increasing function of $K'$, hence by Jensen's inequality
$$
\bE_{\mathcal{CA}(G')}[S_i] \leq f(i,n-1,k e^{-c/m})
$$
and hence
$$
\bE_{\mathcal{CA}(G)}[S_i] \leq c/m + f(i,n-1,k e^{-c/m})
$$

The neighborhood of each vertex of $G$ defines a cut, and for $n \geq 3$ these are all distinct. Hence at most $k$ vertices may have the minimum degree $c$, while the others must have degree at least $c+1$.

First suppose $k \leq n$. That implies that 
$$
m \geq n c/2 + (n-k)/2
$$

By Proposition~\ref{Aprop1x2}, the expression $c/m + f(i,n-1,k e^{-c/m})$ is decreasing in $m$. Then we have the bound
\begin{align*}
\bE_{\mathcal{CA}(G)}[S_i] &\leq \frac{c}{n c/2 + (n-k)/2} + f(i,n-1,k e^{- \frac{c}{n c/2 + (n-k)/2}}) \\
&\leq f(i,n,k) \qquad \text{by Proposition~\ref{A1x3}}.
\end{align*}

Suppose $k \geq n$. Then $m \geq n c/2$, so we have the bound
\begin{align*}
\bE_{\mathcal{CA}(G)}[S_i] &\leq \frac{c}{n c/2} + f(i,n-1,k e^{-\frac{c}{n c/2}} ) \\
&\leq f(i,n,k) \qquad \text{by Proposition~\ref{A1x4}}.
\end{align*}

This completes the induction.
\end{proof}

The function $f$ is technically not defined at $c = 1$. However, it approaches the limit $f(i,n,k) = \log \bigl( (n/i) + (k/i) \log (n/i) \bigr)$, and this function has similar properties to the case when $c > 1$. We omit the full analysis of the case when $c = 1$, as it is essentially identical to Theorem~\ref{Amainoddthm}. For that case, we obtain the result:
\begin{proposition}
\label{Aoddthm1}
Suppose $G$ is a graph with $n$ vertices and $c = 1$, and with $k \leq n$ weight-one cuts (aka bridges). Then for $n \geq i \geq 1$ we have
$$
\bE_{\mathcal{CA}(G)}[S_i] \leq \log\Bigl[(n/i) + (k/i) \log(n/i) \Bigr]
$$
\end{proposition}

\begin{lemma}
\label{Acrudeoddlemma}
Suppose $G$ has minimum cut $c$, for $c$ odd. Then for any integers $1 \leq i \leq n$ we have
$$
\bE[S_i] \leq \frac{2 c}{c+1} \log(n/i) + O(1)
$$
\end{lemma}
\begin{proof}
Observe that $S_1 \leq S_3 + O(1)$, so it suffices to show this for $3 \leq i \leq n$. Suppose that the graph $G$ has $n$ vertices and $k \leq 2 n$ cuts of weight $c$. If  $c \geq 3$, then by Theorem~\ref{Amainoddthm} we have
\begin{align*}
\bE_{\mathcal{CA}(G)} [S_i] &\leq f(i,n,2 n) \\
&=\log \left(\frac{(c+3) (n/i)^{2-\frac{2}{c+1}}-4 (n/i)}{(c-1)}\right) \\
&= (2 - \frac{2}{c+1}) \log(n/i) + \log \left(\frac{(c+3) - 4 (i/n)^{1 - \frac{2}{c+1}}}{(c-1)}\right) \\
&= \frac{2 c}{c+1} \log(n/i) + O(1)
\end{align*}

A similar proof, using Proposition~\ref{Aoddthm1}, applies when $c = 1$.
\end{proof}

We can now estimate the probability of selecting the $\alpha$-cut $C$:
\begin{theorem}
\label{Aoddcutthm}
Suppose $c$ is odd and $C$ is an $\alpha$-cut. 
Then the Contraction Algorithm with parameter $\alpha$ selects $C$ with probability $\Omega(n^{-2 \alpha \frac{c}{c+1}})$.
\end{theorem}
\begin{proof}
First, suppose $n \geq \alpha c^2$ (which is the most interesting and difficult case). We then apply Corollary~\ref{corr-gslashl} with $L = C, j = \lceil \alpha c^2 \rceil$:
\begin{align*}
\bE_{\mathcal{CP}(G,C)}[ S_{\lceil 2 \alpha \rceil}] & \leq \bE_{\mathcal{CP}(G \slash C, C)}[ S_{\lceil \alpha c^2 \rceil}] + 2 \log(\frac{\alpha c + \lceil \alpha c^2 \rceil}{2 \alpha})  \\
 &\leq \bE_{\mathcal{CA}(G \slash C)}[S_{\lceil \alpha c^2 \rceil}] + 2 \log(\frac{c(c+1)}{2}) + 1
\end{align*}
where here we note that $\mathcal{CP}(G \slash C, C) = \mathcal{CA}(G \slash C)$.

Let $r$ be the number of vertices in $G \slash C$. Applying Lemma~\ref{Acrudeoddlemma}:
\begin{align*}
\bE_{\mathcal{CA}(G \slash C)}[ S_{\lceil \alpha c^2 \rceil}] &\leq \max(0, \frac{2 c}{c+1} \log ( \frac{r}{\lceil \alpha c^2 \rceil} )) + O(1) \leq \frac{2 c}{c+1} \log ( \frac{n}{\alpha c^2 } ) + O(1)
\end{align*}

So
\begin{align*}
\bE_{\mathcal{CP}(G,C)}[ S_{\lceil 2 \alpha \rceil}] &\leq  \frac{2 c}{c+1} \log ( \frac{n}{\alpha c^2} ) + 2 \log(\frac{c(c+1)}{2}) + O(1)\\
 &=  \frac{2 c}{c+1} \log n - \Omega(\log \alpha)  - \frac{4 c}{c+1} \log c + 2 \log \bigl( c(c+1) \bigr) + O(1)\\
 &\leq  \frac{2 c}{c+1} \log n - \Omega(\log \alpha) + O(1)
\end{align*}

By Lemma~\ref{Aslemma2}, the probability of selecting $C$ is bounded by:
\begin{align*}
\bP(\text{Contraction Algorithm selects $C$}) & \geq \exp( -O(\alpha) - \alpha \bE_{\mathcal{CP}(G, C)}[S_{\lceil 2 \alpha \rceil}] ) \\
&\geq \exp(-O(\alpha) - \frac{2 c}{c+1} \alpha \log n + \Omega(\alpha \log \alpha)  ) \\
&\geq \Omega( n^{ -\frac{2 c}{c+1} \alpha} ) 
\end{align*}

To finish the proof, we need to deal with the cases in which $n$ is small. When $n \leq 2 \alpha$, then observe that cut $C$ is selected with probability exactly $2^{1-n} \geq \Omega(n^{-2 \alpha \frac{c}{c+1}})$. Next suppose that $2 \alpha < n < \alpha c^2$. Then we apply Corollary~\ref{Asimplecorr}, which does not take account of the fact that $c$ is odd:
\begin{align*}
\bP( \text{Select $C$} ) &\geq e^{-O(\alpha)} ( \frac{n}{\alpha} \bigr)^{-2 \alpha} \geq n^{-\alpha \frac{2 c}{c+1}}  e^{-O(\alpha)} n^{-\frac{2 \alpha}{c+1}} \alpha^{2 \alpha} \geq n^{-\alpha \frac{2 c}{c+1}}  e^{-O(\alpha)} (\alpha c^2)^{-\frac{2 \alpha}{c+1}} \alpha^{2 \alpha} \\
&\geq n^{-\alpha \frac{2 c}{c+1}}  e^{-O(\alpha)} (c^{1/(c+1)})^{-4 \alpha} \alpha^{\alpha} \qquad \text{as $c \geq 1$} \\
&\geq n^{-\alpha \frac{2 c}{c+1}}  e^{-O(\alpha)} \alpha^{\alpha} \qquad \text{as $x^{-1/(x+1)} \geq \Omega(1)$ for $x \geq 1$} \\
&\geq \Omega(n^{-\alpha \frac{2 c}{c+1}})
\end{align*}
\end{proof}

The following example shows that we cannot achieve any probability of the form $n^{-\alpha x}$ where $x$ is a constant with $x < \frac{2 c}{c+1}$:
\begin{observation}
Let $c$ be odd, and let $\alpha = k \frac{c+1}{2 c}$ with $k$ an integer and $\alpha \leq 2 n$. Then there is a graph $G$ whose min-cut has weight $c$ and which has $\Omega \Bigl( (\frac{n}{2 \alpha})^{\frac{2 c}{c+1} \alpha} \Bigr)$ distinct $\alpha$-cuts.
\end{observation}
\begin{proof}
Consider a cycle graph consisting of bundles of $(c+1)/2$ edges between each adjacent pair of vertices, except for one edge bundle of $(c-1)/2$ edges. This graph has minimum cut $c$. 
Suppose we select any $k$ edge-bundles which use the $(c+1)/2$ edges to fail. This gives a cut of weight $k (c+1)/2$, which is $k \frac{c+1}{2 c}$ times the minimum cut as indicated.
There are $\binom{n-1}{k}$ such choices, so the total number of such $\alpha$-cuts is at least
$$
\binom{n-1}{k} = \binom{n-1}{\frac{2 \alpha c}{c+1}} = \Omega\Bigl(  (\frac{n}{2 \alpha})^{\frac{2 c}{c+1} \alpha} \Bigr)
$$
\end{proof}

\section{The effects of graph reliability on the Contraction Algorithm}
\label{Azsec}
We now show how the graph reliability affects the behavior of the Contraction Algorithm. Our goal in this section is to obtain a quantitative form of the following simple relationship: a graph which has high reliability will have fewer $\alpha$-cuts than the cycle graph, for any value of $\alpha \geq 1$. The key to this is to draw a connection between $\bar Z$ (the expected number of failed cuts) and the behavior of the Contraction Algorithm. If $\bar Z$ is large (say as large as $n^2 p^c$), then the graph $G$ could be essentially like the cycle graph, in which case Karger's analysis of the Contraction Algorithm would be tight. However, when $\bar Z$ is small, then the graphs encountered during the Contraction Algorithm have more edges which implies that any target cut $C$ is retained with higher probability.

\textbf{In this section, we will assume $U(p) < n^{-K}$ for some constant $K$; thus the conclusions of Proposition~\ref{k-prop} hold.}

Throughout this section, we recall that
$$
p^c = n^{-2-\delta} \qquad \bar Z = n^{-2-\delta+\beta}
$$
and that $\delta \geq \delta_0 > 0$, where $\delta_0$ is an arbitrarily small constant.

We begin with a warm-up exercise, which will illustrate some of the ideas in the proof. We then discuss how this simple proof falls short and how to improve it, at the cost of greater complexity.

\begin{proposition}
\label{simplecaprop}
Define
$$
h(x) = \begin{cases}
2(2+\delta)( \log(3 - \beta + \delta) - \log(2 - \beta + \delta + x) ) & \text{if $x \geq \beta$} \\
2(2+\delta)( \log(3 - \beta + \delta) - \log(2 + \delta) ) + 2 (\beta - x) & \text{if $x < \beta \leq 1$} \\
2 (1 - x) & \text{if $\beta \geq 1$}
\end{cases}
$$
Then for any integer $i \in [2,n]$ we have
$$
S_{i} \leq h \bigl( \frac{\log i}{\log n} \bigr) \log n
$$

Here, $x$ should be thought of as a parameter which measures how far along we are in the contraction process: $x = 1$ is the beginning of the process, while $x = o(1)$ is the end.
\end{proposition}

Before we give the proof, we give a more intuitive description of the somewhat mysterious function $h(x)$. Suppose that we run the Contraction Algorithm starting at the original graph, but terminate it after reaching a subgraph with $n^x$ vertices (instead of running it all the way to a graph with $2 \alpha$ vertices). The probability that the cut $C$ survives to $G_{n^x}$ is roughly $n^{-\alpha h(x)}$. So when $x = 1$, then $h(x) = 0$ and this makes sense because the cut $C$ survives with probability one. If we assume that $\alpha$ is a small constant, then we terminate the algorithm at $x \approx 0$ and so the cut $C$ survives with probability $n^{-\alpha h(0)}$; this is essentially the probability that the cut $C$ is actually selected by the Contraction Algorithm.

\begin{proof}
Let $x = \frac{\log i}{\log n}$.  For $\beta \geq 1$, this is simply Proposition~\ref{Aeqn:simple-s-bound}.

Let us next consider the case $\beta \leq x \leq 1$. Consider the subgraph $G_r$ during during the Contraction Algorithm for $r > i$. As $r \geq 3$, the neighborhood of each vertex $v$ determines a distinct cut of $G_r$ with weight $d_v$, where $d_v$ is the degree of $v$. As all the cuts of $G_r$ are cuts of $G$ as well, we must have
$$
\sum_v p^{d_v} \leq \sum_{\text{$C$ a cut of $G$}} p^{|C|} = \bar Z
$$

Now note that $r p^{\frac{\sum_v d_v}{r}}\leq \sum_v p^{d_v}$, hence  $r p^{2 M_r/r} \leq \bar Z$, which in turn implies that
$$
M_r \geq \frac{r \log(\bar Z/r) }{2 \log p}
$$

Hence 
\begin{align*}
S_i &\leq \sum_{r =  i+1}^n \frac{2 c \log p}{r \log(\bar Z/r)} \leq \int_{r = i}^n \frac{2 c \log p}{r \log(\bar Z/r)} \  dr = -2 c \log p (\log \log (r/\bar Z)) \Bigr|_{r=n^x}^n \\
&= 2 (2 + \delta) \log n (\log (3 - \beta + \delta) - \log(2 - \beta + \delta + x)) \\
&= h(x) \log n
\end{align*}

A similar calculation applies in the case $x \leq \beta$.
\end{proof}

This simple proof demonstrates the effect of the graph reliability on the behavior of the Contraction Algorithm. In fact, this Proposition~\ref{simplecaprop} is already sufficient to obtain a substantially improved run-time compared to Karger's algorithm. We will also need this Proposition~\ref{simplecaprop} for technical reasons later in this paper.

However, Proposition~\ref{simplecaprop} is not tight, because it assumes that the value of $\bar Z$ does not change during the execution of the Contraction Algorithm. But in fact, most of the cuts in the graph are being contracted away, so $\bar Z$ should be decreasing rapidly. 

We will track this behavior in terms of a parameter which is slightly more general than $\bar Z$. Suppose we are given a target $\alpha$-cut $C$. For any contraction subgraph $H$, we define the associated graph parameter 
\begin{equation}
\label{aeqn}
A_{\gamma}^{H} = \negthickspace \negthickspace  \sum_{\text{cuts $C'$ of $H$}} \negthickspace \negthickspace e^{-\gamma (|C' - C|)/c}
\end{equation}

This is the expected number of failed cuts of $H$, if all the edges in $C$ fail and all the edges outside $C$ fail independently with probability $e^{-\gamma/c}$. Be aware that $C$ may not be a cut of the graph $H$ (if some edges of $C$ were contracted); in (\ref{aeqn}), one should interpret $C$ as a collection of edges which survive in $H$.

Proposition~\ref{step-a-prop} shows how the parameter $A$ changes over the course of a Contraction Process.

\begin{proposition}
\label{step-a-prop}
Suppose that $C$ is a cut of $G$, and $L$ is a subset of the edges of $C$. Suppose that $H$ is a contraction-subgraph of $G$ with $m$ edges. Let $H'$ denote the graph obtained from $H$ by one step of the Contraction Process for $L$ (that is, selecting an edge of $H - L$ uniformly at random and contracting it.) Then
$$
\bE[A_{\gamma - c/m}^{H'}] \leq A_{\gamma}^H
$$
\end{proposition}
\begin{proof}
Consider some cut $C'$. In a single iteration of the contraction process, we select an edge of $G - L$ uniformly at random. There are $|E(H) - L|$ such edges, hence we select an edge of $C' - L$ with probability $\frac{|C' - L|}{|E(H) - L|}$. So the probability of retaining $C'$ is
$$
\bP(\text{retain $C'$}) \leq 1 - \frac{|C' - L|}{|E(H) - L|} \leq 1 - \frac{|C' - C|}{m} \leq \exp(-\frac{|C'-C|}{m})
$$
The contribution of $C'$ to $A_{\gamma-c/m}^{H'}$ is $\exp( -(\gamma-c/m) \frac{|C' - C|}{c} )$ if $C'$ is retained, and is zero otherwise. Summing over all such $C'$:
\begin{align*}
\bE[A_{\gamma-c/m}^{H'}] &\leq \sum_{C'} \exp( -(\gamma-c/m) \frac{|C' - C|}{c} ) \exp(-\frac{|C' - C|}{m}) \\
&= \sum_{C'} \exp( -\gamma \frac{|C' - C|}{c} ) = A_{\gamma}^H
\end{align*}
\end{proof}

Our strategy here is to guess a formula for $\bE[S_i]$ in terms of various parameters, and then prove it true by an induction. The induction proof is not informative on its own, so we give a heuristic derivation of our formula. This amounts to estimating $M_r$ for the various stages of the Contraction Process for $C$. In this discussion, we ignore any considerations of the fact that $S_i$ is a random variable, and we treat it as essentially deterministic.

Suppose that the cuts of $G$, other than the target cut $C$, have negligible overlap with $C$. In this case, $A^G_{-\log(p^c)} = \sum_{C'} p^{|C' - C|} \approx \sum_{C'} p^{C'} = \bar Z$. Also, applying Proposition~\ref{step-a-prop} multiple times, and ignoring all distinctions between expected value and actual values for the variables and all variable dependencies, we have that for the graph $G_i$:
$$
A^{G_i}_{-\log(p^c) - S_i} = A^{G_i}_{-(\log p^c) -c/M_{i+1} - c/M_{i+2} + \dots - c/M_n} \leq A^{G}_{-(\log p^c)} = \bar Z
$$

Now suppose that in $G_i$, each vertex has degree $d$. Since the neighborhood of each vertex of $G_i$ defines a distinct cut, and these cuts are all assumed to be nearly disjoint from $C$, 
$$
\sum_v \exp \bigl( (c \log p + S_i) d/c \bigr) \leq A^{G_i}_{-\log(p^c) + S_i} = \bar Z
$$

Which implies $$
d \geq \frac{c \log (\bar Z/i)}{S_i + c \log p} 
$$
and hence
$$
M_i = i d/2 \geq \frac{i c/2 \log (\bar Z/i)}{S_i + c \log p}
$$

Thus we have a recurrence relation in $S_i$:
$$
S_{i-1} = S_i + c/M_i \leq S_i + \frac{2 (c \log p+S_i)}{i \log(\bar Z/i)}
$$

We relax this to a differential equation
\begin{align*}
\frac{d S_i}{di} &= -\frac{2 (S_i + c \log p)}{i \log (Z/\bar i)} \\
S_n &= 0
\end{align*}
which can be solved in closed form to obtain
$$
S_i =  (-\log(p^c)) \Bigl( 1 - \frac{\log^2(\bar Z/i)}{\log^2(\bar Z/n)}  \Bigr)
$$

This derivation is completely non-rigorous, as it makes a number of unwarranted independence and monotonicity assumptions. However, as we will see, all of these assumptions turn out to be the worst-case behavior. Hence the above bound is essentially correct.

As in Section~\ref{Asec:oddc}, we will use this bound as the basis of an induction argument. However, the situation is more complicated because of the way we must transform the contraction process for $C$ into an unconditioned contraction algorithm. In Section~\ref{Asec:oddc}, we handled this by applying Lemma~\ref{Alemma:stoch-dom} with $L = C$. This worked because $C$ was small compared to $n$. We cannot assume this here; instead, we will define $L$ to be a small random subset of the edges of $C$. This will remove the influence of $C$ from cuts $C'$ which heavily overlap $C$. We will show that in fact the worst-case behavior of the contraction process comes about when all other cuts of $G$ have negligible overlap with $C$; in that case, we can essentially ignore $C$.

\begin{lemma}
\label{Acontractlemma2}
There is a subset of the edges $L \subseteq C$ with the following properties:
\begin{enumerate}
\item $|L| \leq \lceil \alpha (2+\delta) \log  n \rceil$
\item $A_{-\log(p^c)}^{G \slash L} \leq \bar Z$
\end{enumerate}

(We can interpret this as follows. The parameter $A_{-\log(p^c)}^{G \slash L}$ counts the expected number of failed cuts, given that all edges in $L$ are \emph{retained} and all edges in $C - L$ are \emph{removed} and all other edges fail \emph{randomly with probability $p$})
\end{lemma}
\begin{proof}
Choose $L$ to be a random subset of $C$ of size exactly $x = \lceil \alpha(2+\delta) \log n\rceil $. (If this number is larger than the number of edges in $C$, we simply set $L = C$, in which case the lemma is trivially true.)  Now consider the contribution of some cut $C'$ to $A_{-\log (p^c)}^{G \slash L}$. If $L$ contains some edge in $C \cap C'$, then $C'$ contributes zero; else it contributes $p^{|C' - C|}$. 

So, supposing that $|C \cap C'| = s c$, we have:
\begin{align*}
\bE[\text{Contribution of cut $C'$ to $A_{-\log(p^c)}^{G \slash L}$}] &= p^{|C'| - s c} \bP( C' \cap L = \emptyset ) \\
&\leq p^{|C'|} p^{-s c} \frac{\binom{\alpha c - s c}{x}}{\binom{\alpha c}{x}} \\
&= p^{|C'|} p^{-s c} \frac{ (\alpha c - s c) (\alpha c - s c - 1) \dots (\alpha c - s c - (x-1))}{ (\alpha c) (\alpha c - 1) \dots (\alpha c - (x-1)) } \\
&\leq p^{|C'|} p^{-s c} (\frac{\alpha c - s c}{\alpha c})^{x} = p^{|C'|} n^{(2+\delta) s} (1 - \frac{s}{\alpha})^{x} \\
&\leq p^{|C'|} n^{s (2+\delta)} e^{-\frac{s}{\alpha} \alpha(2+\delta) \log n}  = p^{|C'|}
\end{align*}

Thus, summing over all $C'$, we have $\bE[A_{-\log(p^c)}^{G \slash L}] \leq \sum_{C'} p^{|C'|} = \bar Z$. In particular, there exists some subset $L$ with $A_{-\log(p^c)}^{G \slash L} \leq \bar Z$.
\end{proof}

We are now ready to introduce the main induction argument:
\begin{theorem}
\label{Amainrelthm}
Define the function
$$
f(i,r,a,\gamma) = \begin{cases}
\gamma \Bigl( 1 - \frac{\log^2 (a/i)}{\log^2 (a/r)} \Bigr) & \text{if $a \in (0,1]$} \\
2 \log(r/i) & \text{if $a > 1$}
\end{cases}
$$

Fix a graph $G$ and a target cut $C$ of $G$, and let $H$ be a contraction-subgraph of $G$ with $r$ vertices. Let $\gamma$ be a real number in the range $[2 \log r, \infty)$ and let $i,r$ be integers with $100 \leq i \leq r \leq n$. Then
$$
\bE_{\mathcal{CP}(H,C)}[ S_i] \leq f(i,r,A^{H}_{\gamma},\gamma)
$$

We emphasize here that $C$ is not necessarily a cut of $H$; in the expression $\mathcal{CP}(H,C)$ and $A^{H}$, we view $C$ simply as a subset of the edges of $H$.
\end{theorem}
\begin{proof}
We defer the proof of some technical properties of $f$ to Appendix~\ref{mainrel-appendix}, for sake of clarity.

We induct on $r$.  When $ r= i$, we have $ S_i = f(i,r,a,\gamma) = 0$.

Now suppose $r \geq i+1$, and $H$ has $m$ edges with $A^{H}_{\gamma} = a$. We may assume that $a \leq 1$, as otherwise this follows immediately from Proposition~\ref{Aeqn:simple-s-bound}.

So $\bE_{\mathcal{CP}(H,C)}[ S_i] = c/m + \bE_{H'} \bE_{\mathcal{CP}(H',C)}[ S_i]$. We have broken the expectation into two components: first, we select an edge $e$ of $H - C$ to contract, leading to the graph $H' = H \slash e$; second, we continue the contraction process on the subgraph $H'$, which has $r-1$ vertices.

Note that $m \geq r c/2$, so $\gamma - c/m \geq 2 \log r - 2/r \geq 2 \log (r-1)$. So the induction hypothesis applies to the graph $H'$ with $\gamma' = \gamma - c/m$ giving
$$
\bE_{\mathcal{CP}(H,C)} \leq c/m + \bE_{H'} \bigl[   f(i,r-1,A^{H'}_{\gamma - c/m}, \gamma  - c/m) \bigr].
$$

By Proposition~\ref{step-a-prop}, we have $\bE_{H'} [A_{\gamma - c/m}^{H'}] \leq a$. So by Proposition~\ref{Aprop2x1},
$$
\bE_{\mathcal{CP}(H,C)} \leq c/m + f(i,r-1,a, \gamma  - c/m).
$$

We will now bound the number of edges $m$ of the graph $H$. For each vertex $v \in H$, let $d_v$ be the number of incident edges \emph{not including the edges of $C$}. Each vertex then corresponds to a cut which has $d_v$ edges outside $C$, and as $r \geq 3$ these cuts are all distinct. So
$$
\sum_{v \in H} \exp(-\gamma d_j/c)\leq A_{\gamma}^{H} = a.
$$

By the arithmetic mean-geometric mean (AM-GM) inequality,
\begin{align*}
\sum_{v \in H} \exp(-\gamma d_j/c) &\geq r \Bigl( \prod_v (\exp(-\gamma d_j/c)) \Bigr)^{1/r} = r \exp( \frac{ -\gamma \sum_v d_j}{r c}) \\
&\geq r \exp( \frac{ -\gamma (m/2)}{r c}) \qquad \text{as $m \geq \sum_v d_v/2$ and $\gamma > 0$}
\end{align*}

This implies that 
$$
m \geq \frac{r c \log(a/r)}{-2 \gamma}
$$

By Proposition~\ref{Aprop2x2}, the quantity $ c/m + f(i,r-1,a,\gamma - c/m) $ is decreasing in $m$. Hence
\begin{align*}
\bE_{\mathcal{CP}(H,C)}[ S_i] &\leq \frac{-2 \gamma}{r \log(a/r)} + f(i,r-1,a,\gamma + \frac{2 \gamma}{r \log(a/r)} ) \\
&\leq f(i,r,a,\gamma) \qquad \text{by Proposition~\ref{Aprop2x3}}
\end{align*}
\end{proof}

We use this to estimate the probability $\bE_{\mathcal{CP}(G,C)}[S_i]$ for a target cut $C$.

\begin{theorem}
\label{Acutboundthm}
Define the function $\bar h: [0,1] \rightarrow \mathbf R_{+}$ by
$$
\bar h(x) = \frac{(\delta+2)(1 - x)(5 - 2 \beta + 2 \delta + x)}{(3 - \beta + \delta)^2}
$$

Let $C$ be an $\alpha$-cut of $G$, and let $i$ be an integer in the range $i = \lceil 2 \alpha \rceil, \dots, n$. Then
$$
\bE_{\mathcal{CP}(G,C)}[S_i] \leq \bar h \bigl( \frac{\log i}{\log n} \bigr) \log n + O(\log  \log  n)
$$

(The function $\bar h$ has the same interpretation as $h$, however $\bar h(x)$ is (typically) smaller than $h(x)$ and so it provides a tighter bound on the probability of the cut $C$ surviving (partially) through the Contraction Algorithm.)

\end{theorem}

\begin{proof}
Let $x = \frac{\log i}{\log n}$, so $x \in [0,1]$.

If $\delta \geq \log n$, the Proposition~\ref{Aeqn:simple-s-bound} gives $S_i \leq 2 \log(n/i) = (2 - 2 x) \log n$. One may verify that $(2 - 2x) \leq \bar h(x) + 3 /\delta$. This implies that $S_i \leq (\bar h(x) + \frac{3}{\log n}) \log n \leq \bar h(x) \log n + O(1)$ and we are done.\footnote{Observe that $\bar h(x)$ is a rational function of the parameters $\beta, \delta, x$, the coefficients of which are rational numbers.  Thus, the statement that $(2 - 2 x) \leq \bar h(x) + 3 /\delta$ for $\beta, \delta, x$ in the appropriate range is equivalent to a first-order sentence in the theory of real-closed fields. A classical result of Tarski is that the theory of real-closed fields is decidable; many symbolic algebra packages have implemented practical algorithms for this.  We used the Mathematica function \emph{Reduce} which implements these algorithms to check that this inequality holds. For the remainder of the paper, whenever we refer to mechanically checking an inequality, we mean that employ the Mathematica software package which proves that it is a validity.} 

So we may assume that $\delta < \log n$ for the remainder.

Let $L \subseteq C$ be as given by Lemma~\ref{Acontractlemma2}, so that $A^{G \slash L}_{-\log(p^c)} \leq \bar Z$ and $L \leq \lceil \alpha (2 + \delta) \log n \rceil \leq O(\alpha \log^2 n)$. Let $H = G \slash L$, and suppose that $H$ has $r \leq n$ vertices. There are now two cases depending on the size of $r, i$.

In the first (and most important) case, suppose that $r \geq i + 100$. Then we apply Corollary~\ref{corr-gslashl} with $j = i + 100$:
\begin{align*}
\bE_{\mathcal{CP}(G,C)}[ S_{i}] & \leq \bE_{\mathcal{CP}(H, C)}[ S_j] + 2 \log ( \frac{|L| + i + 100}{i}) \\
&\leq f(j, r, A^{H}_{-\log(p^c)}, -\log(p^c)) + 2 \log ( \frac{|L| + i + 100}{i}) \qquad \text{by Theorem~\ref{Amainrelthm}, as $r \geq j \geq 100$} \\
&\leq f(j, r, A^{H}_{-\log(p^c)}, -\log(p^c)) + O(\log \log n) \qquad \text{as $i \geq 2 \alpha, |L| \leq \lceil \alpha (2 + \delta) \log n \rceil \leq O(\alpha \log^2 n)$}
\end{align*}

Simple inspection shows that $f(k,\ell,a, \gamma)$ is a decreasing function of $k$ and an increasing function of $\ell$. Thus $f(j,r,A^H_{-\log(p^c)},-\log(p^c)) \leq f(i,n, A^H_{-\log(p^c)},-\log(p^c))$. Furthermore, using Proposition~\ref{Aprop2x1} and the fact that $A^H_{-\log(p^c)} \leq \bar Z \leq 1$, we have that $ f(i,n, A^H_{-\log(p^c)},-\log(p^c)) \leq f(i,n,\bar Z, -\log(p^c))$. So 
$$
\bE_{\mathcal{CP}(G,C)}[ S_{i}] \leq f(i, n, \bar Z, -\log(p^c)) + O(\log \log n)
$$

Now substitute $i = n^x, \bar Z = n^{-2-\delta+\beta}, p^c = n^{-2-\delta}$, and observe that $f(i, n, \bar Z, -\log(p^c)) = \bar h(x) \log n$.

The next case is when $r < i + 100$. We apply Corollary~\ref{corr-gslashl} with $j = i$:
\begin{align*}
\bE_{\mathcal{CP}(G,C)}[S_{i}] & \leq \bE_{\mathcal{CP}(H, C)}[ S_i ] + 2 \log ( \frac{|L| + i}{i} ) \leq 2 \log(r/i) + 2 \log ( 1 + |L|/i ) \\
&\leq 2 \log(\frac{i+100}{i}) + O(\log \log n) \qquad \text{as $i \geq 2 \alpha, |L| \leq \lceil \alpha (2 + \delta) \log n \rceil \leq O(\alpha \log^2 n)$} \\
&\leq O(\log \log n) \leq \bar h(x) \log n+ O(\log \log n)
\end{align*}
\end{proof}

Our next result introduces an important parameter $\bar h(0)$. Essentially, $\bar h(0)$ measures the rate of decay of the number of $\alpha$-cuts. If we make no assumptions about the graph reliability, then the number of $\alpha$-cuts may be as large as $n^{2 \alpha}$. When we take the reliability of $G$ into account, we show that instead the number of $\alpha$-cuts is like $n^{\bar h(0) \alpha}$. Crucially, barring some exceptional cases, we have $\bar h(0) < 2$.

\begin{corollary}
\label{Acutboundthm2}
For $\alpha \in [1,n/2]$ we have 
$$
\bE_{\mathcal{CP}(G,C)} [ S_{\lceil 2 \alpha \rceil}] \leq \bar h(0) \log  n + O(\log  \log  n) - \Omega(\log \alpha)
$$
\end{corollary}
\begin{proof}
We have
\begin{align*}
\bE[S_{\lceil 2 \alpha \rceil}] &\leq \bar h(\frac{\log \lceil 2 \alpha \rceil}{\log n}) \log n + O(\log \log n) \qquad \text{by Theorem~\ref{Acutboundthm}} \\
&= \bar h(0) \log n + \log n \int_{x=0}^{\frac{\log \lceil 2 \alpha \rceil}{\log n}} \bar h'(x) dx + O(\log \log n) \\
&\leq \bar h(0) \log n - \log n \times \frac{\log \lceil 2 \alpha \rceil}{\log n} \times \Omega(1) + O(\log \log n) \qquad \text{by Proposition~\ref{hbar-deriv-prop}} \\
&= \bar h(0) \log n - \Omega(\log \alpha) + O(\log \log n)
\end{align*}

\end{proof}

\begin{corollary}
\label{Acutboundthm3}
The Contraction Algorithm with parameter $\alpha$ outputs any given $\alpha$-cut with probability at least $n^{-\alpha \bar h(0)} \log ^{-O(\alpha)} n$. In particular, the number of such cuts is at most $n^{\alpha \bar h(0)} \log ^{O(\alpha)} n$.
\end{corollary}
\begin{proof}
If $n \leq 2 \alpha$, then $C$ is selected with probability exactly $2^{1-n} \geq \Omega(n^{-\alpha})$; by Proposition~\ref{hbar-min}, $\bar h(0) \geq 10/9$ so this is  $\geq n^{-\alpha \bar h(0)}$.

Otherwise, by Lemma~\ref{Aslemma2}, the probability of selecting any given $\alpha$-cut $C$ is at least 
$$
\bP(\text{select $C$}) \geq \exp(-O(\alpha) -\alpha \bE_{\mathcal{CP}(G,C)}[ S_{\lceil 2 \alpha \rceil}]).
$$
By Corollary~\ref{Acutboundthm2} this is at least $\exp(-O(\alpha)) \exp(\Omega(\alpha \log  \alpha)) n^{-\alpha \bar h(0)} \log ^{-O(\alpha)} n \geq n^{-\alpha \bar h(0)} \log^{-O(\alpha) } n$.
\end{proof}

\section{Bounds on $\alpha^*$}
\label{Asec:smallcutunreliability}
\textbf{In this section, we will assume $U(p) < n^{-K}$ for some constant $K$; thus the conclusions of Proposition~\ref{k-prop} hold.}

Recall that the approach of \cite{karger} is to show that most unreliability is due to the failure of a small cut. We can use our bound on the number of cuts to estimate this more precisely than Theorem~\ref{Akargerthm1}.

\begin{proposition}
\label{Afaillemma}
Define $U_{\alpha}(p)$ to be the probability that some cut of weight $\leq \alpha c$ fails. Then for all $\alpha \geq 1$ we have the bounds
\begin{align*}
U(p) - U_{\alpha}(p) &\leq n^{\alpha (\bar h(0) - 2 - \delta)} \log ^{O(\alpha)} n \\
U(p) - U_{\alpha}(p) &\leq O(n^{-\alpha \delta})
\end{align*}

\end{proposition}
\begin{proof}
We have
\begin{align*}
U(p) - U_{\alpha}(p) &= \bP( \text{Some cut of weight $> \alpha c$ fails and no cut of weight $\leq \alpha c$ fails}) \\
&\leq \bP( \text{Some cut of weight $> \alpha c$ fails})
\end{align*}

By Corollary~\ref{Acutfailcorr1}, this is at most $O(n^{-\alpha \delta})$, thus proving the second bound. If $\bar h(0) \geq 2$, then $n^{-\alpha \delta} \leq n^{\alpha(\bar h(0) - 2 - \delta)}$ and we automatically prove the first bound as well.

So we may suppose that $\bar h(0) \geq 2$. Then we estimate this by the union bound
$$
\bP( \text{Some cut of weight $> \alpha c$ fails}) \leq \sum_{|C| > \alpha c} p^{|C|}
$$

Proposition~\ref{Acutboundthm3} states that the number of cuts of weight $x c$ is at most $F(x) = n^{\bar h(0) x} \log^{O(x)} n$. Hence we may apply
Proposition~\ref{Aexpsumlemma}: this sum is bounded by
\begin{align*}
\sum_{C: |C| > \alpha c} p^{|C|} &\leq F(\alpha) p^{\alpha c} + \int_{x = \alpha}^{\infty} p^{x c} F'(x) dx \\
&= n^{h(0) \alpha} n^{\alpha(-2-\delta)} \log^{O(\alpha)} n + \int_{x = \alpha}^{\infty} n^{x(-2-\delta)} n^{\bar h(0) x} (\log^{O(x)} n) (\bar h(0) \log n + O(\log \log n)) dx \\
&\leq n^{\alpha (\bar h(0) - 2 - \delta)} \log^{O(\alpha)} n \qquad \text{for $\bar h(0) \leq 2$ and $\delta \geq \delta_0 > 0$} \\
\end{align*}

\end{proof}

As in Karger \cite{karger}, we estimate $U(p)$ by examining only the smallest cuts. We define $\alpha^*$ to be the minimal value of $\alpha$ such that
$$
U(p) (1 - \epsilon) \leq U_{\alpha^*}(p) \leq U(p)
$$

The parameter $\alpha^*$ plays a crucial role in the analysis. We can estimate $U(p)$ to the desired relative error by estimating $U_{\alpha^*}(p)$, which we can do by enumerating all the $\alpha^*$-cuts. This is a relatively small collection, which we can explicitly collect and list. 

For the rest of the paper, we define
\begin{equation}
\label{Aeqn:rho}
\rho = \log  (1/\epsilon)/\log  n.
\end{equation}
Although in general we allow $\rho$ to increase arbitrarily, we find that these bounds can become confusing because of the interplay between the asymptotic growth of $n$ and $\epsilon$. One can get most of the intuition behind these results by restricting attention to the case $\rho = O(1)$.

Using our improved bounds on the Contraction Algorithm, we tighten Theorem~\ref{Akargerthm1}.
\begin{lemma}
\label{Aalphastarlemma1}
We have
$$
\alpha^* \leq \frac{2 - \beta + \delta + \rho}{\delta} + O(\frac{1}{\log n})
$$

If $\beta \leq 3/2$, then we have
$$
\alpha^* \leq \frac{(3 - \beta + \delta)^2(2 - \beta + \delta + \rho)}{(2 - \beta + \delta)^2(2 + \delta)} + O(\frac{(1+\rho) \log \log n}{\log n})
$$

\end{lemma}

\begin{proof}

The first bound is easier. As shown in Proposition~\ref{Afaillemma}, the absolute error introduced by ignoring cuts of weight $\geq \alpha c$ is $U(p) - U_{\alpha} (p) = O(n^{-\alpha \delta})$. So the relative error is at most
\begin{align*}
\text{Rel Err} &=  O \bigl( \frac{n^{-\alpha \delta)}}{U(p)} \bigr) \\
&=  O \bigl( \frac{n^{-\alpha \delta}}{\bar Z} \bigr) \qquad \text{by Proposition~\ref{Aestimatelemma}} \\
&\leq  C n^{-\alpha \delta} n^{2 + \delta - \beta} \qquad \text{for some constant $C$}
\end{align*}

Elementary algebra shows that this is $\leq \epsilon$ for $\alpha \geq \frac{\log C - \log \epsilon + (2 - \beta + \delta) \log n}{\delta \log n} = \frac{2 - \beta + \delta + \rho}{\delta} + O(\frac{1}{\log n})$.
 
Next let us consider the case when $\beta \leq 3/2$. By Proposition~\ref{hbar-min} this implies $\bar h(0) \leq 2$. As shown in Proposition~\ref{Afaillemma}, the absolute error introduced by ignoring cuts of weight $\geq \alpha c$ is $U(p) - U_{\alpha} (p) = n^{\alpha (\bar h(0) - 2 - \delta)} \log ^{O(\alpha)} n$. So the relative error is at most
\begin{align*}
\text{Rel Err} &=  O \bigl( \frac{n^{\alpha (\bar h(0) - 2 - \delta)} \log ^{O(\alpha)} n}{U(p)} \bigr) \\
&=  O \bigl( \frac{n^{\alpha (\bar h(0)- 2 - \delta)  } \log ^{O(\alpha)} n}{\bar Z} \bigr) \qquad \text{by Proposition~\ref{Aestimatelemma}} \\
&=  n^{\alpha (\bar h(0) - 2 - \delta)} n^{2 + \delta - \beta} \log ^{O(\alpha)} n
\end{align*}

Set $\alpha = \frac{\rho + 2 - \beta + \delta}{2 + \delta - \bar h(0)} + \frac{\phi(1+\rho) \log \log n}{\log n}$, where $\phi$ is some constant be specified. Note that $\alpha \leq O(1+\rho)$, as $\delta > \delta_0 > 0$ and $\bar h(0) \leq 2$. 

\begin{align*}
\text{Rel Err} &\leq n^{-(\rho+2-\beta+\delta)} n^{\frac{\phi (1 + \rho) \log \log n (\bar h(0) - 2 - \delta)}{\log n}} (\log n)^{O(\alpha)} \\
&\leq n^{ -\rho - \Omega(\frac{\phi (1 + \rho) \log \log n}{\log n})} (\log n)^{O(1+\rho)} \qquad \text{as $\delta \geq \beta$ and $\alpha \leq O(1+\rho)$} \\
&= \epsilon (\log n)^{-\Omega(\phi (1 + \rho))} (\log n)^{O(1+\rho)} \\
&\leq \epsilon \qquad \text{for $\phi$ a sufficiently large constant}
\end{align*}

So 
$$
\alpha^* \leq  \frac{\rho + 2 - \beta + \delta}{2 + \delta - \bar h(0)} + O(\frac{(1+\rho) \log \log n}{\log n}) =  \frac{(3 - \beta + \delta)^2(2 - \beta + \delta + \rho)}{(2 - \beta + \delta)^2(2 + \delta)} + O(\frac{(1+\rho) \log \log n}{\log n})
$$
in this case.

\end{proof}

If we only want a very crude estimate for $\alpha^*$ we can use the following, which holds irrespective of the precise value of $K, \beta, \delta$:
\begin{corollary}
\label{amax-prop}
We have $\alpha^* \leq O(1 + \rho)$
\end{corollary}
\begin{proof}
By Proposition~\ref{k-prop}, $\delta \geq K- 2 > 0$ and $\delta \geq \beta$ for $n$ sufficiently large.

By Lemma~\ref{Aalphastarlemma1}, 
\begin{align*}
\alpha^* \leq \frac{2 - \beta + \delta + \rho}{\delta} + O(\frac{1}{\log n})
\end{align*}

The restrictions $\delta \geq \delta_0 > 0$ and $\delta \geq \beta$ ensures that this is $O(1+\rho)$.
\end{proof}

\smallskip 

\textbf{Remark:} In light of Corollary~\ref{amax-prop}, it will be convenient to define a term $\alpha^*_{\text{max}}$, which is a computable upper bound of $\alpha^*$ and which satisfies $\alpha^*_{\text{max}} \leq O(1+\rho)$. The precise value of the constant term here will not be relevant.

\smallskip

It is tempting to simply modify Karger's original algorithm, using this improved estimate for $\alpha^*$ in place of Karger's estimate. We cannot do this directly, as this estimate depends on the parameter $\beta$ which we cannot simply compute. 

One approach would be to estimate $\alpha^*$ using a worst-case estimate for $\beta$. This would give usable and tighter bound than Karger's. For example, suppose that we use Monte Carlo estimation when $U(p) \geq n^{-2.73}$ and cut-enumeration when $U(p) < n^{-2.73}$. Some simple analysis (which we omit here) shows that $\alpha^* \leq 1.87 + O(\rho)$ whenever $U(p) < n^{-2.73}$. using Karger's algorithm for the remainder, this would give us a total runtime of $n^{3.73} \epsilon^{-O(1)}$. 

This is good, but we will take another approach which will allow us to do better. To explain our approach intuitively, consider the following process: We run $n^{\alpha^* \bar h(0) } (n/\epsilon)^{o(1)} $ independent executions of the Contraction Algorithm with parameter $\alpha^*_{\text{max}}$. This generates a large collection $\mathcal A$ of cuts, of various sizes. As we show in Proposition~\ref{av-prop-1}, with high probability $\mathcal A$ contains all the $\alpha^*$-cuts. As described in Section~\ref{Asec:statistical}, we may then use the collection $\mathcal A$ to estimate $U(p)$.

\begin{proposition}
\label{av-prop-1}
The Contraction Algorithm with parameter $\alpha^*_{\text{max}}$ finds any $\alpha^*$-cut $C$ with probability at least $n^{-\alpha^* \bar h(0)} (\epsilon/n)^{o(1)}$.
\end{proposition}
\begin{proof}
By Corollary~\ref{amax-prop}, we have $\alpha^* \leq \alpha^*_{\text{max}}$. 

If $\alpha^* \geq n/2$, then the cut $C$ is found with probability $2^{1-n}$. By Proposition~\ref{hbar-min}, $n^{-\alpha^* \bar h(0)} \leq n^{-\Theta(n)}$, which is smaller than $2^{1-n}$ for $n$ sufficiently large. So we may assume that $\alpha^* \leq n/2$.

By Lemma~\ref{Aslemma2}, this cut $C$ is produced by the Contraction Algorithm with probability $\geq e^{-O(\alpha^*_{\text{max}})} e^{-\alpha^* \bE[S_{\lceil 2 \alpha^* \rceil }]}$. 

Note that $e^{-O(\alpha^*_{\text{max}})} \geq e^{-O(1 + \rho)}$. Also observe that $e^\rho = e^{-\log(1/\epsilon)/\log n} \geq \epsilon^{o(1)}$. So $e^{-O(\alpha^*_{\text{max}})} \geq \epsilon^{o(1)}$.

Also by Corollary~\ref{Acutboundthm2}, we have $\bE[S_{\lceil 2 \alpha^* \rceil}] \leq \bar h(0) \log n + O(\log \log n)$.  We have that $e^{-\alpha^* \bE[S_{\lceil 2 \alpha^* \rceil}]} \geq e^{-\alpha^* (\bar h(0) \log n + O(\log \log n) - \Omega(\log \alpha'))} \geq e^{-\alpha^* \bar h(0) \log n} \times (\log n)^{-O(\alpha^*)} \geq e^{-\alpha^* \bar h(0)} (\epsilon/n)^{o(1)}$.
\end{proof}

If we use this algorithm directly, then we must run completely separate instances of the Contraction Algorithm for each sample. Each execution of the Contraction Algorithm takes time $O(n^2)$ (to process the graph), and so the total work would be roughly $n^{\alpha^* \bar h(0) + 2}$.  Our goal now will be to reduce the work closer to $n^{\alpha^* \bar h(0)}$.

In \cite{karger-stein}, an efficient algorithm called the \emph{Recursive Contraction Algorithm} (RCA) was introduced for running multiple samples of the Contraction Algorithm simultaneously. This amortizes the work of processing the graph across the multiple iterations, effectively reducing the time for a single execution of the Contraction Algorithm from $O(n^2)$ to $O(1)$. 

In the next section, we will discuss how to combine the Recursive Contraction Algorithm with our improved analysis of the Contraction Algorithm. Unlike in \cite{karger-stein}, we will not be able to show that a single application of the RCA is as powerful as $n^2$ independent applications of the Contraction Algorithm. We will still show it is powerful enough to substantially reduce the cost of multiple applications of the Contraction Algorithm. The next section will examine the Recursive Contraction Algorithm in detail.

We note that we do not need to bound either $\alpha^*$ or $\bar h(0)$ individually, which would depend on knowing $\beta$. It suffices to give an upper bound on the \emph{product} $\alpha^* \bar h(0)$ irrespective of $\beta$.

\section{The Recursive Contraction Algorithm}
\label{Asec:rca}
The basic method of finding the $\alpha^*$-cuts is to run the Contraction Algorithm for many iterations and collect all the resulting cuts. Each iteration requires $O(n^2)$ work (to process the entire graph). As described in \cite{karger-stein}, this data-processing can be amortized across the multiple iterations. The resulting Recursive Contraction Algorithm, can enumerate all the $\alpha$-cuts in time $O(n^{2 \alpha} \log ^{2} n)$.

We will run the RCA for a large, fixed number of iterations. This will produce a large collection of cuts  $\mathcal A$, of various sizes. We will then show that the resulting collection will contain all the $\alpha^*$-cuts with high probability. Note that we do not know the precise value of $\alpha^*$ (it may depend on $\beta$), so we cannot simply discard cuts of weight $> \alpha^* c$ at this stage.

\textbf{In this section (with the exception of Section~\ref{further-app-sec}), we will assume $U(p) < n^{-K}$ for some constant $K > 2$; thus the conclusions of Proposition~\ref{k-prop} hold.}

We define the RCA with parameter $\alpha$ as follows:\footnote{This version is slightly different from that of \cite{karger-stein}; the algorithm in that paper uses 2 independent executions to reduce the graph to $n/\sqrt{2}$ vertices. The only effect of this change is to remove some quantization effects, which would complicate our proof.}
\ttfamily
\begin{enumerate}
\item[1.] Randomly contract edges of $G$ until the resulting graph $G'$ has at most $\min (n/2, \lceil 2 \alpha \rceil)$ vertices.  Run the RCA with parameter $\alpha$ on the subgraph $G'$.
\item[2.] If $G'$ has $\lceil 2 \alpha \rceil$ vertices, output a random cut of $G$.
\item[3.] Otherwise, perform 4 independent executions of the RCA on the graph $G'$.
\end{enumerate}
\rmfamily

The RCA was introduced in \cite{karger-stein}, which showed the following bounds on its performance:
\begin{theorem}[\cite{karger-stein}]
\label{rca-perf}
For any graph $G$, the RCA with parameter $\alpha'$ executes in time $O(n^2 \log n)$. For $\alpha > 1 + \Omega(1)$, it finds a target $\alpha$-cut $C$ with probability at least $\Omega(2^{-\alpha'} n^{2-2\alpha})$. 
\end{theorem}

The analysis in \cite{karger-stein} does not take into account the graph reliability. We will show a better success probability than $n^{2 - 2 \alpha}$ taking account of this information. Our goal is to show that we can enumerate all the $\alpha^*$ cuts with high probability using $n^{1+o(1)} \epsilon^{-1.99}$ independent applications of the RCA. As the running time of RCA is $n^{2+o(1)}$, the overall run time will be $\leq n^{3+o(1)} \epsilon^{-1.99}$.

We observe first that the number of $\alpha^*$-cuts is at most $n^{2 \alpha^*}$. Thus, suppose we are able to show that any $\alpha^*$-cut $C$ is found with probability at least $n^{-1-o(1)} \epsilon^{1.98}$ in a single execution of the RCA; then by a standard analysis of the Coupon Collector Problem, this implies that after $O( \log n^{2 \alpha^*} \times n^{1+o(1)} \epsilon^{-1.98})$ iterations of the RCA that all such $\alpha^*$-cuts are found, with high probability. Observing that $\log (n^{2 \alpha^*}) = (n/\epsilon)^{o(1)}$, the total number of iterations is then bounded by $n^{1+o(1)} \epsilon^{-1.98-o(1)} \leq n^{1+o(1)} \epsilon^{-1.99}$, as desired.

This observation means that it suffices to consider a fixed target cut $C$ of weight $\alpha c$, for $\alpha \leq \alpha^*$, and to show that $C$ is found with probability at least $n^{-1-o(1)} \epsilon^{1.98}$. This will be our goal for the remainder of this section. \emph{For the remainder of this section, we let $C$ be a fixed target cut and $\alpha c$ its weight.}

(Note: significantly better exponents than $3, 1.99$ can be shown for $n, \epsilon$ respectively; but this requires more careful calculations. Improvements in these exponents would not improve the overall run time of our algorithm, so these are omitted.)

\subsection{Storing the cuts produced by the RCA}
As we have already discussed in Section~\ref{dsa-sec}, we do not wish to explicitly store all the cuts that will be output by the RCA. Rather, we wish to store them in a compressed format. We briefly discuss how the RCA is to able to output each cut in terms of this compressed representation.

First, we must store a unique index for each cut $C$, which we define by
$$
H(C) = \sum_{v \in A} H(v) \qquad \text{modulo $2^b$}
$$
where $A$ is the first shore of $C$.

To compute this, we note that during the intermediate steps of the RCA, every node in the resulting subgraphs $G_i$ corresponds to a set of vertices of $G$. For each such vertex $v \in G_i$, which corresponds to some subset $X_v \subseteq V$, we define $H(v) = \sum_{x \in X_v} H(x) \text{ mod $2^b$}$. This can be updated at each contraction step. At the end of the RCA, we must select a cut of the graph $G_{\lceil 2 \alpha \rceil}$. The cut $C'$ of the graph $G_{\lceil 2 \alpha \rceil}$ corresponds to a cut $C$ of $G$ and $H(C) = \sum_{v \in C'} H(v)$. A similar process can be used to count the weight of $C$.

Finally, we must keep track of a pointer which allows us to fully reconstruct any cut in its uncompressed form. To do so, we first store the complete random tape used by applications of the RCA. Then any cut output by the RCA can be reconstructed in terms of its tree-path through the recursive calls to the RCA.

These details are discussed more in \cite{karger}.

\subsection{Handling the easy cases}
Our analysis of the RCA will be technically challenging. As a preliminary, we can deal with some of the easy cases first.
\begin{proposition}
\label{Arcathm2}
The RCA enumerates $C$  with probability $\geq n^{-1-o(1)} \epsilon^{1.98}$, assuming that one of the following conditions is satisfied:
\begin{enumerate}

\item $\alpha \leq 3/2$
\item $\rho \geq 2.1$
\item $\beta \geq 1$
\end{enumerate}
\end{proposition}
\begin{proof}
By Proposition~\ref{k-prop}, we have $\delta \geq \beta + \Omega(1)$ and $\delta \geq K - 2 > 0$ and $\beta \in [0,2]$.

We break this proof into a number of cases.

\textbf{Case I: \boldmath${\alpha \leq 3/2}$\unboldmath.} By Theorem~\ref{rca-perf} the cut $C$ is found with probability at least $n^{2-2 \alpha + o(1)} 2^{-\alpha^*_{\text{max}}} \geq n^{-1-o(1)} 2^{-O(1 + \rho)} \geq n^{-1-o(1)} \epsilon^{-o(1)}$ as desired. Hence for the remainder of this proof we assume $\alpha \geq 3/2$.

\textbf{Case II: \boldmath${\beta \geq 3/2}$\unboldmath.} By Theorem~\ref{rca-perf}, $C$ is found with probability at least $\Omega(2^{-{\alpha^*_{\text{max}}}} n^{2-2 \alpha}) \geq 2^{-O(1+\rho)} n^{2-2 \alpha*} \geq (\epsilon/n)^{o(1)} n^{2-2 \alpha^*}$.

By Lemma~\ref{Aalphastarlemma1}, we have $\alpha^* \leq \frac{2 - \beta + \delta + \rho}{\delta} + O(1/\log n)$. For $\beta \geq 3/2$ we have $\alpha^* \leq 4/3 + 2/3 \rho$. So $C$ is found with probability at least $(\epsilon/n)^{o(1)} n^{2 -4/3 - 2/3 \rho} = n^{-2/3 - o(1)} \epsilon^{2/3 +o(1)}$ as desired.

Thus, for the remainder of the proof, we assume $\beta \leq 3/2$.

\textbf{Case III: \boldmath$\rho > 2.1$\unboldmath.} In this case, we make the simple observation that a single application of the RCA is at least as powerful as a single application of the Contraction Algorithm (if we simply ignore all but one branch of the RCA). By Proposition~\ref{av-prop-1}, then $C$ is enumerated with probability at least $n^{-\alpha^* \bar h(0)} (n/\epsilon)^{-o(1)}$. 

Assuming that $\beta \leq 3/2$ (otherwise case II holds), then by Lemma~\ref{Aalphastarlemma1}:
\begin{align*}
\alpha^* \bar h(0) &\leq \Bigl( \frac{(3 - \beta + \delta)^2(2 - \beta + \delta + \rho)}{(2 - \beta + \delta)^2(2 + \delta)} + O(\frac{(1+\rho) \log  \log  n}{\log  n})) \Bigr) \bar h(0) \\
&= \frac{(5 -2 \beta +2 \delta) (2 -\beta +\delta +\rho)}{(2 -\beta +\delta)^2} + o(1+\rho) \\
&= \Bigl( \frac{5 - 2 \beta + 2 \delta}{2 - \beta + \delta} - o(1) \Bigr) + \rho \Bigl( \frac{5 - 2 \beta + 2 \delta}{(2 - \beta + \delta)^2} - o(1) \Bigr) 
\end{align*}

It can be verified mechanically that
\begin{align*}
\frac{5 - 2 \beta + 2 \delta}{2 - \beta + \delta} \leq 5/2 - \Omega(1) & \qquad \frac{5 - 2 \beta + 2 \delta}{(2 - \beta + \delta)^2} \leq 5/4 - \Omega(1)
\end{align*}

Hence $\alpha^* \bar h(0) \leq 5/2 + 5/4 \rho$. When $\rho > 2.1$, this quantity is less than $1 + 1.97 \rho$. Recalling that $\rho = \log(1/\epsilon)/\log n$, this means that $n^{-\alpha^* \bar h(0)} \geq n^{-1-o(1)} \epsilon^{1.98}$ as desired.

\textbf{Case IV: \boldmath$1 \leq \beta \leq 3/2$\unboldmath.} By Lemma~\ref{Aalphastarlemma1}
$$
\alpha^* \leq \frac{(3 - \beta + \delta)^2(2 - \beta + \delta + \rho)}{(2 - \beta + \delta)^2(2 + \delta)} + o(1+\rho)
$$
It can be verified mechanically that this is at most $ 3/2 +  (3/4) \rho$ for all $\beta \in [1,3/2]$ and $\delta \geq \beta$. So again by Theorem~\ref{rca-perf}, then $C$ is found with probability at least $n^{2 -2 \alpha - o(1)} \geq n^{-1 - o(1) - 3/2 \rho - o(1) \rho} \geq n^{-1-o(1)} \epsilon^{1.51}$
\end{proof}

\textbf{Simplifying assumptions.} Hence we can make the simplifying assumptions that $\alpha \geq 3/2, \rho \leq 2.1, \beta \in [0,1]$. As $\alpha^* \leq \alpha^*_{\text{max}} = O(1+\rho)$ by Proposition~\ref{amax-prop}, this implies that $\alpha \leq O(1)$ as well.  We summarize the following assumptions which we adopt in Sections~\ref{rcasec1}, \ref{rcasec2}, \ref{rcasec3}, and which which will not be stated explicitly there:

\begin{enumerate}
\item $\rho \leq 2.1$
\item $3/2 \leq \alpha \leq \alpha^* \leq \alpha^*_{\text{max}} \leq O(1)$
\item $\delta > \delta_0 > 0$
\item $\delta \geq \beta$
\item $\beta \in [0, 1]$
\end{enumerate}

In light of these simplifying assumptions, we can simplify Theorem~\ref{rca-perf}.
\begin{proposition}
\label{rca-perf2}
Let $H$ be a contraction-subgraph of $G$ with $r \leq n$ vertices. Then the RCA with parameter $\alpha^*_{\text{max}}$ selects cut $C$ with probability $\Omega(r^{2-2\alpha})$
\end{proposition}
 \begin{proof}
 We have $2^{-\alpha^*_{\text{max}}} \geq 2^{-O(1)} = \Omega(1)$ and $\alpha \geq 3/2 > 1 + \Omega(1)$
 \end{proof}

\subsection{The RCA as a branching process}
\label{rcasec1}
It is useful to view the sequence of graphs produced by the RCA in terms of a branching tree, in which there is a node in this branching process for each recursive call to the RCA. Every node in this process corresponds to an partial execution of the Contraction Algorithm. We view the root node (depth $0$) as corresponding to the original graph $G$.  The leaf nodes correspond to subgraphs which have $\leq \lceil 2 \alpha \rceil$ vertices. If our target cut $C$ appears in any of the subgraphs corresponding to a leaf node, it will be output at least once by the RCA. For any node $v$ of this branching process, we let $G^v$ be the corresponding subgraph. For the root node $v$ we have $G^v = G$. We say that node $v$ \emph{succeeds} if the cut $C$ remains in $G^v$.

To avoid quantization issues, it will be convenient to assume that $n$ is a power of two. We can always modify our original input graph to achieve this, by adding extra dummy vertices; each dummy vertex is connected to every other vertex (including the other dummy vertices) by infinitely many edges. This modification does not change $U(p)$, and it only changes $n$ by a constant factor. As we aim to give running times which are polynomial in $n$ and do not depend on $m$, this will not affect our runtime analysis by more than a constant factor.

Once we make this assumption, then we can see that, aside from the leaf nodes, the nodes at depth $i$ of the RCA branching process correspond to running the Contraction Algorithm from the original graph $G$ to a graph with exactly $n 2^{-i}$ nodes. The overall depth of this tree is $H = \lfloor \log_2( \frac{n}{\lceil 2 \alpha^*_{\text{max}} \rceil} ) \rfloor$. As $\alpha^*_{\text{max}} = O(1)$, we have that $\log_2 n - O(1) \leq H \leq \log_2 n$.

Our goal is to calculate the probability that some leaf node succeeds. In the ideal case (if the leaves corresponded to independent executions of the Contraction Algorithm), then the probability of a successful leaf node is nearly equal to the expected number of leaf nodes. 

The analysis of \cite{karger-stein} viewed the RCA as a type of percolation process, in which the survival of a target cut $C$ percolated from the root node to the leaf nodes. Starting at any given node at depth $i$, there is a certain probability that a cut $C$ is retained throughout all the contractions from that node to its children. All the survival events are independent.

Our analysis, however, must also keep track of the stochastic evolution of the subgraphs $G_i$ during the Contraction Process. Our improvements have been based on bounding the \emph{expected} number of edges in these graphs. Thus, we cannot assume that starting at a node $v$ at depth $i$, the probability that cut $C$ survives in two nodes is the square of the probability that it survives in one; the reason is that these events are both dependent on the random variable $G^v$. 

As a result of this, there is a significant correlation between the successes at nearby nodes of the RCA tree. This in turn lowers the overall probability that there is at least one successful leaf node. \emph{Dealing with this correlation among nearby nodes will be the major technical challenge for the analysis of the RCA}. This correlation becomes more severe at greater depths of the RCA tree.

\subsection{Marking nodes}
\label{rcasec2}
We wish to show that there is a good probability of having at least one successful node. We will do so by using Inclusion-Exclusion (IE):
$$
\bP( \text{At least one successful node} ) \geq \bE[ \text{\# successful nodes}] - \bE[ \text{\# pairs of successful nodes}]
$$

One somewhat paradoxical feature of Inclusion-Exclusion is that sometimes it can give worse estimates when the expected number of successes goes up. For example, suppose we have $n$ independent Bernoulli-$q$ random variables $X_1, \dots, X_n$. By IE, we have $\bP( \sum X_i \geq 1) \geq n q - \binom{n}{2} q^2$. When $q \rightarrow 0$, then this is essentially an accurate estimate, but when $q \rightarrow 1$ then the RHS becomes negative --- yielding a completely useless estimate.

We will avoid this problem for our analysis as follows. Every node $v$ at depth $k$ of the RCA corresponds to a partial execution of the Contraction Algorithm, up to a subgraph with $n 2^{-k}$ vertices. We can define the random variables $S_i, M_i, G_i, F_i$ for $i = n 2^{-k}, \dots, n$ for this node; we denote them by $S_i^v, M_i^v, G_i^v, F_i^v$. Observe that $G_{n 2^{-k}}^v = G^v$.

Now, if the cut $C$ survives to $G^v$, then mark the node $v$ with probability $\min(1, e^{\alpha(S^v_{n 2^{-k}} - B)})$, where $B$ is a threshold value which we will specify. All such markings are done independently. If the cut $C$ does not survive in $G^v$, then we do not mark $v$.

Clearly, the probability that some node survives to depth $k$ is at least equal to the probability that there is a marked node at depth $k$. It is the latter that we will estimate via IE.

For any execution of the Contraction Algorithm and any integer $k \geq 0$, we define
$$
T_k = \max(S_{n 2^{-k}},B)
$$
 
We will show a lower bound on the probability that a depth-$k$ node is marked, and an upper bound on the probability that two depth-$k$ nodes are both simultaneously marked.

\begin{proposition}
\label{mark-node-prop1}
Let $v$ be a node at depth $k$ of the RCA tree. The probability that node $v$ is marked is at least
$$
\bP( \text{$v$ is marked} ) \geq \Omega(\bE_{\mathcal{CP}(G,C)} [e^{-\alpha T_k}])
$$
\end{proposition}
\begin{proof}
Let $i = n 2^{-k}$. The partial history $F^v_{> i}$ for the node $v$ has a distribution which is simply the Contraction Algorithm up to stage $i$. In order for the node $v$ to be marked, the edges of that history must all be disjoint to $C$. So, integrating over such edges,
{\allowdisplaybreaks
\begin{align*}
\bP( \text{$v$ marked} ) &= \sum_{ \substack{e_{i+1}, \dots, e_n\\\text{disjoint to $C$}}} \bP_{\mathcal{CA}(G)}( F_{>i} = \langle e_n, \dots, e_{i+1} \rangle) \times \min(1, e^{\alpha(S_i - B)}) \\
&= \sum_{ \substack{e_{i+1}, \dots, e_n\\\text{disjoint to $C$}}} \bP_{\mathcal{CP}(G,C)} ( F_{>i} = \langle e_n, \dots, e_{i+1} \rangle) \min(1, e^{\alpha(S_i - B)}) \prod_{r=i+1}^n (1 - \frac{|C|}{M_r})  \\
& \qquad \text{by Proposition~\ref{cp-ca-prop}} \\
&\geq 16^{-\alpha} \sum_{ \substack{e_{i+1}, \dots, e_n\\\text{disjoint to $C$}}} \bP_{\mathcal{CP}(G,C)} ( F_{>i} = \langle e_n, \dots, e_{i+1} \rangle) \min(1, e^{\alpha(S_i - B)}) e^{-\alpha S_i} \qquad \text{by Lemma~\ref{Aslemma}} \\
&\geq 16^{-\alpha} \sum_{ \substack{e_{i+1}, \dots, e_n\\\text{disjoint to $C$}}} \bP_{\mathcal{CP}(G,C)} ( F_{>i} = \langle e_n, \dots, e_{i+1} \rangle)  \min(e^{-\alpha S_i}, e^{-\alpha B}) \\
&= 16^{-\alpha} \sum_{ \substack{e_{i+1}, \dots, e_n\\\text{disjoint to $C$}}} \bP_{\mathcal{CP}(G,C)} ( F_{>i} = \langle e_n, \dots, e_{i+1} \rangle) e^{-\alpha T_k} \\
&= 16^{-\alpha} \bE_{\mathcal{CP}(G,C)}[e^{-\alpha T_k}]
\end{align*}
}

Now as $\alpha \leq \alpha^*_{\text{max}} \leq O(1)$ we have the claimed result.

\end{proof}

\begin{proposition}
\label{mark-node-prop2}
Let $v, v'$ be two nodes at depth $k$ of the RCA tree, whose nearest common ancestor $w$ is at depth $j \leq k$. The probability that $v, v'$ are both marked satisfies the bounds
\begin{align*}
\bP( \text{$v$ and $v'$ marked} ) & \leq \bE_{\mathcal{CP}(G,C)} [e^{-\alpha T_k}] \\
\bP( \text{$v$ and $v'$ marked} ) & \leq n^{\alpha h(1 - \frac{j}{\log_2 n})} \bE_{\mathcal{CP}(G,C)} [e^{-2 \alpha T_k}]
\end{align*}
\end{proposition}
\begin{proof}
To show the first bound, we simply ignore node $v'$ completely. In this case, the proof is nearly identical to Proposition~\ref{mark-node-prop1}, except we use an upper bound instead of a lower bound. 

Let us turn our attention to the second bound. Let $i = n 2^{-k}$ and $t = n 2^{-j}$. The graph $G_w$ is derived by running the Contraction Algorithm up to $j$ nodes. Conditional on $G_w$, the nodes $v, v'$ are independent, and they are derived by running the Contraction Algorithm starting at $G_w$ down to $i$ nodes. We can thus view the joint distribution on the histories at $v, v'$ in terms of a ``double Contraction Algorithm'' denoted $\mathcal{DCA}(G)$. One can similarly define a ``double Contraction Process'' denoted $\mathcal{DCP}(G,C)$, where at every stage we select an edge uniformly from the edges outside $C$.

For this process, let $F_{i+1}, \dots, F_n$ be the resulting set of edges for the first leg, and let $F'_{i+1}, \dots, F'_n$ be the resulting set of edges for the second leg. Both $F_{>i}$ and $F'_{> i}$ are distributed according to the Contraction Process, but are not independent. Also, we must have $F_{\ell} = F'_{\ell}$ for $\ell = t+1, \dots, n$.

Now suppose we are given a sequence of edges $f_i, \dots, f_n, f'_i, \dots, f'_n$, with $f_{\ell} = f'_{\ell}$ for $\ell = t+1, \dots, n$; we wish to analyze whether they could be produced by double Contraction Algorithm. In the usual way for analyzing the Contraction Algorithm, let the resulting subgraphs by defined by $G_n = G'_n = G$, and $G_{\ell -1} = G_{\ell} \slash f_{\ell}$ etc. We similarly define $M_{\ell}, M'_{\ell}, S_{\ell}, S'_{\ell}$. Finally, we define $T_k = \max(S_i, B)$ and similarly for $T'_k = \max(S'_i, B)$. Finally we define $H = G_t$; this is the graph at which the two legs of the double Contraction Process first diverge.
{\allowdisplaybreaks
\begin{align*}
&\bP_{\mathcal{DCA}(G,C)} ( F_{>i} \vec f \wedge F_{>i} = \vec f') \\
&\qquad = \bP_{\mathcal{CA}(G)} (F_{>t} = \langle f_n, \dots, f_{t+1} \rangle)  \bP_{\mathcal{CA}(H)} (F_{>i} = \langle f_t, \dots, f_{i+1} \rangle) \bP_{\mathcal{CA}(H)} (F_{>i} = \langle f'_t, \dots, f'_{i+1} \rangle) \\
& \qquad = \bP_{\mathcal{CP}(G,C)} (F_{>t} = \vec f) \prod_{r=t+1}^n (1 - \frac{|C|}{M_r})  \bP_{\mathcal{CA}(H)} (F_{>i} = \langle f_t, \dots, f_{i+1} \rangle)  \prod_{r=i+1}^t (1 - \frac{|C|}{M_r}) \\
&\qquad \qquad \times \bP_{\mathcal{CA}(H)} (F_{>i} = \langle f'_t, \dots, f'_{i+1} \rangle)  \prod_{r=i+1}^t (1 - \frac{|C|}{M'_r})
\qquad \qquad \text{by Proposition~\ref{cp-ca-prop}} \\
&\qquad = \bP_{\mathcal{DCP}(G,C)} (F_{>i} = \vec f \wedge F'_{>i} = \vec f') \prod_{r=t+1}^n (1 - \frac{|C|}{M_r}) \prod_{r=i+1}^t (1 - \frac{|C|}{M_r}) \prod_{r=i+1}^t (1 - \frac{|C|}{M'_r}) \\
&\qquad \leq \bP_{\mathcal{DCP}(G,C)} (F_{>i} = \vec f \wedge F'_{>i} = \vec f') e^{-\alpha S_t} e^{-\alpha (S_i - S_t)} e^{-\alpha(S'_i - S_t)} \qquad \text{by Lemma~\ref{Aslemma}} \\
&\qquad = \bP_{\mathcal{DCP}(G,C)} (F_{>i} = \vec f \wedge F'_{>i} = \vec f') e^{\alpha S_t} e^{-\alpha S_i} e^{-\alpha S'_i} \\
&\qquad \leq n^{\alpha h(\frac{\log t}{\log n})} \bP_{\mathcal{DCP}(G,C)} (F_{>i} = \vec f \wedge F'_{>i} = \vec f')e^{-\alpha S_i} e^{-\alpha S'_i} \qquad \text{by Proposition~\ref{simplecaprop}}
\end{align*}
}

If we include the probability that both nodes are also marked, we have:
\begin{align*}
\bP( F^v_{>i} = \vec f \wedge F^{v'}_{>i} = \vec f' \wedge \text{$v, v'$ marked}) &= \bP_{\mathcal{DCA}(G,C)} ( F_{>i} = \vec f \wedge F_{>i} = \vec f') \wedge \bP( \text{$v, v'$ marked}) \\
&\leq n^{\alpha h(\frac{\log t}{\log n})} \bP_{\mathcal{DCP}(G,C)} (F_{>i} = \vec f \wedge F'_{>i} = \vec f')e^{-\alpha S_i} e^{-\alpha S'_i} \\
& \qquad \times \min(1, e^{\alpha(S_i - B)})  \min(1, e^{\alpha(S'_i - B)}) \\
&\leq n^{\alpha h( \frac{\log t}{\log n})} \bP_{\mathcal{DCP}(G,C)} (F_{>i} = \vec f \wedge F'_{>i} = \vec f')e^{-\alpha T_k} e^{-\alpha T'_k}
\end{align*}

Now, integrating over $\vec f, \vec f'$,
\begin{align*}
\bP( F^v_{>i} = \vec f \wedge F^{v'}_{>i} = \vec f' \wedge \text{$v, v'$ marked}) &\leq \sum_{\vec f, \vec f'} n^{\alpha h( \frac{\log t}{\log n})} \bP_{\mathcal{DCP}(G,C)} (F_{>i} = \vec f \wedge F'_{>i} = \vec f') e^{-\alpha T_k} e^{-\alpha T'_k} \\
&= n^{\alpha h( \frac{\log t}{\log n})} \bE_{T_k, T'_k \sim \mathcal{DCP}(G,C)} [e^{-\alpha T_k} e^{-\alpha T'_k}] \\
&\leq n^{\alpha h( \frac{\log t}{\log n})} \sqrt{ \bE_{T_k, T'_k \sim \mathcal{DCP}(G,C)} [e^{-2 \alpha T_k}] \bE_{T_k, T'_k \sim \mathcal{DCP}(G,C)} [e^{-2 \alpha T'_k}]} \\
& \qquad \text{by Cauchy-Schwarz} \\
&= n^{\alpha h( \frac{\log t}{\log n})} \sqrt{ \bE_{T_k \sim \mathcal{CP}(G,C)} [e^{-2 \alpha T_k}] \bE_{T'_k \sim \mathcal{CP}(G,C)} [e^{-2 \alpha T'_k}] } \\
&= n^{\alpha h( \frac{\log t}{\log n})} \bE_{T_k \sim \mathcal{CP}(G,C)} [e^{-2 \alpha T_k}]
\end{align*}

Finally, observe that $t = n 2^{-j}$, giving the final result.
\end{proof}

\subsection{Inclusion-Exclusion analysis}
\label{rcasec3}
We now move on to use Inclusion-Exclusion. We will not directly apply it to show that there is a leaf node where $C$ survives --- the correlation among the nodes is too severe for nodes deep in the RCA tree. We will count smaller-depth nodes instead.

\begin{proposition}
\label{mark-node-prop3}
For any integers $j, k$ with $0 \leq j \leq k \leq H$, we have
$$
\bE[ \text{ \# pairs of marked nodes at depth $k$} ] \leq O \Bigl( \bE_{T_k \sim \mathcal{CP}(G,C)} \Bigl[ 4^{2 k - j} e^{-\alpha T_k} + n^{\alpha h(1 - \frac{j}{\log_2 n})} e^{-2 \alpha T_k} \Bigr] \Bigr) \\
$$
\end{proposition}
\begin{proof}
For each $i \leq k$, we count the number of level-$k$ nodes which are marked and have a nearest common ancestor at level $i$. As every node has four children, the total number of such nodes pairs is $4^k \times 3 \times 4^{k-i-1}$. For each such pair, we use Proposition~\ref{mark-node-prop2} to estimate the probability that both nodes are marked; we use the cruder estimate $\bE_{\mathcal{CP}(G,C)} [e^{-\alpha T_k}]$ for $i \geq j$ and the more refined estimate $n^{\alpha h(\frac{\log j}{\log n})} \bE_{\mathcal{CP}(G,C)} [e^{-2 \alpha T_k}]$ for $i < j$.

Thus, summing over all $i = 0, \dots, k$, we have
\begin{align*}
\bE[ \text{ \# pairs of marked nodes at depth $k$} ] &\leq \sum_{i=j+1}^k 4^k \times 3 \times 4^{k-i-1}  \bE_{\mathcal{CP}(G,C)} [e^{-\alpha T_k}] \\
&\qquad + \sum_{i=0}^j 4^k \times 3 \times 4^{k-i-1} n^{\alpha h(1 - \frac{i}{\log_2 n})} \bE_{\mathcal{CP}(G,C)} [e^{-2 \alpha T_k}] \\
&\leq 4^{2 k - j} \bE_{\mathcal{CP}(G,C)} [e^{-\alpha T_k}] + 4^{2 k} \bE_{\mathcal{CP}(G,C)} [e^{-2 \alpha T_k}] \sum_{i=0}^j 4^{-i} n^{\alpha h( 1 - \frac{i}{\log_2 n})} 
\end{align*}

By Proposition~\ref{h-deriv-prop} we have that $h'(x) \leq -4/3 - \Omega(1)$ for all $x \in [0,1]$.  Thus, for all $x \geq 1 - \frac{j}{\log_2 n}$, we have $h(x) \leq h( 1 - \frac{j}{\log_2 n}) - (4/3 - \Omega(1)) (x - (1 - \frac{j}{\log_2 n}))$. So we estimate
\begin{align*}
4^{-i} n^{\alpha h(1 - \frac{i}{\log_2 n})} &\leq 4^{-i} n^{\alpha h(1 - \frac{j}{\log_2 n}) - \frac{\alpha (4/3 - \Omega(1))}{\log_2 n}} \leq 4^{-i} n^{\alpha h(1 - \frac{j}{\log_2 n})} 2^{-\alpha (j-i) (4/3 - \Omega(1))} \\
&\leq 4^{-i} n^{\alpha h(1 - \frac{j}{\log_2 n})} 2^{-3/2 (j-i) (4/3 - \Omega(1))} \qquad \text{as $\alpha > 3/2$} \\
&= 4^{-j} n^{\alpha h(1 - \frac{j}{\log_2 n})} (1 - \Omega(1))^{j - i}
\end{align*}

So $\sum_{i=0}^j 4^{-i} n^{\alpha h( 1 - \frac{i}{\log_2 n})} \leq O \Bigl( 4^{-j} n^{\alpha h(1 - \frac{j}{\log_2 n})} \Bigr)$. This gives us our final estimate

\begin{align*}
\bE[ \text{ \# pairs of marked nodes at depth $k$} ] &\leq  4^{2 k - j} \bE_{\mathcal{CP}(G,C)} [e^{-\alpha T_k}]) + O( 4^{2 k - j}  n^{\alpha h(1 - \frac{j}{\log_2 n})}  \bE_{\mathcal{CP}(G,C)} e^{-2 \alpha T_k} ) \\
&= O \Bigl( \bE_{\mathcal{CP}(G,C)} [ 4^{2 k - j} e^{-\alpha T_k} + n^{\alpha h(1 - \frac{j}{\log_2 n})} e^{-2 \alpha T_k} ] \Bigr)
\end{align*}
\end{proof}

\begin{lemma}
\label{Arcalemma2}
Suppose that $j,k$ are integers in the range $0 \leq j \leq k \leq H$.  Let $\psi > 0$ be any sufficiently small constant. Then the RCA with parameter $\alpha^*_{\text{max}} $ finds the target cut $C$ with probability at least
\begin{align*}
&\bP( \text{RCA tree includes $C$} ) \geq \Omega \Bigl( \bE_{\mathcal{CP}(C,G)} \Bigl[ 4^{\alpha k} n^{2 - 2 \alpha} e^{-\alpha T_k} - \psi 4^{2 \alpha k - j} n^{4 - 4 \alpha} \bigl( e^{-\alpha T_k} + n^{\alpha h(1 - \frac{j}{\log_2 n})} e^{-2 \alpha T_k} \bigr) \Bigr] \Bigr)
\end{align*}
\end{lemma}
\begin{proof}
Each marked node at depth $k$ of the RCA tree corresponds to a subgraph with $n 2^{-k}$ vertices which contains cut $C$. Hence, if we execute the remaining stages of the RCA, by Proposition~\ref{rca-perf2} it will generate some leaf node containing $C$ with probability $\Omega( (n 2^{-k})^{2 - 2 \alpha})$. So each such marked node $v$ has some independent probability $q_v \geq \phi' (n 2^{-k})^{2 - 2 \alpha}$ that a leaf node below it contains cut $C$; here $\phi'$ is some constant.

Let $\phi$ be any constant such that $\phi \in (0,\phi']$ and let $q = \phi (n 2^{-k})^{2 - 2 \alpha}$. Now we apply an additional randomization step: we \emph{discard} $v$ and all its subtree below $v$ with probability $q_v - q$. This is a valid probability as $q_v \geq q$. Also, it is easy to see now that for any marked node, the probability that there is a non-discarded leaf node containing $C$ is \emph{exactly} equal to $q$.

Note that we have applied two ``attenuation'' steps which can only reduce the probability of finding $C$ among the nodes of the RCA-tree. First, some unmarked nodes may have a leaf node containing $C$, and some marked nodes will contain $C$ but be discarded. 
If a marked node $v$ at depth $k$ contains some leaf node containing $C$, and $v$ is not discarded, we say that $v$ is a \emph{winner}. Each marked node is a winner with probability exactly equal to $q$.

By Proposition~\ref{mark-node-prop1}, the expected number of marked leaf nodes is $\Omega(4^k \bE_{\mathcal{CP}(G,C)}[e^{-\alpha T_k}])$ and so the expected number of winners is $\Omega(q 4^k \bE_{\mathcal{CP}(G,C)}[e^{-\alpha T_k}])$.

 By Proposition~\ref{mark-node-prop3}, the expected number of pairs of marked leaf nodes is $O \Bigl( \bE_{\mathcal{CP}(G,C)} \Bigl[ 4^{2 k - j} e^{-\alpha T_k} + n^{\alpha h(1 - \frac{j}{\log_2 n})} e^{-2 \alpha T_k} \Bigr] \Bigr)$ and so the expected number of pairs of winners is $O \Bigl( q^2 4^{2 k - j} \bE_{\mathcal{CP}(G,C)} \Bigl[ e^{-\alpha T_k} + n^{\alpha h(1 - \frac{j}{\log_2 n})} e^{-2 \alpha T_k} \Bigr] \Bigr)$.

By Inclusion-Exclusion, the probability that there is at least one winner can be lower bounded by:
\begin{align*}
\bP( \text{RCA tree includes $C$} ) &\geq \bP( \text{There is a winner at depth $k$}) \\
&\geq \Omega \Bigl( q 4^k \bE_{\mathcal{CP}(G,C)}[e^{-\alpha T_k}]) - O\Bigl( q^2 4^{2 k - j} \bE_{\mathcal{CP}(G,C)} \Bigl[ e^{-\alpha T_k} + n^{\alpha h(1 - \frac{j}{\log_2 n})} e^{-2 \alpha T_k} \Bigr]\Bigr) \Bigr) \\
&= \bE_{\mathcal{CP}(G,C)} \Bigl[ c_1 \phi 4^{\alpha k} n^{2 - 2 \alpha} e^{-\alpha T_k} - c_2 \phi^2 4^{2 \alpha k - j} n^{4 - 4 \alpha} \bigl( e^{-\alpha T_k} + n^{\alpha h(1 - \frac{j}{\log_2 n})} e^{-2 \alpha T_k} \bigr) \Bigr] \\
& \qquad \text{where $c_1, c_2$ are constants}
\end{align*}

Set $\phi = \min(\phi', \frac{\psi c_1}{c_2})$ and we lower-bound this expression as
$$
 \Omega \Bigl( \bE_{\mathcal{CP}(C,G)} \Bigl[ 4^{\alpha k} n^{2 - 2 \alpha} e^{-\alpha T_k} - \psi 4^{2 \alpha k - j} n^{4 - 4 \alpha} \bigl( e^{-\alpha T_k} + n^{\alpha h(1 - \frac{j}{\log_2 n})} e^{-2 \alpha T_k} \bigr) \Bigr] \Bigr)
$$
\end{proof}

\begin{lemma}
\label{Arcalemma3}
Let $y \in [0,1]$, and define $t = \max(h(y), \bar h(y/\alpha))$. Then the RCA with parameter $\alpha_{\text{max}}^*$ finds the target cut $C$ with probability at least
$$
\bP(\text{RCA succeeds}) \geq n^{2 - \alpha t - 2 y - o(1)}
$$
\end{lemma}
\begin{proof}
Let $x = y/\alpha$, so that $t = \max(h(y), \bar h(x))$.  We set \begin{align*}
j &= \min \bigl( H,  \lfloor (1 - y) \log_2 n  \rfloor \bigr) \\
k &= \min \bigl( H,  \lfloor (1 - x)  \log_2 n \rfloor \bigr)
\end{align*}

Clearly $j,k$ are integers and $0 \leq j \leq k \leq H$. Also, as $H \geq \log_2 n - O(1)$, we have that $| j - (1-y) \log_2 n | \leq O(1)$ and $|k - (1 - x) \log_2 n| \leq O(1)$.

We now observe that $h(1 - \frac{j}{\log_2 n}) = h( y \pm O(1/\log n) )$; as the derivative of $h$ is bounded by constants (Proposition~\ref{h-deriv-prop}), this is at most $t + O(1/\log n)$. Similarly, the derivative of $\bar h$ is bounded by constants, so $\bar h (1 - \frac{k}{\log_2 n}) \leq t + O(1/\log n)$.

We now apply Lemma~\ref{Arcalemma2} to obtain:
$$
\bP( \text{RCA tree includes $C$} ) \geq \Omega \Bigl( \bE_{\mathcal{CP}(G,C)} \Bigl[ 4^{\alpha k} n^{2 - 2 \alpha} e^{-\alpha T_k} -\psi 4^{2 \alpha k - j} n^{4 - 4 \alpha} \bigl( e^{-\alpha T_k} + n^{\alpha h(1 - \frac{j}{\log_2 n})} e^{-2 \alpha T_k} \bigr) \bigr] \Bigr)
$$
where $\psi$ is a constant parameter to be determined. Using the fact that $\alpha = O(1)$ and $|k - (1-x) \log_2 n| \leq O(1)$, we see that the terms can be bounded as
\begin{align*}
4^{\alpha k} n^{2 - 2 \alpha} &\geq \Omega(n^{2 - 2 y}) \\
4^{2 \alpha k - j} n^{4 - 4 \alpha} &\leq O(n^{2 - 2 y}) \\
n^{h(1 - \frac{j}{\log_2 n})} &\leq O(n^t)
\end{align*}

So
$$
\bP( \text{RCA tree includes $C$} ) \geq n^{2 - 2 y}  \Omega \Bigl( \bE_{\mathcal{CP}(G,C)} \Bigl[ \Omega(e^{-\alpha T_k}) - \psi O \bigl( e^{-\alpha T_k} + n^{\alpha t} e^{-2 \alpha T_k} \bigr) \Bigr] \Bigr)
$$

By choosing $\psi$ to be a sufficiently small constant, we obtain:

\begin{align*}
\bP( \text{RCA tree includes $C$} ) &\geq n^{2 - 2 y} \Omega  \Bigl( \bE_{\mathcal{CP}(G,C)} \Bigl[ e^{-\alpha T_k} - \tfrac{1}{3} \bigl( e^{-\alpha T_k} + n^{\alpha t} e^{-2 \alpha T_k} \bigr) \Bigr] \Bigr) \\
&= n^{2 - 2 y} \Omega  \Bigl( \bE_{\mathcal{CP}(G,C)} \Bigl[ e^{-\alpha T_k} (1 - \tfrac{1}{2} n^{\alpha t} e^{-\alpha T_k}) \Bigr] \Bigr)
\end{align*}

At this point, we are ready to specify the threshold value $B = t \log n$; recall that $T_k = \max(S_{n 2^{-k}},B)$. By definition of $T_k$ we have that $n^{\alpha t} e^{-\alpha T_k} \leq n^{\alpha t} e^{-\alpha B} = 1$ and so
\begin{align*}
\bP( \text{RCA tree includes $C$}) &\geq n^{2 -2 y} \Omega \Bigl( \bE_{\mathcal{CP}(G,C)}[e^{-\alpha T_k} (1 - \tfrac{1}{2}) \Bigr) \\
&= n^{2 - 2 y} \Omega ( \bE_{\mathcal{CP}(G,C)}[e^{-\alpha T_k}])
\end{align*}

We now claim that 
\begin{equation}
\label{goal-eqn1}
\bE_{\mathcal{CP}(G,C)} [e^{-\alpha T_k}] \geq n^{-\alpha t - o(1)}
\end{equation}
which will give us our claimed result. 

Let us write $S = S_{n 2^{-k}}$, and so $T_k = \max(S,B)$. Let $r = \bP(S > B)$. We show (\ref{goal-eqn1}) by considering separate cases for the size of $r$.

\textbf{Case I: \boldmath$r \leq \frac{1 + \alpha t \log n}{2 + \alpha t \log n}$\unboldmath.} In this case,
\begin{align*}
\bE_{\mathcal{CP}(G,C)} [e^{-\alpha T_k}] &\geq \bP(S \leq B) \times \bE_{\mathcal{CP}(G,C)} [e^{-\alpha T_k} \mid S \leq B] \\
&= (1-r) e^{-\alpha B} \\
& \geq (1 - \frac{1 + \alpha t \log n}{2 + \alpha t \log n})  e^{-\alpha B} \\
&= \frac{n^{-\alpha t}}{2 + \alpha t \log n} \\
&\geq n^{-\alpha t - o(1)} \qquad \text{as $\alpha, t \leq O(1)$}
\end{align*}

\textbf{Case II: \boldmath$r > \frac{1 + \alpha t \log n}{2 + \alpha t \log n}$\unboldmath.} Let $s = \bE_{\mathcal{CP}(G,C)}[S]$. Observe that $\bE_{\mathcal{CP}(G,C)}[S \mid S \geq B] \leq \bE_{\mathcal{CP}(G,C)}[S]/P(S \geq B) = s/r$. So
\begin{align*}
\bE_{\mathcal{CP}(G,C)} [e^{-\alpha T_k}] &\geq \bE[e^{-\alpha T_k} \mid S > B] \bP(S > B) \\
&= \bE[e^{-\alpha S} \mid S > B] \bP(S > B) \qquad \text{$S = T_k$ for $S > B$} \\
&\geq  e^{-\alpha\bE[S \mid S > B]} \bP(S > B) \qquad \text{by Jensen's inequality} \\
&\geq r e^{-\alpha s/r} \geq \tfrac{1}{2} e^{-\alpha s \times \frac{2 + \alpha t \log n}{1 + \alpha t \log n}}
\end{align*}

We now estimate $s$ using Theorem~\ref{Acutboundthm} as:
\begin{align*}
s = \bE_{\mathcal{CP}(G,C)}[S] &\leq \bar h(\frac{\log(n 2^{-k})}{\log n}) \log n + O(\log \log n) \\
&\leq \bar h(1 - \frac{k}{\log_2 n}) \log n + O(\log \log n) \\
&\leq \bar h(x) \log n + O(1) +  O(\log \log n) \\
&\leq t \log n + O(\log \log n)
\end{align*}

So 
\begin{align*}
\bE_{\mathcal{CP}(G,C)} [e^{-\alpha T_k}] &\geq \tfrac{1}{2} \exp(-\alpha \bigl( t \log n + O(\log \log n) \bigr)  \times \frac{2 + \alpha t \log n}{1 + \alpha t \log n}) \\
&\geq n^{-o(1)} \exp(-\alpha  t \log n  \times \frac{2 + \alpha t \log n}{1 + \alpha t \log n}) \qquad \text{as $\alpha = O(1)$} \\
&= n^{-o(1)} e^{1 - \alpha t \log n - \frac{2}{2 + \alpha t \log n}} \geq n^{-o(1)} n^{-\alpha t} 
\end{align*}

In either case, we have shown (\ref{goal-eqn1}).
\end{proof}

\begin{theorem}
\label{Arcathm3}
The RCA with parameter $\alpha^*_{\text{max}}$ finds any target $\alpha^*$-cut $C$ with probability at least $n^{-1-o(1)} \epsilon^{1.98}$. 
\end{theorem}
\begin{proof}
If the simplifying assumptions do not hold, then this is already covered by Proposition~\ref{Arcathm2}.

By Lemma~\ref{Arcalemma3}, $C$ is found with probability at least $n^{2 - \alpha \max(h(y), \bar h(y/\alpha)) - 2 y - o(1)}$ in each iteration of the RCA, for any $y \in [0,1]$. By Proposition~\ref{dprop99}, there is $y \in [0,1]$ so that this at least $n^{2 - 3 - 1.5 \rho - o(1)} = n^{-1 - 1.5 \rho - o(1)}$.
\end{proof}

We defer the proof of Proposition~\ref{dprop99} to  Appendix~\ref{bound-appendix}. It is simple conceptually but it is an extremely tedious exercise, based on computer checking of a very large number of cases.

\subsection{Further applications of the Recursive Contraction Algorithm}
\label{further-app-sec}
In this Section~\ref{further-app-sec}, we do not make any assumption on the size of $U(p)$.

The Recursive Contraction Algorithm provides a powerful approach to enumerate all the approximately-minimum cuts in a given graph $G$. Our reliability-estimation algorithm uses it in a very specific way based on bounds for the number of cuts of $G$. However, we can slightly improve the analysis of the Recursive Contraction Algorithm in other situations. Given a particular value of $\alpha$ and a given graph $G$ (with no information about the cut structure or reliability of $G$), how to enumerate the $\alpha$-cuts of $G$. The Recursive Contraction Algorithm was first introduced in \cite{karger-stein} to solve this problem. However, the original description of the RCA was not parametrized in the best way, leading to slightly sub-optimal running time. We improve this here.

\begin{theorem}
There is an algorithm to enumerate, with high probability, all $\alpha$-cuts of $G$ in time $O(n^{2 \alpha} \log  n)$.
\end{theorem}
\begin{proof}
First, suppose $\alpha < 3/2$. Then \cite{lincut} describes a data structure to represent all the $\alpha$-cuts in time $O(n^2)$ and to enumerate all the $\alpha$-cuts in time $O(n^2 \log  n)$.

Finally, suppose $3/2 \leq \alpha \leq \sqrt{n}$. We execute the following variant of the Recursive Contraction Algorithm, which we denote RCA2
\ttfamily
\begin{enumerate}
\item[1.] Randomly contract edges of $G$ until the resulting graph $G'$ has at most $\min(n^{-2/5}, \lceil 2 \alpha \rceil)$ vertices.
\item[2.] If $G'$ has $\lceil 2 \alpha \rceil$ vertices, output a random cut of $G$.
\item[3.] Otherwise, perform 2 independent executions of RCA2 on the graph $G'$.
\end{enumerate}
\rmfamily

A simple application of the Master Theorem for recurrences along with arguments from \cite{karger-stein} shows that this algorithm executes in time $O(n^{5/2})$. Also using arguments from \cite{karger-stein}, we can show that for any $\alpha$-cut $C$, and any execution of this algorithm, the cut is selected with probability $n^{5/2 - 2 \alpha} \exp(\Omega(\alpha \log \alpha))$.

Thus we can view the process of finding all $\alpha$-cuts as a Coupon Collector Problem. As there are at most $n^{2 \alpha}$ cuts, if we run $(2 \alpha \log n) \times n^{2 \alpha - 5/2} \exp(-\alpha \log \alpha) = n^{2 \alpha - 5/2} \log n$ independent executions of RCA2, then we find all cuts with high probability. This gives a total run time of $O(n^{2 \alpha} \log  n)$. 

\end{proof}

\section{Putting it all together}
\label{put-it-sec}
We may now put all the pieces of the algorithm together. Our basic plan is to use the cut-enumeration if $U(p) < n^{-K}$, and use Monte-Carlo sampling when $U(p) \geq n^{-K}$, where $K$ is some chosen parameter close to $2$. 

How do we determine which of these two methods to apply? At first, it would appear that this decision requires knowing $U(p)$, which is what we are trying to determine in the first place.
A simple way to make this decision is to run a preliminary Monte-Carlo sampling for $\phi n^{K}$ trials, where $\phi > 1$ is constant. If during any of these samples we observe the graph become disconnected, we will use Monte-Carlo sampling to estimate $U(p)$; otherwise we use cut-enumeration.

This achieves at least a gross discrimination between the two regimes:

\begin{proposition}
\label{choose-prop}
Suppose $U(p) \geq n^{-K}$. Then the probability of observing the graph become disconnected, is bounded from below by a constant which approaches $1$ for $\phi$ sufficiently large.

Suppose $U(p) \leq \phi' n^{-K}$, for $\phi'$ a constant. Then the probability of observing the graph become disconnected, is bounded from above by a constant which approaches $0$ for $\phi'$ sufficiently small.
\end{proposition}
\begin{proof}
In the first case, the number of times the graph becomes disconnected is a binomial random variable, with $\phi n^{K}$ trials and expected number of successes $\geq \phi$. By Chernoff's bound, the probability that there are zero successes is $\leq e^{-\phi}$; this approaches zero for $\phi$ sufficiently large.

In the second case, the expected number of successes is $\leq \phi'$. By Markov's inequality, the probability that that there is at least one success is $\leq \phi'$; this approaches zero for $\phi'$ sufficiently small.
\end{proof}

Now, when $U(p) \geq n^{-K}$, we will choose Monte-Carlo sampling with probability at least $0.99$. This is good because the cut-enumeration algorithm may not be well-behaved when $U(p) \geq n^{-K}$; most of our theorems completely break down in the regime $U(p) \geq n^{-2}$.

When $U(p) \leq \phi' n^{-K}$, we will use cut-enumeration with probability at least $0.99$.

When $\phi' n^{-K} < U(p) < n^{-K}$, we may use either Monte Carlo sampling or cut-enumeration. In this regime, \emph{either} of these two algorithms gives good performance. The Monte-Carlo sampling will have relative variance $O(n^{-K})$. We have already shown that when $U(p) < n^{-K}$ that cut-enumeration behaves correctly.

We now obtain our main result:

\begin{proposition}
\label{main-result-aa}
Let $\gamma > 0$ be any constant. Then we can estimate $U(p)$ in time $O(n^{3+\gamma} \epsilon^{-2})$ with probability $\geq 3/4$.
\end{proposition}
\begin{proof}
Let $K = 2 + \gamma/2$. 

Suppose that $U(p) < n^{-K}$ and we elect to use the cut-enumeration procedure. We run the RCA for $n^{3+o(1)} \epsilon^{-1.99}$ iterations, and let $\mathcal A$ denote the resulting set of cuts. By Theorem~\ref{Arcathm3}, $\mathcal A$ contains all the $\alpha^*$-cuts with high probability. In particular,
$$
(1-\epsilon) U(p) \leq U_{\alpha^*} (p) \leq U_{\mathcal A} (p) \leq U(p)
$$

So $U_{\mathcal A}(p)$ indeed estimates $U(p)$ to within relative error $\epsilon$. We next use the estimation algorithm of Section~\ref{Asec:statistical} with this collection. By Proposition~\ref{summarize-est-prop}, this estimates $U_{\mathcal A}(p)$, and hence $U(p)$, to within relative error $O(\epsilon)$. The total time for this is $O( |\mathcal A| \text{polylog}|\mathcal A| + n^2/\epsilon^2)$. As $|\mathcal A| \leq n^{3+o(1)} \epsilon^{-1.99}$, the term $|\mathcal A| \text{polylog} \mathcal |A|$ is at most $n^{3+o(1)} \epsilon^{-2}$. So the overall running time is indeed $n^{3+o(1)} \epsilon^{-2}$ as desired.

Suppose instead that $U(p) \geq \phi' n^{-K}$, where $\phi'$ is a constant chosen in Proposition~\ref{choose-prop}, and we elect to use Monte-Carlo sampling. The straightforward strategy for Monte Carlo sampling would require $O(m)$ time per sample. However, a sparsification method of \cite{karger} reduces this is $n^{1+o(1)}$ time per sample. (This sparsification is somewhat involved and we will not modify it; see \cite{karger} for a more thorough description). Overall, the time for Monte Carlo sampling $O(n^{K+1} \epsilon^{-2} \log^{O(1)} n)$.

The total work is $O(n^{3+\gamma} \epsilon^{-2})$ either way.

By Proposition~\ref{choose-prop}, with arbitrarily high constant probability, one of these two cases holds. This implies that we find a good estimate with probability $>3/4$, which is the goal of our FPRAS.

\end{proof}

Note that our estimates have been developed under the assumption that $\delta$ is uniformly bounded away from 0. This means that for any constant $\gamma > 0$, there exists an algorithm with running time $O(n^{3+\gamma} \epsilon^{-2})$. All of the asymptotic terms may depend arbitrarily on $\gamma$. As $\gamma \rightarrow 0$, then $\delta \rightarrow 0$ as well.

In fact one can produce these algorithms in a \emph{uniform} way. That is, there is a Turing machine $M$ which on input $t$ produces a Turing machine $ M_t$; in turn, $ M_t$ produces an estimate for $U(p)$ in time $n^{3+1/t} \epsilon^{-2}$. This follows from the fact that all the hidden constants in our proofs can be upper-bounded by \emph{computable} functions of $\lceil \frac{1}{K-2} \rceil$. (This is easy to show, but we have avoided doing so because explicitly tracking all these constant terms would make our analysis more difficult than it already is. This is a very weak property, as computable functions may grow very quickly.)

We first show a general condition for uniformization of FPRAS algorithm; we then conclude by applying it to estimating $U(p)$.

\begin{lemma}
\label{amalgamate-lemma}
Suppose that $M$ is a Turing machine with the following behavior: on input $t > 0$, it produces a Turing machine $ M_t$. Furthermore, $ M_t$ has the property that it estimates some statistic $X$ with relative error $\epsilon$ with probability $>3/4$, and $ M_t$ runs in time $\leq c_t n^{a+1/t} \epsilon^{-b}$.

Then there is an FPRAS for $X$ running in time $n^{a+o(1)} \epsilon^{-b}$.
\end{lemma}
\begin{proof}
First, by repeating any FPRAS for $\Omega(\log n)$ samples and taking the median, one can increase the success probability from $3/4$ to $1-1/n^2$. Let $ M'_t$ be the FPRAS derived by this amplification; clearly they are also computable and they have running times $O(c_t n^{a+1/t} \epsilon^{-b} \log n)$.

Now, consider the following algorithm: given some input vector $\vec x$ of bit-length $n$, we wish to estimate $X(\vec x)$. To do so, we dovetail the following computations: we run $ M(1), \dots,  M(\log n)$; as soon as $ M(i)$ terminates, we immediately run the resulting $ M'_i(\vec x)$. If any computation $ M'_i$ terminates, we output its answer and terminate.

We first claim that with probability $\geq 3/4$ this gives a correct answer. For, a necessary condition for this to fail is that one of $M'_i, \dots, M'_{\log n}$ fails on input $\vec x$; the total probability of this is at most $\log n \times 1/n^2 \leq 1/4$.

Next, we claim that for any fixed integer $s > 0$ that this algorithm runs in time $\leq n^{a+1/s} \epsilon^{-b}$ for $n$ sufficiently large. Let $t = 2 s$. Suppose that $M(t)$ takes time $k$ to compute $ M_t$. Then, for $\log n \geq t$, our algorithm runs $M'_t$ (among other algorithms). It would take time $k + c_t n^{a+1/t} \epsilon^{-b} \log n$ for algorithm $ M'_t$ to succeed. 

We are dovetailing $\log n$ separate computations. If the total run-time exceeds $(k + c_t n^{a+1/t} \epsilon^{-b} \log n) \log n$, then algorithm $ M_t$ will receive at least a share of $k + c_t n^{a+1/t} \epsilon^{-b} \log n$ runtime, and it will terminate with the answer (unless another process has terminated earlier). Thus, the overall run-time of this algorithm is at most $(k + c_t n^{a+1/t} \epsilon^{-b} \log n) \log n$. This is $\leq n^{a+1/s} \epsilon^{-b}$ for $n$ sufficiently large as $t > s$.

So, in all, our algorithm runs in time $\leq n^{a+1/t} \epsilon^{-b}$ for any $t > 0$ and $n$ sufficiently large. This implies it runs in time $n^{a+o(1)} \epsilon^{-b}$. It returns a correct answer with probability $\geq 3/4$, as required.
\end{proof}

\begin{theorem}
There is an algorithm to estimate $U(p)$ in time $n^{3+o(1)} \epsilon^{-2}$.
\end{theorem}
\begin{proof}
This follows from Lemma~\ref{amalgamate-lemma} and Proposition~\ref{main-result-aa}.
\end{proof}

\section{Concluding Remarks}
\label{sec:concl}
Two natural open questions are to see if a run-time bound such as $O(n^2/\epsilon^2)$ is possible, and to identify other possible applications of differential-equation approximations such as ours. 

It remains an outstanding open question to make progress on the approximability of $R(G,p) = R(p)$, the probability of $G$ remaining \emph{connected}: as pointed out by Leslie Goldberg to us, even very weak approximations here can be turned into PTAS-type approximations, and a proof of this from \cite{goldberg-jerrum:potts} is as follows. 
For example, suppose for some $\epsilon>0$ that there is a polynomial-time algorithm that
given as input any connected graph $H$ with $n_H$ vertices, produces a value $\hat{R}(H,p)$ satisfying
\begin{equation}
\label{eqn:weak-approx}
(1-\epsilon)^{\sqrt{n_H}} R(H,p) \leq \hat{R}(H,p) \leq (1+\epsilon)^{\sqrt{n_H}} R(H,p);
\end{equation}
note that this is a very weak approximation algorithm. However, it can be turned into essentially an $(1 \pm \epsilon)$--approximation algorithm as follows.
Take the connected input graph $G$ for which we want to well-approximate $R(G,p)$. Choose an arbitrary vertex $v$ of $G$, and construct a graph $H$ by taking $n_G$ copies of $G$, but by also graph-theoretically ``identifying" (i.e., fusing together) the copies of vertex $v$. 
Then $R(H,p) = R(G,p)^{n_G}$ and $n_H = n_G^2 - n_G + 1 \sim n_G^2$; in particular, $\sqrt{n_H} = n_G - 1/2 + o(1)$. 
Thus, by (\ref{eqn:weak-approx}),
\begin{eqnarray*}
(1 - \epsilon)^{\sqrt{n_H}} R(G,p)^{n_G} & \leq \hat{R}(H,p) \leq & (1+\epsilon)^{\sqrt{n_H}} R(G,p)^{n_G}, ~~\mbox{i.e.}, \\
(1 - \epsilon)^{1 - o(1)} R(G,p) & \leq (\hat{R}(H,p))^{1/n_G} \leq & (1+\epsilon)^{1 - o(1)} R(G,p);
\end{eqnarray*}
thus, $(\hat{R}(H,p))^{1/n_G}$ is an excellent approximation to $R(G,p)$. Thus, the approximability of $R(G,p)$ is either excellent, or very poor; determining the truth here is a very intriguing open question. 

\section{Acknowledgments}  We thank the SODA 2014 referees for their helpful comments, and Leslie Goldberg for a valuable discussion on $R(p)$. The discussion with Leslie Goldberg that is included in Section~\ref{sec:concl}, as well as part of Aravind Srinivasan's writing, transpired while Srinivasan attended the ICERM Workshop on \emph{Stochastic Graph Models} in March 2014. We thank the workshop organizers for their kind invitation. 

Thanks to Thomas Lange and to the anonymous journal reviewers for their many helpful comments and thorough reading of this paper.

\appendix

\section{Numerical analysis for Section~\ref{Asec:oddc}}
\label{oddc-appendix}

To complete the proof of Theorem~\ref{Amainoddthm}, we must verify certain concavity and monotonicity properties of the function $f$ defined by
$$
f(i,n,k) = \log \Bigl[\frac{(i/n)^{\frac{2}{c+1}-1}  (2 k + (c-1) n) - 2 k}{(c-1) i} \Bigr]
$$

\begin{proposition}
\label{Aprop1x5}
For $1 \leq i \leq n$, the function $f$ is well-defined.
\end{proposition}
\begin{proof}
We need to show that the argument of the logarithm is strictly greater than 0. By Proposition~\ref{Aprop1x1}, the argument is strictly increasing in $k$, so it suffices to show that it is positive at $k = 0$. When $k = 0$, it is equal to
$$
\frac{(i/n)^{\frac{2}{c+1}-1}  ((c-1) n)}{(c-1) i}
$$
which is clearly positive.
\end{proof}

\begin{proposition}
\label{Aprop1x1}
For $1 \leq i \leq n$ and $c \geq 1$ the function $f(i,n,k)$ is an increasing, concave-down function of $k$.
\end{proposition}
\begin{proof}
The function $f$ applies the logarithm function to an argument which is a linear function of $k$. The slope of that
linear function is 
$$
\frac{2 n ( (i/n)^{\frac{2}{c+1}}- (i/n)) }{(c-1)
   i^2} > 0
$$for $c > 1$ and $1 \leq i \leq n$.
\end{proof}

\begin{proposition}
\label{Aprop1x2}
Suppose $c \geq 3$ and $i,k,n \geq 0$. Then the expression
$$
c/m + f(i,n,ke^{-c/m})
$$
is a decreasing function of $m$.
\end{proposition}
\begin{proof}
Let $t = c/m$. We wish to show that $t + f(i,r,k e^{-t})$ is increasing in $t$.

Differentiating with respect to $t$, it suffices to show that
$$
\frac{2 k((i/n) - (i/n)^{\frac{2}{c+1}})}{(i/n)^{\frac{2}{c+1}} ((c-1) n e^t + 2k) - 2 (i/n) k} + 1 \geq 0 \
$$

By differentiating with respect to $i$ (viewing $i$ as a continuous parameter), we see that the left-hand side of the above expression is increasing in $i$. Hence it suffices to show that this holds when $i \rightarrow 0$. In this case, as $i \rightarrow 0$, the terms $(i/n)^{2/(c+1)}$ are negligible compared to the linear terms in $(i/n)$, so the above approaches a limit as $i \rightarrow 0$ of 
$$
\frac{(c-1) n e^t}{(c-1) n e^t + 2 k}
$$
which is obviously positive.
\end{proof}

\begin{proposition}
\label{A1x3}
For $c \geq 3, n \geq i+1$ and $k \leq n$ we have
$$
 \frac{c}{n c/2 + (n-k)/2} + f(i,n-1,k e^{- \frac{c}{n c/2 + (n-k)/2}}) \leq f(i,n,k)
$$
\end{proposition}
\begin{proof}
It suffices to show that
$$
\exp\Bigl(  \frac{c}{n c/2 + (n-k)/2} + f(i,n-1,ke^{ -\frac{c}{n c/2 + (n-k)/2}}) \Bigr) - \exp(f(i,n,k)) \leq 0
$$

This expression can be rewritten as
\begin{equation}
\label{Ae1}
\frac{ -2 k \Bigl( n^{\frac{c-1}{c+1}} - (n-1)^{\frac{c-1}{c+1}} \Bigr) - (c-1) \Bigl( n^{\frac{2 c}{c+1}}-e^{\frac{2 c}{c n+n}} (n-1)^{\frac{2 c}{c+1}} \Bigr) }
{(c-1)
   i^{\frac{2 c}{c+1}}} \leq 0
   \end{equation}

The denominator is positive, and the term $n^{1- \frac{2}{c+1}} - (n-1)^{1 - \frac{2}{c+1}}$ is positive. Finally, the coefficient $\Bigl( n^{\frac{2 c}{c+1}}-e^{\frac{2 c}{c n+n}} (n-1)^{\frac{2 c}{c+1}} \Bigr)$ can be bounded as:
\begin{align*}
n^{\frac{2 c}{c+1}}-e^{\frac{2 c}{c n+n}} (n-1)^{\frac{2 c}{c+1}} &= n^{\frac{2 c}{c+1}} \Bigl( 1 - e^{\frac{2 c}{n(c+1)}} \bigl( 1- 1/n)^{\frac{2 c}{c+1}} \Bigr) \geq n^{\frac{2 c}{c+1}} \Bigl( 1 - e^{\frac{2 c}{n(c+1)}} (e^{-1/n})^{\frac{2 c}{c+1}} \Bigr) =0
\end{align*}

Thus, we verify that (\ref{Ae1}) holds.
\end{proof}

\begin{proposition}
\label{A1x4}
For $c \geq 3, k \geq n \geq i+1$ we have
$$
\frac{2}{n} + f(i,n-1,ke^{-\frac{2}{n}}) \leq f(i,n,k)
$$
\end{proposition}
\begin{proof}
It suffices to show that
$$
\exp\Bigl( \frac{2}{n} + f(i,n-1,ke^{-\frac{2}{n}}) \Bigr) - \exp(f(i,n,k)) \leq 0
$$

This simplifies to showing that
\begin{equation}
\label{Ae11}
\frac{-2 k \Bigl( n^{\frac{c-1}{c+1}} - (n-1)^{\frac{c-1}{c+1}} \Bigr) - (c-1) \Bigl( n^{\frac{2 c}{c+1}} - e^{2/n} (n-1)^{\frac{2 c}{c+1}} \Bigr)}{(c-1)
   i^{ \frac{2 c}{c+1}}} \leq 0
\end{equation}

The denominator of (\ref{Ae11}) is positive, and the numerator is a decreasing function of $k$. Hence, as $k \geq n$, it suffices to show that this holds at $k = n$. There, the numerator of (\ref{Ae11}) is equal to:
\begin{equation}
\label{Ae12}
-(c+1) n^{\frac{2 c}{c+1}}+2 n (n-1)^{\frac{c-1}{c+1}}+(c-1) e^{2/n}
   (n-1)^{\frac{2 c}{c+1}}
   \end{equation}
   
So, we need to show (\ref{Ae12}) is negative. Here we have:
\begin{align*}
(\ref{Ae12}) &=  n^{\frac{2 c}{c+1}} \Bigl( -(c+1) + (1-1/n)^{\frac{c-1}{c+1}} (2 + (c-1) e^{2/n} (1 - 1/n)) \Big) \\
&\leq  n^{\frac{2 c}{c+1}} \Bigl( -(c+1) + (1-1/n)^{\frac{c-1}{c+1}} (2 + (c-1) (1+1/n) \Bigr) \qquad \text{as $e^{2 x} (1-x) \leq 1+x$ for $x \in [0,1]$}
\end{align*}

The derivative of the quantity $-(c+1) + (1-1/n)^{\frac{c-1}{c+1}} (2 + (c-1) (1+1/n))$ with respect to $n$ is $\frac{2 (c-1) c \left(\frac{n-1}{n}\right)^{-\frac{2}{c+1}}}{(c+1) n^3}$, which is clearly positive. Also clearly this quantity is negative at $n = 1$ (it approaches $-(c+1)$) and approaches 0 as $n \rightarrow \infty$. This implies that it is always negative, which in turn implies that (\ref{Ae12}) is negative as desired.

\end{proof}

\section{Numerical analysis for Section~\ref{Azsec}}
\label{mainrel-appendix}
\textbf{In this section, we will assume $U(p) < n^{-K}$ for some constant $K > 2$; thus the conclusions of Proposition~\ref{k-prop} hold.}

To complete the proof of Theorem~\ref{Amainrelthm}, we must verify certain concavity and monotonicity properties of the function $f$ defined by
$$
f(i,r,a,\gamma) = \begin{cases}
\gamma \Bigl( 1 - \frac{\log^2 (a/i)}{\log^2 (a/r)} \Bigr) & \text{if $a \in (0,1]$} \\
2 \log(r/i) & \text{if $a > 1$}
\end{cases}
$$

\begin{proposition}
\label{Aprop2x1}
Let $r, i$ be integers with $r \geq i \geq 100$ and let $\gamma, a$ be real numbers with $\gamma \geq 2 \log r, a \in (0,1)$. Let $A$ be any random variable with domain $(0,\infty)$ satisfying $\bE[A] \leq a$. Then
$$
\bE[ f(i,r,A,\gamma) ] \leq f(i,r,a,\gamma)
$$
\end{proposition}
\begin{proof}

Let us fix $i,r, \gamma$ and define the function $g:[0,1] \rightarrow \mathbf R$ by
$$
g(x) = f(i,r,x, \gamma)
$$

We claim first that $g'(x) \geq 0$ and $g''(x) \leq 0$; so $g$ is increasing and concave-down. For, we have that:
$$
g'(x) = \frac{-2 \gamma \log(i/r) \log(i/x)}{x \log^3 (r/x)} 
$$
As $r \geq i \geq 1$ and $x \in [0,1]$, the term $\log(i/r)$ is negative, the terms $\log(i/x), \log(r/x)$ are positive, and $\gamma$ is positive.
So overall $g'(x) \geq 0$.

The second derivative is given by
$$
g''(x) = \frac{-2 \gamma  \log (r/i) \bigl( \log (r/x) + (\log(r/x) - 3) \log(i/x) \bigr)}{x^2 \log^4 (r/x)}
$$
For $r \geq 21$, we have $\log(r/x) - 3 \geq 0$. The other terms here $\log(r/i), \log(r/x), \log(i/x), \gamma$ are all clearly positive. So overall $g''(x) \leq 0$.

We next claim that $g(1) \geq 2 \log (r/i)$, that is,
$$
\gamma (1 - \frac{\log^2 i}{\log^2 r}) \geq 2 \log (r/i)
$$

As $\gamma \geq 2 \log r$, it suffices to show that
$$
2 \log r (1 - \frac{\log^2 i}{\log^2 r}) - 2 \log (r/i) \geq 0
$$
The simplifies to showing $2 \times \log i \times (1 - \frac{\log i}{\log r}) \geq 0$, which is clear.

Thus, as $g(x)$ is increasing on $x \in (0,1]$ and $g(1) \geq 2 \log(r/i)$, it follows that $g(x) \leq g(1)$ for all $x \in (0,\infty)$.

Having shown these concavity properties of the function $g(x)$, we continue to prove the claimed result.  Consider the random variable $A' = \min(A,1)$. Observe that $g(A) \leq g(A')$. Also, clearly $\bE[A'] \leq \bE[A] \leq a$.  By Jensen's inequality $\bE[g(A')] \leq g(\bE[A'])$. As $g(x) $ is increasing on $x \in (0,1]$, we have that $g(\bE[A']) \leq g(a)$. This shows that $\bE[g(A)] \leq g(a)$ as desired.
\end{proof}

\begin{proposition}
\label{Aprop2x2}
For $a \leq 1$ and $r \geq i+1$, the expression
$$
c/m + f(i,r-1,a,\gamma - c/m)
$$
is non-increasing in $m$.
\end{proposition}
\begin{proof}
For $a \leq 1$, the derivative of this quantity with respect to $m$ is given by
$$
\frac{- c \log ^2\left(\frac{a}{i}\right)}{m^2 \log
   ^2\left(\frac{a}{r - 1}\right)}
$$
which is clearly $\leq 0$.
\end{proof}

\begin{proposition}
\label{Aprop2x3}
For $a \in [0,1], \gamma \geq 0$, and $r \geq 100$, we have
$$
\frac{-2 \gamma}{r \log(a/r)} + f(i,r-1,a,\gamma + \frac{2 \gamma}{r \log(a/r)} ) \leq f(i,r,a,\gamma)
$$
\end{proposition}
\begin{proof}
We compute the RHS minus the LHS. After collecting terms this is given by:

$$
\frac{\gamma  \log ^2\left(\frac{a}{i}\right) \Bigl (  r (2 \log a) (1 + r \log(1-1/r) ) - 2 r \log r + r^2 ( \log^2(r) - \log^2(r-1) ) \Bigr)}
{r^2  \log^2 ( \frac{a}{r-1} ) \log ^2 (\frac{a}{r} )}
$$
and we want to show this is positive.

Removing terms with obvious signs, it suffices to show that
\begin{equation}
\label{Ae2}
r (2 \log a) (1 + r \log(1-1/r) ) - 2 r \log r + r^2 ( \log^2(r) - \log^2(r-1) ) \geq 0
\end{equation}

Simple calculus shows that $1 + x \log(1-1/x) \leq 0$ for all $x \geq 0$; as $a \leq 0$ this implies that $r (2 \log a) (1 + r \log(1-1/r) ) \geq 0$. Thus, in order to show (\ref{Ae2}) it suffices to show that

\begin{equation}
\label{Ae3}
r^2 ( \log^2 r - \log^2(r-1) ) - 2 r \log r \geq 0
\end{equation}

To show this, we take a second-order Taylor expansion of $\log x$ around $x = r$; this gives us that
$$
\log (r-1) \leq \log r - \frac{1}{r} - \frac{1}{2 r^2}
$$
Hence
\begin{align*}
r^2 ( \log^2 r - \log^2(r-1) ) - 2 r \log r &\geq r^2 (\log^2 r - (\log r - 1/r - 1/(2 r^2))^2) - 2 r \log r = \log r - (1 + \frac{1}{2 r})^2.
\end{align*}
Simple calculus shows that this is positive for $r \geq 4$.
\end{proof}

\section{Selecting $y \in [0,1]$ for Lemma~\ref{Arcalemma3}}
\label{bound-appendix}
\textbf{In this section, we will assume $U(p) < n^{-K}$ for some constant $K > 2$; thus the conclusions of Proposition~\ref{k-prop} hold. We also suppose throughout that \boldmath$3/2 \leq \alpha \leq \alpha^* \leq O(1), \rho \leq 2.1, \beta \leq 1, \delta \geq \beta$\unboldmath.} (Other cases are covered already in Proposition~\ref{Arcathm2}.)

For any $y \in [0,1]$, let us define
$$
T(y,\alpha) = \alpha \max(h(y), \bar h(y/\alpha))
$$

For any $y \in [0,1]$, Lemma~\ref{Arcalemma3} shows that the success probability of the RCA will be at least $n^{2 - 2 y - T(y,\alpha) - o(1)}$. We claim that it is always possible to select $y$ so that this quantity is at least $n^{-1 - o(1)} \epsilon^{-1.5-o(1)}$.  Thus, it suffices to show that
\begin{equation}
\label{dbnd}
2 y + T(y,\alpha) \leq 3 + 1.5 \rho
\end{equation}

Let us define
$$
a = \frac{(3 - \beta + \delta)^2(2 - \beta + \delta + \rho)}{(2 - \beta + \delta)^2(2 + \delta)}
$$

We prove (\ref{dbnd}) in two parts. First, we show that it suffices to bound $2 y + T(y,a)$ --- that is, when we replace  $\alpha$ by a quantity which is (approximately) an appropriate upper bound on $\alpha^*$.  Next, we show that $2 y + T(y, a)$ can be regarded as a scalar function depending on only a few real parameters, which can be tested exhaustively.

\begin{proposition}
\label{dprop3}
For any $y \in [0,1]$ we have
$$
T(y, \alpha) \leq T(y, a) + o(1)
$$
\end{proposition}
\begin{proof}
By Lemma~\ref{Aalphastarlemma1}, we have $\alpha \leq \alpha^* \leq a + o(1 + \rho)$. As $\rho = O(1)$, this implies that $\alpha \leq a + o(1)$.

First, suppose that $\alpha = a + t$, where $0 \leq t \leq o(1)$. By Proposition~\ref{h-deriv-prop}, the derivatives of $h, \bar h$ are bounded by constants, and so $\alpha h(y) \leq a h(y) + o(1)$ and similarly $\alpha \bar h(y/\alpha) \leq a \bar h(y/a) + o(1)$. So $T(y,\alpha) \leq T(y,a) + o(1)$ as desired.

Next, suppose that $\alpha \leq a$. Then $\alpha h(y) \leq a h(y)$. Also, it can be mechanically verified that $x \bar h(y/x)$ is an increasing function of $x$; thus $\alpha \bar h(y/\alpha) \leq a \bar h(y/a)$. This implies that $T(y,\alpha) \leq T(y,a)$
\end{proof}

Next, we bound $2 y + T(y,a)$ by using the crude but effective method of using a $\epsilon$-net and essentially trying every value of $\beta, \delta$ exhaustively. 
\begin{proposition}
\label{dprop7}
There is $y \in [0,1]$ with $2 y + T (y,a) \leq 1 + 1.5 \rho + o(1)$.
\end{proposition}
\begin{proof}
First, suppose that $\delta \geq 3/2$. Then we set $y = 1$ and obtain that
$$
2 y + T(y,a) \leq 2 + a \max(h(1), \bar h(1/a))
$$

As $h(1) = 0$, this implies that $2 y + T (y,a) \leq 2 + a \bar h(1/a)$. This can be written as a rational function of $\beta, \delta, \rho$, and  it can be verified mechanically that $2 + a \bar h(1/a) \leq 3 + 1.5 \rho$.

We next assume that $\beta \in [0,1], \delta \in [\beta, 3/2]$. We will show that in this range we can choose $y$ such that
$$
2 y + T(y,a) \leq 2.95 + 1.2 \rho
$$

We loop over choices of $b = 0, s, 2s ,\dots, 1 - s$ and $d = b-s, b, b+s, b+2d, \dots, 3/2 - s$ and $r = 0, s, 2 s, \dots, 2.1$ where $s = 0.001$. For each such choice $b,d,r$ we choose some $y \in [0,1]$ and compute an upper bound on $2 y + T(y,a)$ for $\beta, \delta, \rho$ in the range $\beta \in [b, b+s], \delta \in [d,d+s], \rho \in [r, r+s]$. We denote this interval by $I_{b,d,r}$. (Note that the inequality $d \geq b - s$ comes from the fact that we must have $\delta \geq \beta$.)

Now, suppose we have fixed some choice of $b,d,y$ and we want to show that in this region 
\begin{equation}
\label{ye1}
2 y + a \max(h(y), \bar h(y/a) ) \leq 2.95 + 1.5 \rho
\end{equation}

It can be mechanically verified that in the interval $I_{b,d,r}$ we have 
$$
a \leq a_{\text{max}} = \frac{ (3 -b + d)^2 }{(2+d)(2 - b + d)} + \frac{ (r+s) (3 - b + d - s)^2}{(2+d)(2 - b + d - s)^2}
$$

As $\bar h$ is an increasing function, a sufficient condition to show (\ref{ye1}) is that
\begin{equation}
\label{ye2}
2 y + a_{\text{max}} \max(h(y), \bar h(y/a_{\text{max}})) \leq 2.95 + 1.5 r
\end{equation}

It is easy to compute using standard numerical techniques an upper bound on the quantity $\max( h(y), \bar h(y/a_{\text{max}}))$ over the interval $I_{b,d,r}$. (This upper bound is not affected by the value of $\rho$) For example, $\bar h(z)$ is an increasing function of $\beta, \delta$. 

Thus, we can compute an upper bound on $2 y + a_{\text{max}} \max(h(y), \bar h(y/a_{\text{max}}))$ in the interval $I_{b,d,r}$; if this upper bound is at most $2.95 + 1.5 r$, then this shows that (\ref{ye2}) and hence (\ref{ye1}) holds.

We wrote C code (omitted here) to loop over all $b, d, r$ and perform the above calculation. This procedure works for any choice of $y$. The method we used was to approximately satisfy the constraint $h(y) = \bar h(y/a_{\text{max}})$ at $\beta = b, \delta = d$. With this choice of $y$, we satisfy (\ref{ye2}) for all intervals.\footnote{Note that the relevant calculations were performed using double-precision floating point (with the default rounding mode), instead of rigorous interval arithmetic. This was done for the sake of easy implementation using standard numerical libraries. However, in fact, we observed that these constraints were both satisfied with significant slack --- we saw that $2 y + a_{\text{max}} \max(h(y), \bar h(y/a_{\text{max}})) - 2.95 - 1.5 r \leq -0.01$. Thus we do not expect that this slight numerical error will affect the validity of the results.}
\end{proof}

Putting these cases together, we  have our desired result:
\begin{proposition}
\label{dprop99}
Suppose $3/2 \leq \alpha \leq \alpha^*$ and $\rho \leq 2.1$ and $\beta \in [0,1]$. Then there is $y \in [0,1]$ such that
$$
2 y + T(y,\alpha) \leq 2.95 + 1.5 \rho + o(1)
$$
\end{proposition}

\section{Properties of the functions $h, \bar h$}
We discuss here some properties of the functions $h(x), \bar h(x)$. We repeat the functions $h, \bar h$ for convenience. \textbf{In this section, we will assume $U(p) < n^{-K}$ for some constant $K > 2$; thus the conclusions of Proposition~\ref{k-prop} hold.}

\begin{align*}
h(x) &= \begin{cases}
2(2+\delta)( \log(3 - \beta + \delta) - \log(2 - \beta + \delta + x) ) & \text{if $x \geq \beta$} \\
2(2+\delta)( \log(3 - \beta + \delta) - \log(2 + \delta) ) + 2 (\beta - x) & \text{if $x \leq \beta \leq 1$} \\
2 (1 - x) & \text{if $\beta \geq 1$}
\end{cases} \\
\bar h(x) &= \frac{(\delta+2)(1 - x)(5 - 2 \beta + 2 \delta + x)}{(3 - \beta + \delta)^2} \\
\end{align*}

\begin{proposition}
\label{hbar-min}
We have $10/9 \leq \bar h(0) \leq 20/9$. Furthermore, if $\beta \leq 3/2$ then $\bar h(0) \leq 2$.
\end{proposition}
\begin{proof}
$\bar h(0)$ is a rational function of $\beta, \delta$ so this fact can be verified mechanically for $\delta \geq \beta, \beta \in [0,2]$.
\end{proof}

\begin{proposition}
\label{h-deriv-prop}
We have $-2  \leq h'(x) \leq -4/3 - \Omega(1)$
\end{proposition}
\begin{proof}
When $\beta \geq 1$ or $x \leq \beta$ then $h'(x) = -2$ and so this is clear. When
$\beta \leq x \leq 1$ we have
$$
h'(x) = \frac{-2(2+\delta)}{2 - \beta + \delta + x}
$$
As $\beta \leq 2 + o(1)$ and $\delta \geq \delta_0$, for $n$ sufficiently large this is negative. Hence the maximum value occurs at $\beta = 0, x = 1$ yielding
$$
h'(x) \leq -2 + \frac{2}{3 + \delta} \leq -2 + \frac{2}{3 + \delta_0} \leq -4/3 - \Omega(1)
$$

The minimum value of $h'(x)$ occurs at $\beta = x$ yielding $h'(x) \geq 2$.
\end{proof}

\begin{proposition}
\label{hbar-deriv-prop}
We have $-4 + \Omega(1)  \leq \bar h '(x) \leq - \Omega(1)$.
\end{proposition}
\begin{proof}
We compute the derivative of $\bar h$ as:
$$
\bar h'(x) = -\frac{2 (\delta +2) (-\beta +\delta +x+2)}{(-\beta +\delta +3)^2}
$$

This is clearly a decreasing function of $x$, so it achieves its minimum value at $x=1$.  There it is equal to
\begin{align*}
\bar h'(1) &= \frac{-2 (2+\delta)}{3 - \beta + \delta} \geq \frac{-2 (2+\delta)}{3 - 2 + \delta} \geq \frac{-2(2+\delta_0)}{1 + \delta_0} \geq -4 + \Omega(1)
\end{align*}

Similarly, $\bar h'(x)$ achieves its maximum value at $x = 0$ where we have
$$
\bar h'(0) = \frac{2 (\delta +2) (\beta -\delta -2)}{(-\beta +\delta +3)^2}
$$

One can see that this is U-shaped as a function of $\beta$. So its maximum value occurs at either $\beta = 0$ or $\beta = 2$. At those points respectively we have
$$
\bar h'(0) = \frac{-2(2+\delta)^2}{(3+\delta)^2} \leq -8/9 
$$
and 
\begin{align*}
\bar h'(0) &= -2 + \frac{2}{(1 + \delta)^2} \leq -2 + \frac{2}{(1 + \delta_0)^2} \leq -\Omega(1)
\end{align*}
\end{proof}

\begin{proposition}
\label{h2-prop}
For $\beta \in [0,1]$ and $x \in [0,1]$ we have $\bar h(x) \leq h(x)$.
\end{proposition}
\begin{proof}
Observe that $h(1) = \bar h(1) = 0$. So it suffices to show that $\bar h'(x)  \geq h'(x)$ for $x \in [0,1]$. 
We have that $\bar h'(x) = \frac{-2 (\delta +2) (-\beta +\delta +x+2)}{(-\beta +\delta +3)^2}$. This is a rational function, and so it can be mechanically verified that $\bar h'(x) \geq -2$ for $x \in [0,1]$ and $\beta, \delta \in [0,1]$. As $h'(x) = -2$ when $x \leq \beta$, this automatically implies that $\bar h'(x) \geq h'(x)$ for $x \leq \beta$. 

When $x \geq \beta$, we have that $h'(x) = \frac{-2(2+\delta)}{2 - \beta + \delta + x}$. So to show that $\bar h'(x) \geq h'(x)$, we must show that
$$
\frac{-2 (\delta +2) (-\beta +\delta +x+2)}{(-\beta +\delta +3)^2} \geq \frac{-2(2+\delta)}{2 - \beta + \delta + x} \geq 0
$$

Again, this is a rational function, so it can be mechanically verified that this holds in the indicated range.
\end{proof}

 \end{document}